\documentclass[12pt,a4paper,oneside]{article}

\usepackage[lmargin=18mm,rmargin=18mm,top=3cm]{geometry}
\usepackage{lineno}

\usepackage{natbib}
\usepackage{amssymb,xcolor,graphicx,amsthm}
\usepackage{amsmath}
\usepackage{bbm}
\usepackage{url}
\usepackage{placeins}
\usepackage{lipsum}

\newtheorem{theorem}{Theorem}[section]
\newtheorem{lemma}[theorem]{Lemma}
\newtheorem{proposition}[theorem]{Proposition}
\newtheorem{corollary}[theorem]{Corollary}

\newtheorem{definition}[theorem]{Definition}
 
\newtheorem{example}[theorem]{Example} 

\newcommand{\ben}{\begin{enumerate}}
	\newcommand{\een}{\end{enumerate}}

\newcommand{\bt}{\begin{theorem}}
	\newcommand{\et}{\end{theorem}}
\newcommand{\bl}{\begin{lemma}}
	\newcommand{\el}{\end{lemma}}
\newcommand{\bc}{\begin{corollary}}
	\newcommand{\ec}{\end{corollary}}
\newcommand{\bp}{\begin{proposition}}
	\newcommand{\ep}{\end{proposition}}
\newcommand{\br}{\begin{remark}}
	\newcommand{\er}{\end{remark}}
	\newcommand{\be}{\begin{example}}
	\newcommand{\ee}{\end{example}}
\newcommand{\iid}{\stackrel{\mbox{\scriptsize iid}}{\sim}}
\renewcommand{\mid}{\,|\,}

\newcommand{\E}{{\mathrm E}}

\usepackage[normalem]{ulem}
 
\newcommand{\MBtext}[1]{{{#1}}}

\newtheorem{remark}{Remark}

\providecommand{\keywords}[1]
{
	\small	
	{\noindent\textit{Keywords: }} #1
}

\usepackage{amsmath}
\usepackage{tabularx}
\usepackage{multirow}
\usepackage{xcolor}
\usepackage{booktabs}
\usepackage{lineno}
\usepackage{setspace}
\usepackage{hyperref}
\usepackage{url}

\def\spacingset#1{\renewcommand{\baselinestretch}%
{#1}\small\normalsize} \spacingset{1}

\begin{document}

\title{Sufficient digits and density estimation: A Bayesian nonparametric approach using generalized finite P{\'o}lya trees}%

\author{M. Beraha\footnote{Department of Economics, Management and Statistics, University of Milano-Bicocca, Milan, Italy. \\E-mail: \url{mario.beraha@unimib.it}} \ and J. M\o ller\footnote{Department of Mathematical Sciences, Aalborg University, Aalborg Ø, 9220, Denmark. \\E-mail: \url{jm@math.aau.dk}}}
\date{}

\maketitle

\begin{abstract}
This paper proposes a novel approach for statistical modelling of a continuous random variable $X$ on $[0, 1)$, based on its digit representation $X=.X_1X_2\ldots$.
In general, $X$ can be coupled with a latent random variable $N$ so that $(X_1,\ldots,X_N)$ becomes a sufficient statistics and $.X_{N+1}X_{N+2}\ldots$ is uniformly distributed. 
In line with this fact, and focusing on binary digits for simplicity,
we propose a family of generalized finite P{\'o}lya trees that induces a random density for a sample, which becomes a flexible tool for density estimation. 
Here, the
digit system may be random and learned from the data. We provide a detailed Bayesian analysis, including closed form expression for the posterior distribution.
We analyse the frequentist properties as the sample size increases, and provide sufficient conditions for consistency of the posterior distributions of the random density and $N$. We consider an extension to data spanning multiple orders of magnitude, and propose a prior distribution that encodes the so-called extended Newcomb-Benford law. Such a model shows promising results for density estimation of human-activity data. Our methodology is illustrated on several synthetic and real datasets.
\end{abstract}


\spacingset{1.4}

\section{Introduction}\label{s:intro}

Consider a continuous random variable $X$ defined on the unit interval and let $X_1,X_2,\ldots$ be its binary digits such that $X = \sum_{n=1}^\infty X_n 2^{-n}$. 
In practice, due to a round-off error or the finite arithmetic precision of computers, the data consist only of a finite number of digits. 
In fact, under very weak conditions, there is a coupling of $X$ with a non-negative integer-valued latent random variable $N$ such that $X_1,\ldots,X_N$ become the sufficient digits, i.e., 
they are independent of the remaining digits $X_{N+1},X_{N+2},\ldots$ which do not depend on the model for $X$, and thus contain no information and are not relevant for statistical inference, cf.\ \cite{moeller}.

\subsection{Overview on models and results}\label{s:overview}

In this paper, we build on this coupling result from \cite{moeller} and propose a novel approach to statistical modelling for a random variable based on its digit representation.
As detailed in Theorem~\ref{t:1} and Remark~\ref{r:suff}, the coupling result can be presented in a very general setting with random digits defined in terms of nested partitions of the unit interval, where the intervals of each partition are not necessarily of equal length as explained in Definition~\ref {def:1} and Remark \ref{r:1} below.
Thinking of $N$ as a latent variable which determines the needed number of digits, and working with binary digits for simplicity, conditioned on $N=n$, we propose models where for $j=1,\ldots,n$, we let $X_{j} \mid X_1,\ldots,X_{j-1}$ be Bernoulli-distributed with a beta-distributed probability parameter depending on $X_1,\ldots,X_{j-1}$, and where $X_{n+1}, X_{n+2}, \ldots$ are independent of the first $n$ digits and contain no information. 
Such a construction, which arises naturally from our random-digit modelling approach, can be seen as a
generalizations of P{\'o}lya trees (PT) \citep{lavine,lavine2,MauldinEtAl} with a finite number of partition levels, and are therefore termed \emph{generalized finite P{\'o}lya trees of type 1} (for short GFPT1; we will introduce the type 2 models below).
PT priors play an important role in non-parametric Bayesian inference \citep[][]{AwayaMa, BergerGuglielmi, Castillo17,giordano, Hanson, Holmes2015, BNPDA, Paddock, WongMa}.
A GFPT1 prior specifies the distribution of a random probability density function (PDF) used for the data model of
{exchangeable} observations. This random PDF is a mixture over $N$ of functions on the intervals of the binary partitions. In particular, it is easy to see that GFPT1 priors correspond to the finite PT priors of \cite{Hanson} with an additional prior for the partition depth.

{\color{black}The novelty of our contribution does not lie in a generalization of finite P\'olya trees, but in the \emph{random-digit framework} that considers digits as the primary modelling object. 
By placing a prior on the number of sufficient digits $N$ we turn it from a fixed tuning parameter (as in standard finite PTs) into an identifiable latent quantity amenable to inference and uncertainty quantification. 
By Theorem~\ref{t:1}, the trailing digits beyond $N$ are ancillary and, therefore, genuinely non-informative for inference; modelling only the informative prefix yields likelihoods and priors that correctly account for round-off and finite precision. 
Most importantly, working at the digit level exposes a direct link to significant-digit phenomena such as the Newcomb–Benford law \citep{Berger2,Clauset,Hill}, which routinely arises in human data. 
In brief, the Newcomb-Benford law asserts that small digits are more likely to occur as leading digit in human data, and we show in Section~\ref{s:elicitation} how to incorporate this as prior information by extending certain GFPT1s to what we call \emph{multiscale Benford P{\'o}lya trees} (MBPT).
We find that the MBPT priors provide superior performance for density estimation in problems where data span multiple orders of magnitude.
Such a setting is notoriously difficult in Bayesian density estimation due to heavy tails and different scales in the densities \citep{Tokdar}.}

{\color{black}Focusing on GFPT1 priors,} we provide a closed-form expression for the posterior distribution which allows for straightforward inference without resorting to Markov chain Monte Carlo algorithms.
In fact, our models can be fitted to tens of millions of datapoints in a matter of a fraction of a second on a standard laptop.
We investigate the frequentist properties of the Bayesian estimators obtained from the posterior distribution when assuming each data point
was generated from a true PDF $f^*$:
First, we consider the posterior distribution of $N$. Assuming that $n^*$ digits are sufficient for $f^*$ (in the sense of Theorem~\ref{t:1}), we establish conditions under which the posterior of $N$ concentrates on $n^*$, assuming that we have observed either the data perfectly (i.e., without round-off errors), or that we have access only to the first $\bar n$ digits, which is arguably the most common situation in applied settings. 
In the latter case, we identify a curious situation where ignoring the round-off errors leads to $N$ diverging to infinity a posteriori. Instead, assuming that the round-off error is uniform leads to the correct asymptotic behaviour.
Second, we consider the convergence of the {\color{black}the posterior distribution of the} random PDF, establishing consistency with respect to both the Kolmogorov-Smirnov and the Hellinger distance.

{\color{black}Beyond the MBPT prior discussed above, we propose a furhter generalization of GFPT1 priors}, termed \emph{generalized finite P{\'o}lya trees of type 2} (GFPT2), where we let the nested binary partitions be random such that the %step sizes and the 
lengths of the intervals are obtained from sequences of beta-distributed random variables as detailed later in Definition~\ref{def:6}.
This additional stochastic layer is essential:  Under a fixed deterministic partition one may require an unbounded number of digits ($n^{\ast}=\infty$) even when the true PDF $f^{\ast}$ is merely piecewise constant, simply because its discontinuities do not coincide with the pre-specified nested partitions. %grid. 
By randomising the partitions, a GFPT2 prior circumvents this limitation and restores the finite digit sufficiency property for a broad class of densities.

% The second generalization, discussed in Section~\ref{sec:mbpt}, proposes how to model a sample consisting of positive numbers by combining a GFPT1 prior with a prior distribution for the order of magnitude of each data point. %a positive continuous random variable.
% We find such a model to be particularly useful for density estimation in problems where data span multiple orders of magnitude and obey the so-called extended Newcomb-Benford law which routinely happens in human data \cite{Berger2,Clauset,Hill}. In brief, the Newcomb-Benford law asserts that small digits are more likely to occur as leading digit in human data, and we show in Section~\ref{s:elicitation} how to incorporate this as prior information by extending certain GFPT1s to what we call \emph{multiscale Benford P{\'o}lya trees} (MBPT).

%\subsection{Data consisting of real numbers}

We mainly restrict attention to the case $0<X<1$, since any real number can be transformed to a number between 0 and 1 by a bijective mapping (e.g.\ $x\mapsto\exp(x)/(1+\exp(x))$) 
and since data consisting of positive numbers can be modelled by using MBPT priors.
%combining certain GFPT1 models with a model for the order of magnitude of a positive continuous random variable as discussed above and further in Section~\ref{sec:mbpt}. 
%Specifically, we consider an extension {to data spanning multiple orders of magnitude \cite{Clauset}, and propose a prior distribution that encodes the so-called Newcomb-Benford law \cite{Berger2,Hill}, and so we call it a \emph{multiscale Benford P{\'o}lya tree}. As we shall see, such a model shows promising results for density estimation of human-activity data. 
To cover data consisting of real numbers, MBPT priors may of course be extended by including a prior for the sign.

\subsection{General numeral systems}\label{s:gen-num}

It is only for specificity and simplicity that we consider a binary numeral system. Indeed, everything in this paper easily extends to the case of base-$q$ representations of numbers when $q\ge 2$ is an integer  (further details are given in
 Section~\ref{s:elicitation}). This includes the common cases of a binary ($q=2$) and a decimal numeral system ($q=10$) with equal interval lengths in each partition, and
for applications it seems not important how $q$ is specified: Suppose $N_{q}$ is the discrete random variable corresponding to $N_2=N$ above but obtained by using base-$q$ digits instead, whereby $q^{N_q}$ becomes the number of possible values of the sufficient digits $X_1,\ldots,X_{N_q}$ when the value of $N_q$ is fixed.  Considering simulation studies and statistical analysis of real data examples, applying our prior models such that $2^{N_2}$ and $10^{N_{10}}$ are closely distributed, we observe  in Section~\ref{sec:comparison_real_data} that rather similar conclusions are obtained  no matter if $q=2$ or $q=10$.

\subsection{Outline}

The paper is organized as follows. Section~\ref{s:notation} specifies the notation used throughout this text. 
Section~\ref{s:new} explains the meaning of sufficient digits (as briefly discussed in Section~\ref{s:overview}). 
Models for data on the unit interval are studied in Section~\ref{s:models}:
Section~\ref{s:finite_pt} recalls the definition of P{\'o}lya tree priors and their finite versions, 
Section~\ref{s:generelized_fpt} introduces our generalized finite P{\'o}lya tree priors GFPT1 and GFPT2 (cf.\ Section~\ref{s:overview}),  Section~\ref{s:elicitation} discusses the choice of hyperparameters for such priors and the common P{\'o}lya tree priors, {\color{black}and Section~\ref{sec:mbpt} studies our generalization of GFPT1 priors via the extended Newcomb-Benford law to MBPT priors, cf.\ Section~\ref{s:overview}.} Various theoretical results for Bayesian analysis based on GFPT1 and GFPT2 priors are established in Section~\ref{s:Bayes}: Section~\ref{s:posterior} shows how to easily simulate from the posterior distribution. Mainly restricting attention to GFPT1 priors and shortly commenting on GFPT2 {\color{black}and MBPT priors,} Section~\ref{sec:cons_n} deals with consistency results for the posterior distribution of $N$ and Section~\ref{s:consistency_f} with consistency results for the posterior {\color{black}distribution of the random PDF} (as briefly discussed in Section~\ref{s:overview}).   
Section~\ref{s:num} considers numerical simulation results:
Section~\ref{sec:s1} shows how the choice of hyperparameters affects posterior inference when using GFPT1 priors and suggests reasonable default values.
Section~\ref{sec:simu2} compares the {\color{black}density estimate given by the (estimated) posterior expected random PDF} when using PT, GFPT1, and GFPT2 priors as well as the Optional P\'olya tree \citep{WongMa} and the APT model of \cite{Ma17}, assessing whether the increased flexibility of GFPT2 priors yields practical gains over competitors. Section~\ref{sec:comparison_real_data} compares results for Bayesian density estimation when using GFPT1 and GFPT2 priors, base-2 and base-10 MBPT priors,
and the standard prior given by a Dirichlet process mixture of Gaussian densities \citep[see e.g.][]{gvdv}. Finally, Section~\ref{s:conclusion} summarizes our findings and points to future research.

% All proofs of our theorems in Sections~\ref{s:models} and \ref{s:Bayes} are deferred to  Supplemental material {\color{black} (Text to be finished later on ...)}

The Supplemental material collects the missing proofs in the present paper {\color{black} and, in continuation of our remarks in the last paragraph of Section~\ref{s:overview}, discusses the case of data that lies outside the unit interval, see also the summary at the beginning of Section~\ref{s:num}.}

{Supporting \texttt{Julia} code implementing the proposed methodology and the numerical simulations is available at \url{https://github.com/mberaha/FinitePolyaTrees.git}}

\section{Notation}\label{s:notation}

%The following notation is used throughout this text.

Let $\mathbb N=\{1,2,\ldots\}$, $\mathbb N_0=\{0,1,2,\ldots\}$, $\mathbb R=(-\infty,\infty)$, and $\mathbb R_+ = (0, +\infty)$ be the sets of positive integers, non-negative
integers, real numbers, and positive real numbers, respectively. 

For $n\in\mathbb N$ and $(x_1,\ldots,x_n)\in\mathbb R^n$, let $x_{1:n}=(x_1,\ldots,x_n)$, and for $k\in\{0,1\}$, identify $x_{1:n-1},k$ by $(x_1,\ldots,x_{n-1},k)$ if $n>1$ and by $k$ if $n=1$. Set $x_{1:0}=\emptyset$, $\{0,1\}^0=\{\emptyset\}$, and $x_{\emptyset,k}=x_k$ for $k=0,1$. We may interpret $x_{1:n}$ as an ordered configuration of $n$ points if $n>0$ and
 as the empty point configuration if $n=0$ \citep[this is similar to the terminology used in point process theory, see][]{JM}.  

The lexicographic order for finite binary sequences of possibly different lengths, denoted by $\prec$, is defined as follows. For $m,n\in\mathbb N$, $x_{1:n} \in \{0, 1\}^n$, and $z_{1:m} \in \{0, 1\}^m$, write $x_{1:n} \prec z_{1:m}$ if either  $m \ge n$ and $z_{1:n} = x_{1:n}$, or there exists a non-negative integer $k < \min\{m, n\}$  such that $x_{1:k} = z_{1:k}$, $x_{k+1} = 0$, and $z_{k+1} = 1$. 

Denote $\{0,1\}^\infty$ the set of all infinite sequences $(x_1,x_2,\ldots)$ with $x_n\in\{0,1\}, n=1,2,\ldots$. 
Let $\mathcal I$ be the collection of all half-open intervals $[a,b)$ with $0\le a<b<1$ (where $a$ is included and $b$ is excluded). For a set $A$ with either $A=\mathcal I$ or $\emptyset\not=A\subseteq\mathbb R$, for $(x_1,x_2,\ldots)\in\{0,1\}^\infty$, and for intervals/numbers 
$z_{x_{1:j}}\in A$ with $j\in\mathbb N$, consider 
 $z=(z_{x_{1:j}}\mid x_{1:j}\in\{0,1\}^j,\,j\in\mathbb N)$ as an infinite sequence with entries 
$z_{x_{1:j}}$ which appear in accordance to the ordering of the $x_{1:j}$ with respect to $\prec$. 
 For $n\in\mathbb N$, restricting $z$ to its first $2^n$ entries, we obtain the ordered finite sequence $z^{(n)}=(z_{x_{1:j}}\mid x_{1:j}\in\{0,1\}^j,\,j=1,\ldots,n)$. For notional convenience, let $z^{(n)}=\emptyset$ denote the empty set (sequence).
 For $n\in\mathbb N_0$, let 
$z^{(>n)}=(z_{x_{1:j}}\mid x_{1:j}\in\{0,1\}^j,j=n+1,n+2,\ldots)$.
Thus, $z^{(>0)}=z$.
  
Denote the beta-distribution on $[0,1)$ with shape parameters $\alpha_0>0$ and $\alpha_1>0$ by {\color{black}$\mathrm{Beta}(\alpha_0,\alpha_1)$, and its PDF by $\mathrm{Beta}(\cdot\mid\alpha_0,\alpha_1)$.} For later use, extend this PDF to $[0,1]$ by setting $\mathrm{Beta}(1\mid\alpha_0,\alpha_1)=0$ (or any other arbitrary value). For $k\in\{0,1\}$, $\alpha_k=0$, and $\alpha_{1-k}>0$, let $\mathrm{Beta}(\alpha_0,\alpha_1)$ be the degenerated probability distribution on $[0,1]$ which is concentrated at the point $k$ (so the value of $\alpha_{1-k}$ plays no role), and for $x\in[0,1]\setminus\{k\}$, let $\mathrm{Beta}(k\mid\alpha_0,\alpha_1)=1$ and $\mathrm{Beta}(x\mid\alpha_0,\alpha_1)=0$.

{\color{black}Denote the beta function by $\mathrm{beta}(x,y) = \int_0^1 t^{x - 1} (1 - t)^{y-1}\, \mathrm d t$ for $x>0$ and $y>0$. Let $\mathrm{beta}(0,x)=\mathrm{beta}(x,0)=1$ for $x>0$.} 

{\color{black} Denote $\mathrm{Unif}[a,b)$ the uniform distribution on $[a,b)$ when $a<b$.}

Finally, set $0^0=1$.

\section{The sufficient digits}\label{s:new}

The material in this section is used to motivate our new models introduced in Section~\ref{s:models}. For the following definition, let %$I_\emptyset=[0,1)$ be the half-open unit interval with left endpoint  
$a_{\emptyset}=0$ and length $\ell_\emptyset=1$.  

\begin{definition}\label{def:1} Suppose $I=(I_{x_{1:n}}\mid x_{1:n}\in\{0,1\}^n,\,n\in\mathbb N)$ where each 
 $I_{x_{1:n}}=[a_{x_{1:n}},a_{x_{1:n}}+\ell_{x_{1:n}})$ has left endpoint $a_{x_{1:n}}\in[0,1)$ and length $\ell_{x_{1:n}}\in(0,1)$ such that for every
 $n \in \mathbb N_0$ and every $x_{1:n}\in\{0,1\}^{n}$, 
\begin{equation}\label{e:def1}
 a_{x_{1:n}}=a_{x_{1:n},0},\quad a_{x_{1:n},1}=a_{x_{1:n},0}+\ell_{x_{1:n},0},\quad \ell_{x_{1:n}}=\ell_{x_{1:n},0}+\ell_{x_{1:n},1}.
\end{equation}
We call $I$ a NBP of $[0,1)$ (this abbreviation is explained in Remark~\ref{r:first} below). Finally, define the digits of every $x\in [0,1)$ as the unique sequence $(x_1,x_2,\ldots)\in\{0,1\}^\infty$ such that
$x\in I_{x_{1:n}}$ for $n=1,2,\ldots$.
\end{definition}

%Thus, $I_{x_{1:n}}=[a_{x_{1:n}},a_{x_{1:n}}+\ell_{x_{1:n}})$. 
For instance, for the usual binary numeral system, $a_{x_{1:n}}=\sum_{i=1}^n x_i2^{-i}$ and $\ell_{x_{1:n}}=2^{-n}$. Then we refer to $I$ as the standard diadic partitions. 

\begin{remark}\label{r:first} Let the situation be as in Definition~\ref{def:1}. For all $x_{1:n} \in \{0, 1\}^n$ with $n\in\mathbb N_0$, the intervals $I_{x_{1:n},0}$ and $I_{x_{1:n},1}$ are disjoint, $a_{x_{1:n},1}$ is the right endpoint of $I_{x_{1:n},0}$ and the left endpoint of $I_{x_{1:n},1}$, and $I_{x_{1:n}}=I_{x_{1:n},0}\cup I_{x_{1:n},1}$. We call $I_{x_{1:n},0}$ the left interval of this split. For every $n\in\mathbb N$, 
\begin{equation}\label{e:1}
[0,1)=\bigcup_{x_{1:n}\in\{0,1\}^n}I_{x_{1:n}}, 
\end{equation}
where the $2^n$ sets on the right hand side are  pairwise disjoint sets, and so by \eqref{e:def1}, the
infinite sequence
 $I^{(1)},I^{(2)},\ldots$ constitutes a collection of  \underline{n}ested \underline{b}inary \underline{p}artitions of $[0,1)$ (explaining the abbreviation NBP).
By  \eqref{e:def1},
the left endpoints of the intervals in $I$ are ordered in accordance to $\prec$, i.e., for any  $x_{1:n} \in \{0, 1\}^n$ and  $z_{1:m} \in \{0, 1\}^m$ with $n, m\in\mathbb N$ and $x_{1:n} \prec z_{1:m}$, we have $a_{x_{1:n}}<a_{z_{1:m}}$.
By  \eqref{e:def1}, for every $n\in\mathbb N$, $I^{(n)}$ is determined by $(\ell_{x_{1:j-1},0}\mid x_{1:j}\in \{0,1\}^j,j=1,\ldots,n)$ which is the ordered sequence of the 
lengths of the left intervals up to level $n$.
 \end{remark}
 
\begin{remark}\label{r:1} 
Henceforth, assume for every $x\in [0,1)$ with digits $(x_1,x_2,\ldots)$ that $\ell_{x_{1:n}}\to0$ as $n\to\infty$. Hence, $x$ is determined  by its digits and there is a one-to-one correspondence between $[0,1)$ and $\{0,1\}^\infty$. We express this by writing $x=.x_1x_2\ldots$. %Note that $x=\lim_{n\to\infty}a_{x_{1:n}}$.

  Definition~\ref{def:1} may be extended by  allowing an interval $I_{x_{1:n}}$ to be empty. Such an extension is relevant for certain number systems, e.g.\ in relation to so-called pseudo $\beta$-expansions, cf.\ \cite{HerbstEtAl2} and the references therein. It is  rather straightforward to modify the ideas and results in this paper to such situations.
\end{remark} 

\begin{remark} The remainder of this paper considers the following setting. Let $X$ be a random variable with a PDF $f$ concentrated on $H\subseteq [0,1)$ where $H$ has Lebesgue measure 1, and let $X_1,X_2,\ldots$ be the random digits of $X=.X_1X_2\ldots$. 
Assume that $f$ is lower semi-continuous (LSC) on $H$. Indeed, this is a very mild condition, cf.\ \cite{HerbstEtA23}. 
\end{remark} 

We need some notation for the following theorem.
Define $c_\emptyset=\inf_{H} f$.
For $n\in\mathbb N$ and $x_{1:n}\in\{0,1\}^n$, let $i_{{1:n}}=\inf_{H\cap I_{x_{1:n}}}f$ and $c_{x_{1:n}}=i_{{1:n}}-i_{{1:n-1}}$, where we let the infimum over the empty set be 0.  
%\[
%c_{x_{1:n}}=i_{{1:n}}-i_{{1:n-1}}\]
%{where }
%\[i_{{1:n}}=\begin{cases}
%\inf_{H\cap I_{x_{1:n}}}f&\text{if }H\cap I_{x_{1:n}}\not=\emptyset\\
%0&\text{if }H\cap I_{x_{1:n}}=\emptyset.
%\end{cases}
%\]
 Denote $2^{\mathbb N_0}$ the set of all subsets of $\mathbb N_0$, and  $\mathcal B$ the set of Borel sets included in $H$. Equip the product space $H\times \mathbb N_0$ with the product $\sigma$-algebra of $\mathcal B$ and $2^{\mathbb N_0}$. Let $\eta$ be the product measure of 
Lebesgue measure on $\mathcal B$ and counting measure on $2^{\mathbb N_0}$. We use the abbreviation PMF for a probability mass function.
 
\begin{theorem}\label{t:1} Under the conditions above, there is a
coupling of $X$ with a random variable $N\in\mathbb N_0$ such that the distribution of $(X,N)$ is absolutely continuous with respect to $\eta$, with a density for any $(x,n)\in H\times \mathbb N_0$ given by
\begin{equation}\label{e:res1}
f(x,n)=c_{x_{1:n}} \quad\mbox{if }x=.x_1x_2\ldots
\end{equation}
%and the remainder $0.X_{N+1}X_{N+2}\ldots$ is $\mathrm{Unif}(0,1)$-distributed.

Conversely, suppose $Q$ is a probability measure on $H\times \mathbb N_0$ which is absolutely continuous with respect to $\eta$ such that 
its density $q={\mathrm dQ}/{\mathrm d\eta}$ is of the form
\[q(x,n)= g_n(x)\quad\mbox{for all }(x,n)\in H\times\mathbb N_0\]  
where each $g_n$ is a non-negative LSC function on $H$ with finite Lebesgue measure. 
Defining the PDF
\[g(x)=\sum_{n=0}^\infty g_n(x)\quad\mbox{for all }x\in H\]
and letting $N$ be a random variable with PMF
 $p_n=\int_H g_n$ for $n\in\mathbb N_0$, 
then $g$ is LSC on $H$ and  $X\mid N=n$ has PDF $f_n=g_n/p_n$ for $p_n>0$.
In particular, \eqref{e:res1} is the special case where 
each $g_n$ is a non-negative constant function on each interval $I_{x_{1:n}}$ such that $\sum_{n=0}^\infty\int_0^1g_n(x)\,\mathrm dx=1$. 
\end{theorem}

The proof of Theorem~\ref{t:1} is found in the Supplemental material.

\begin{remark}\label{r:suff}
Clearly, $(X,N)$ and $(X_{1:N},.X_{N+1}X_{N+2}\ldots)$ are in a one-to-one correspondence.
Thinking of $X$ as data and $N$ as a latent variable, considering  a statistical model of $f$ and imaging we could observe the missing data $N$, 
 \eqref{e:res1} shows that $X_{1:N}$ is a sufficient statistic with distribution 
\[\mathrm P(X_{1:N}=x_{1:n})=\ell_{x_{1:n}}c_{x_{1:n}}\quad\mbox{for every $n\in\mathbb N_0$ and every $x_{1:n}\in\{0,1\}^n$.}\] 
 On the other hand, 
 the remainder $.X_{N+1}X_{N+2}\ldots$ is an ancillary statistic, and it is shown in \cite{moeller} \citep[see also][]{HerbstEtA23,HerbstEtAl} that 
\begin{equation}\label{e:unifremainder}
\mbox{$.X_{N+1}X_{N+2}\ldots\sim\mathrm{Unif}[0,1)$ is independent of $X_{1:N}$.}
\end{equation}

The second part of Theorem~\ref{t:1} specifies a more general structure, where it would be interesting to see how the ideas introduced in the present paper might be extended.
\end{remark}

We interpret $X_{1:N}$ as a finite ordered point process, where $N$ is the random number of points and the state space is $\Omega=\cup_{n=0}^\infty\{0,1\}^n$. Equip $\Omega$ with the $\sigma$-algebra $\mathcal G$ generated by the sets $A\subseteq\{0,1\}^n$ with $n\in\mathbb N_0$. Let $\mathcal P_{\mathrm{FOPP}}$ denote the class of all probability distributions on $(\Omega,\mathcal G)$ and $\mathcal P_{\mathrm{LSC}}$ the class of all absolutely continuous probability distributions on $[0,1)$ with a density which is almost everywhere LSC. Theorem~\ref{t:1} gives immediately the following corollaries.

\begin{corollary}\label{c:t:1}
The coupling construction of $(X,N)$ in Theorem~\ref{t:1} establishes a one-to-one correspondence between $\mathcal P_{\mathrm{FOPP}}$ and $\mathcal P_{\mathrm{LSC}}$. 
\end{corollary}

\begin{remark}
Henceforth, since $N$ is not observable, we impose a prior PMF  
$$p_n=\mathrm P(N=n),\quad n\in\mathbb N_0.$$
\end{remark}

\begin{corollary}\label{c:t:2}
Under the coupling construction of $(X,N)$ in Theorem~\ref{t:1},  
$X$ conditioned on $N$ has PDF 
\begin{equation*}\label{e:random-PDF-gen}
f(x\mid N=n)= \mathrm P(X_{1:n} = x_{1:n} )/{\ell_{x_{1:n}}}\quad\mbox{if $x=.x_1x_2\ldots\in[0,1)$ and $n\in\mathbb N_0$.}
\end{equation*} 
\end{corollary}

Corollary~\ref{c:t:2} is in accordance with \eqref{e:unifremainder} and to complete the description of the distribution of $X$ it remains only for every $n\in\mathbb N$ to specify the PMF $\mathrm P(X_{1:n} = x_{1:n} )$ for $x_{1:n}\in\{0,1\}^n$. This is the subject of Section~\ref{s:models}.

\section{Models}\label{s:models}

\subsection{Finite P{\'o}lya tree distributions}\label{s:finite_pt}\label{s:3.1}

%For the definition of a finite P{\'o}lya tree distribution, we need first the following Definitions~\ref{def:1} and \ref{def:beta_seq}. For notational convenience, let  $[0,1)=[0,1)$ be the half-open unit interval with left endpoint $a_{\emptyset}=0$ and length $\ell_\emptyset=1$.  

\begin{definition}\label{def:beta_seq} 
Let $\alpha=(\alpha_{x_{1:n}}\mid x_{1:n}\in\{0,1\}^n,\,n\in\mathbb N)$ be a given parameter such that 
%$\alpha_k>0$ for $k=0,1$ and
 $\alpha_{x_{1:n}}\ge0$ and $\alpha_{x_{1:n-1},0}+\alpha_{x_{1:n-1},1}>0$ whenever $x_{1:n}\in\{0,1\}^n$ and $n\in\mathbb N$. %whenever $n>1$. 
 Suppose $$Y=(Y_{x_{1:n}}\mid x_{1:n}\in\{0,1\}^n,\,n\in\mathbb N)$$ is a stochastic process with each $Y_{x_{1:n-1},0}\sim \mathrm{Beta}(\alpha_{x_{1:n-1},0},\alpha_{x_{1:j-1},1})$ and $Y_{x_{1:j-1},1}=1-Y_{x_{1:j-1},0}$, and where the $Y_{x_{1:n-1},0}$'s with $n\in\mathbb N$ are independent. 
Then $Y$ is said to be
	 an infinite \underline{b}eta-distributed \underline{s}equence with parameter $\alpha$, and for any $n\in\mathbb N$,
	$Y^{(n)}$ to be a \underline{f}inite
	  \underline{b}eta-distributed \underline{s}equence with parameter $\alpha^{(n)}$ and $Y^{(>n)}$ to be an infinite
	  \underline{b}eta-distributed \underline{s}equence with parameter $\alpha^{(>n)}$.
	For short, write $$Y\sim \mathrm{BS}(\alpha),\quad Y^{(n)}\sim \mathrm{FBS}(\alpha^{(n)}),\quad Y^{(>n)}\sim \mathrm{BS}(\alpha^{(>n)}).$$
The finite beta-distributed sequence induces a random PDF {\color{black}on $[0,1)$ which is given by}
\begin{equation}\label{e:finite_pt_n0}
	f\big(x \mid Y^{(n)} \big) %=  p(x_{1:N}\mid Y^{(N)})\times\frac{1}{\ell_{x_{1:N}}} % p(.0\ldots0x_{N+1} x_{N+2} \ldots \mid Y^{(N)})\nonumber\\
	=\frac{\prod_{j=1}^n Y_{x_{1:j-1},0}^{1-x_j}Y_{x_{1:j-1},1}^{x_j}}{\ell_{x_{1:n}}}\quad\mbox{if $x =.x_1x_2\ldots\in[0,1)$}
\end{equation} 
and we write 
$$f\big(\cdot \mid Y^{(n)} \big)\sim\Pi^{(n)}=\mathrm{FPT}(\alpha^{(n)},I^{(n)})$$
 for its distribution which is called a \underline{f}inite \underline{P}{\'o}lya \underline{t}ree with parameter $(\alpha^{(n)},I^{(n)})$.
\end{definition}

%In connection to Definition \ref{def:beta_seq}, % and Remark \ref{def:cpt1}, 
%recall that $\alpha_{x_{1:j-1},x_j}$ is allowed to be 0: If $\alpha_{x_{1:j-1},x_j}=0$ then $Y_{x_{1:j-1},x_j}=x_j$ (see Section~\ref{s:notation}) which implies $X_j=1-x_j$ . %, cf.\ Definition~\ref{def:cpt1}.

\begin{remark}\label{def:cpt1} The random PDF $f\big(\cdot \mid Y^{(n)} \big)$ is constant on every interval $I_{x_{1:n}}$ with $x_{1:n}\in\{0,1\}^n$.
Using $\Pi^{(n)}$ %\mathrm{FPT}(\alpha^{(n)},I^{(n)})$ 
as a prior distribution
generates a distribution of $X=.X_1X_2\ldots$ where $X\mid Y^{(n)}\sim f(\dot\mid Y^{(n)})$. Then
by \eqref{e:finite_pt_n0}, for $j=1,\ldots,n$, conditioned on both $Y^{(n)}$ and $X_{1:j-1}$, the probability that $X_j = 0$ is $Y_{X_{1:j-1}, 0}$, and $.X_{n+1}X_{n+2}\ldots\sim\mathrm{Unif}[0,1)$ is independent of $X_{1:n}$.
Recall that $\alpha_{x_{1:j-1},x_j}$ is allowed to be 0: If $\alpha_{x_{1:j-1},x_j}=0$ then $Y_{x_{1:j-1},x_j}=x_j$ (see Section~\ref{s:notation}) which implies $X_j=1-x_j$ . %, cf.\ Definition~\ref{def:cpt1}.
\end{remark}

\begin{definition}\label{def:aaaa}
Letting $n\to\infty$ in Definition~\ref{def:beta_seq} , we recover the definition of a P{\'o}lya tree $\Pi^{(\infty)}=\mathrm{PT}(\alpha,I)$ as given by \cite{lavine}. We refer to
$\Pi^{(\infty)}$ as the PT prior. % and to $f\big(\cdot \mid Y^{(N)} \big)$ as the random GFPT1 PDF.
 \end{definition}
 
 \begin{remark}\label{r:bbb} 
Arguments for considering a finite P{\'o}lya tree distribution can be found in \cite{lavine2} and \cite{WongMa}. For instance,  $\Pi^{(\infty)}$ %$\mathrm{CPT}(Y,I)$ 
is not always the distribution of a random PDF unless the $\alpha$'s increase sufficiently rapidly which leads to problems of robustness, cf.\ \cite{lavine2}. Further, if $\Pi^{(\infty)}$ is the distribution of a random PDF, this 
 random PDF is almost surely discontinuous almost everywhere, cf.\ \cite{Ferguson1974}. 
%In contrast, for every $n\in\mathbb N_0$, $\mathrm{CFPT}(Y^{(n)},I^{(n)})$ has a random PDF which is constant on every interval $I_{x_{1:n}}$ with $x_{1:n}\in\{0,1\}^n$ (this random PDF is specified by \eqref{e:finite_pt_N} below). 
\end{remark}

\begin{definition}\label{def:cpt2} Let $\mathrm{FBS}(\alpha^{(0)})$ denote the degenerated distribution concentrated at $Y^{(0)}=\emptyset$. In accordance with Corollary~\ref{c:t:2}, extend \eqref{e:finite_pt_n0} to the case $n=0$ by setting $\prod_{j=1}^n\cdots=1$ so that $f\big(x \mid Y^{(0)} \big)=1$ is the uniform PDF and $\Pi^{(0)}=\mathrm{FPT}(\alpha^{(0)},I^{(0)})=\{f\big(\cdot \mid Y^{(0)} \big)\}$. Here, $\alpha^{(0)}$ and $I^{(0)}$ have no meaning and are just introduced for notional convenience.
\end{definition}

\subsection{Generalized finite P{\'o}lya tree distributions}\label{s:generelized_fpt}

{\color{black}This section generalizes the finite P{\'o}lya tree distribution in two directions by assuming that the truncation level $n$ is random and possibly also the nested binary partitions $I^{(1)},I^{(2)},\ldots$ are random. 
The motivation for using a random truncation level $N$ is given by Theorem~\ref{t:1} and Corollary~\ref{c:t:1}.}
Note that $n\in\mathbb N_0$ is determined by $Y^{(n)}$, since $Y^{(n)}$ has dimension $2^n$.
   
\begin{definition}\label{def:5} Given a NBP $I$ as in Definition~\ref{def:1}, an infinite beta-distributed sequence 
  $Y\sim\mathrm{BS}(\alpha)$ as in Definition~\ref{def:beta_seq}, $Y_0=\emptyset$ as in Definition~\ref{def:cpt2}, 
  and a discrete random variable $N$ on $\mathbb N_0$ which is independent of $Y$, 
then   
 $(N,Y)$ induces a random PDF
depending only on $Y^{(N)}$ and given by
\begin{equation}\label{e:finite_pt_N}
	f\big(x \mid Y^{(N)} \big) %=  p(x_{1:N}\mid Y^{(N)})\times\frac{1}{\ell_{x_{1:N}}} % p(.0\ldots0x_{N+1} x_{N+2} \ldots \mid Y^{(N)})\nonumber\\
	=\frac{\prod_{j=1}^N Y_{x_{1:j-1},0}^{1-x_j}Y_{x_{1:j-1},1}^{x_j}}{\ell_{x_{1:N}}}\quad\mbox{if $x =.x_1x_2\ldots\in[0,1)$}.
\end{equation} 
 We write 
\begin{equation}\label{e:prior11}
f\big(\cdot \mid Y^{(N)} \big)\sim\Pi_1= \mathrm{GFPT1}(\alpha,I)
\end{equation}
 for its distribution which is called a \underline{g}eneralized \underline{f}inite \underline{P}{\'o}lya \underline{t}ree of type 1 with parameter $(\alpha,I)$. We refer to $f\big(\cdot \mid Y^{(N)} \big)$ as the random GFPT1 PDF and to $\Pi_1$ as the GFPT1 prior. We also refer 
 to the distribution of $(Y,N)$ as the GFPT1 prior (it will be clear by the context what is meant). 
\end{definition} 

{\color{black}
For prior elictation pruposes, it might be of interest to understand the finite dimensional laws of $f(\cdot\mid Y^{(N)})$. The next proposition sheds light on some aspects of such finite dimensional laws by providing an explicit expression for the mean.
\begin{proposition}\label{prop:mean}
	Let $P(\cdot)$ be the random probability measure given by the  random PDF %density $f(\cdot\mid Y^{(N)})$ as 
	in \eqref{e:finite_pt_N}. Using notation as in Definition~\ref{def:1}, for every $d\in\mathbb N$ and $x_{1:d}\in\{0,1\}^d$, we have 
\[
\mathrm E\big[P(I_{x_{1:d}})\big]
=\sum_{n=0}^{d-1} p_n\,\frac{\ell_{x_{1:d}}}{\ell_{x_{1:n}}}\,\prod_{j=1}^{n}\frac{\alpha_{x_{1:j}}}{\alpha_{x_{1:j-1},0}+\alpha_{x_{1:j-1},1}}
+\mathrm P(N \ge d) \prod_{j=1}^{d}\frac{\alpha_{x_{1:j}}}{\alpha_{x_{1:j-1},0}+\alpha_{x_{1:j-1},1}}.
\]	
%		For $x_{1:j} \in \{0, 1\}^j$ define $B_{x_{1:j}}=(a_{x_{1:j}},a_{x_{1:j}}+\ell_{x_{1:j}})$. Then,
%	\[
%\mathrm E\big[P(B_{x_{1:d}})\big]
%=\sum_{n=0}^{d-1} p_n\,\frac{\ell_{x_{1:d}}}{\ell_{x_{1:n}}}\,\prod_{j=1}^{n}\frac{\alpha_{x_{1:j}}}{\alpha_{x_{1:j-1},0}+\alpha_{x_{1:j-1},1}}
%+\mathrm P(N \ge d) \prod_{j=1}^{d}\frac{\alpha_{x_{1:j}}}{\alpha_{x_{1:j-1},0}+\alpha_{x_{1:j-1},1}}.
%\]
%% Moreover,
%% \begin{align*}
%% 	\mathrm E\big[P(B_{x_{1:d}})^2\big]
%% &=\sum_{n=0}^{d-1} p_n\,\Big(\frac{\ell_{x_{1:d}}}{\ell_{x_{1:n}}}\Big)^{2}\,\prod_{j=1}^{n}\frac{\alpha_{x_{1:j}}(\alpha_{x_{1:j}}+1)}{(\alpha_{x_{1:j-1},0}+\alpha_{x_{1:j-1},1})(\alpha_{x_{1:j-1},0}+\alpha_{x_{1:j-1},1}+1)} \\
%% &+\mathrm P(N \ge d)\,\prod_{j=1}^{d}\frac{\alpha_{x_{1:j}}(\alpha_{x_{1:j}}+1)}{(\alpha_{x_{1:j-1},0}+\alpha_{x_{1:j-1},1})(\alpha_{x_{1:j-1},0}+\alpha_{x_{1:j-1},1}+1)}.
%% \end{align*}
%% Hence
%% \[
%% \operatorname{Var}\big[P(B_{x_{1:d}})\big]
%% =\mathbb E\big[P(B_{x_{1:d}})^2\big]-\big(\mathbb E[P(B_{x_{1:d}})]\big)^2.
%% \]
\end{proposition}
Along the same lines of Proposition~\ref{prop:mean} it is possible to obtain analytical expressions for higher order moments of $P(I_{x_{1:d}})$, including the variance. However, their expression is rather complex which makes them less useful for prior elicitation purposes.
By the $\pi$-$\lambda$ theorem, the expectation of $P(\cdot)$ can be extended to all measurable sets.
Observe that the expression in Proposition~\ref{prop:mean} reduces to the one in \cite{lavine} if $N=\infty$ almost surely.
}

\begin{remark}\label{rem:bayes_model}
Let $X\sim f\big(\cdot \mid Y^{(N)} \big)\sim \Pi_1$. Conditioned on $(Y,N)$ (or just $Y^{(N)}$), we have that 
\begin{equation}\label{e:sdf}
.X_{N+1}X_{N+2}\ldots\sim\mathrm{Unif}[0,1)
\end{equation}
 is independent of $X_{1:N}$ which follows the random PMF given by the numerator in \eqref{e:finite_pt_N}. This is in agreement with \eqref{e:unifremainder}, and
the law of $X$ can be expressed by a Bayesian hierarchical model with three steps:  
\begin{align}\label{eq:gfpt1_bayes}
	X \mid Y^{(N)}&\sim f\big(x \mid Y^{(N)}\big)  \\
		Y^{(N)}\mid N &\sim \mathrm{FBS}(\alpha^{(N)})\nonumber\\
	N & \sim p\nonumber
\end{align}
where $p=(p_0,p_1,\ldots)$ is the PMF of $N$. Here, 
$Y^{(N)}$ takes the interpretation of an unobserved/latent parameter to which a  prior distribution is assigned in the latter two steps, so that $f\big(x \mid Y^{(N)}\big)$ follows a $\mathrm{GFPT1}(\alpha,I)$ prior. Moreover, $Y^{(>N)}\mid (X,Y^{(N)})\sim\mathrm{BS}(\alpha^{(>N)})$ depends only on $N$.

Given
a sample $Z_1,\ldots,Z_m$ which share the same latent variables $Y$ and $N$ so that
\begin{equation}\label{e:sample1}
Z_1,\ldots,Z_m\mid Y,N\iid f\big(\cdot \mid Y^{(N)} \big)\sim\Pi_1,
\end{equation}  
the inferential goal is to derive the posterior distribution of $(Y,N)$. 
The posterior distribution of the random GFPT1 PDF is denoted $\Pi_1(\cdot \mid Z_{1:m})$ and called the posterior GFPT1 PDF:
\begin{equation}\label{e:postrandomPDF}
f\big(\cdot \mid Y^{(N)} \big) \mid Z_{1:m}  \sim \Pi_1(\cdot \mid Z_{1:m}).
\end{equation}
Our objectives in this paper are the posterior GFPT1 PDF and the posterior distribution of $N$  under the GFPT1 prior.
\end{remark}

\begin{remark}
	{\color{black}The introduction of the latent variable $N$ is clearly motivated by the coupling result in Theorem~\ref{t:1}. However, $N$ plays a crucial role also as a smoothing parameter that improves the fitting process even when it is not directly a parameter of the data generaing process.
% 	Indeed, from Proposition~\ref{prop:mean_var}, it becomes apparent how the law of $N$ influences the prior variance: if $N = \infty$ (as in the case of standard PT) or, more generally, $N \ge d$ almost surely, then 
% 	\begin{equation}\label{eq:var_largeN}
% 	\begin{aligned}
% 		\operatorname{Var}\big[P(B_{x_{1:d}})\big]
% &=\prod_{j=1}^{d}\frac{\alpha_{x_{1:j}}(\alpha_{x_{1:j}}+1)}{(\alpha_{x_{1:j-1},0}+\alpha_{x_{1:j-1},1})(\alpha_{x_{1:j-1},0}+\alpha_{x_{1:j-1},1}+1)}\\
% &-\Big(\prod_{j=1}^{d}\frac{\alpha_{x_{1:j}}}{\alpha_{x_{1:j-1},0}+\alpha_{x_{1:j-1},1}}\Big)^2.
% 	\end{aligned}
% 	\end{equation}
% 	It can be checked that, when $I$ is the standard nested dyadic partition and $\alpha_{x_{1:n}} = \alpha_0 n^2$ is fixed to the default value as discussed in Section~\ref{s:elicitation} below, then the variance in \eqref{eq:var_largeN} is greater or equal than the variance in the general of Proposition~\ref{prop:mean_var}.
% 
	Specifically, the introduction of a prior on $N$ allows to learn a posteriori the smoothness required to represent the density balancing over-fitting (i.e., extremely wiggly density estimates which typically occur with standard PT posteriors) and under-fitting (i.e., density estimates that do not capture important nuisances in the data, which happens with finite PT posteriors if $n$ is too small).}
\end{remark}

So far the NBP $I$ has be given. To obtain further flexibility, randomness of $I$ may be imposed:
For every $n\in\mathbb N$ and every $x_{1:n}\in\{0,1\}^n$, suppose $I_{x_{1:n}}$ has a random length $L_{x_{1:n}}$, let 
\begin{equation}\label{e:xx}
R_{x_{1:n}}=L_{x_{1:n}}/L_{x_{1:n-1}}
\end{equation}
 with $L_{x_{1:0}}=1$, and let 
 $R=(R_{x_{1:n}}\mid x_{1:n}\in\{0,1\}^n,n\in\mathbb N)$. Then there is a one-to-one correspondence between $I^{(n)}$ and $R^{(n)}$, and between $I$ and $R$. 
 Specifically, consider the following case, where we let 
 $\mathrm{FBS}(\beta^{(0)})$ denote the degenerated distribution concentrated at the empty set (sequence) $R^{(0)}=\emptyset$.

 \begin{definition}\label{def:6} Let $\beta=(\beta_{x_{1:n}}\mid x_{1:n}\in\{0,1\}^n,\,n\in\mathbb N)$ be a given parameter with each $\beta_{x_{1:n}}>0$.
Suppose Definition~\ref{def:5} is extended such that $R\sim\mathrm{BS}(\beta)$ is independent of $(N,Y)$, and 
$(N,Y,R)$ induces a random PDF
depending only on $(Y^{(N)},R^{(N)})$ and given by
\begin{equation}\label{e:finite_pt_N2}
		f\big(x \mid Y^{(N)},R^{(N)} \big) = \frac{\prod_{j=1}^N Y_{x_{1:j-1},0}^{1-x_j}Y_{x_{1:j-1},1}^{x_j}}{\MBtext{}L_{x_{1:N}}}\quad\mbox{if $x =.x_1x_2\ldots\in[0,1)$}
\end{equation}
where $x_{1:N}$ is specified by $R^{(N)}$. 
 We write 
$$f\big(\cdot \mid Y^{(N)},R^{(N)} \big)\sim\Pi_2= \mathrm{GFPT2}(\alpha,\beta)$$
 for its distribution which is called a \underline{g}eneralized \underline{f}inite \underline{P}{\'o}lya \underline{t}ree of type 2 with parameter $(\alpha,\beta)$.
 We refer to $f\big(\cdot \mid Y^{(N)},R^{(N)} \big)$ as the random GFPT2 PDF and to
$\Pi_2$ as the GFPT2 prior. We also refer
 to the distribution of $(Y,R,N)$ as the GFPT2 prior (it will be clear by the context what is meant). 
 \end{definition}   
 
 \begin{remark}\label{r:4}
 Let
$X\sim \Pi_2\sim\mathrm{GFPT2}(\alpha,\beta)$.
 Conditioned on $(Y,R,N)$ (or just on $(Y^{(N)},R^{(N)})$), $.X_{N+1}X_{N+2}\ldots\sim\mathrm{Unif}[0,1)$  is independent of $X_{1:N}$ which 
 follows the random PMF given by the numerator in \eqref{e:finite_pt_N2}.
The law of $X$  
 can be expressed by a Bayesian hierarchical model:  
\begin{equation*}
\begin{aligned}
	X \mid Y^{(N)}, R^{(N)}&\sim f\big(x \mid Y^{(N)}, R^{(N)}\big)  \\
		Y^{(N)}\mid N\sim \mathrm{FBS}(\alpha^{(N)})\ & \mbox{and }R^{(N)}\mid N \sim \mathrm{FBS}(\beta^{(N)})\ \mbox{are independent} \\
	N & \sim p
\end{aligned}
\end{equation*}
where %$f\big(x \mid Y^{(N)},I^{(N)}\big)$ is given by the right hand side of \eqref{e:finite_pt_N} and $p=(p_0,p_1,\ldots)$ is the PMF of $N$. Here 
$(Y^{(N)},R^{(N)})$ takes the interpretation of an unobserved/latent parameter to which a prior distribution is assigned in the latter two steps, so that $f\big(x \mid Y^{(N)},R^{(N)}\big)$ follows a $\mathrm{GFPT2}(\alpha,\beta)$ prior. Conditioned on $(X,Y^{(N)},R^{(N)})$ we have that $Y^{(>N)}\sim\mathrm{BS}(\alpha^{(>N)})$ and $R^{(>N)}\sim\mathrm{BS}(\beta^{(>N)})$ are independent and depend only on $N$.

Given
a sample $Z_1,\ldots,Z_m$ which share the same latent variables $Y$, $R$, and $N$ so that
\begin{equation}\label{e:sample2}
Z_1,\ldots,Z_m\mid Y,R,N\iid f\big(\cdot \mid Y^{(N)},R^{(N)} \big)\sim\Pi_2,
\end{equation}
%$Z_j\sim \Pi^{(N)}$ %\mathrm{GFPT1}(\alpha,I)$ 
%for $j=1,\ldots,m$, which share the same latent variables $Y$ and $N$ so that
%$Z_1,\ldots,Z_m\mid (Y,N)\iid \mathrm{CFPT}(Y^{(N)},I^{(N)})$,   
the inferential goal is to derive the posterior distribution of $(Y,R,N)$. 
 %We refer to $f\big(\cdot \mid Y^{(N)},R^{(N)} \big)$ as the prior GFPT2 PDF (of a data point) and to $\Pi_2$ as the GFPT2 prior.
The posterior distribution of the GFPT2 PDF is denoted $\Pi_2(\cdot \mid Z_{1:m})$ and called the posterior GFPT2 PDF:
\begin{equation}\label{e:postrandomPDF2}
f\big(\cdot \mid Y^{(N)},R^{(N)} \big) \,\big|\, Z_{1:m}  \sim \Pi_2(\cdot \mid Z_{1:m}).
\end{equation}
In this paper, we focus on the posterior GFPT2 PDF and the posterior distribution of $N$ under the GFPT2 prior, although the posterior distribution of $R$ may possibly be of interest as well.
\end{remark}
 
  \begin{remark}\label{r:5}
Extending Definition~\ref{def:6} by
 allowing $\beta_{x_{1:n}}$ to be 0 would correspond to relaxing Definition~\ref{def:1} so that $I_{x_{1:n}}$ may be empty, in which case we should let $X_n=1-x_n$. The ideas and results in this paper may rather easily be extended to this situation.
 %This explains why there is no gain in extending Definition~\ref{def:1} by allowing an interval in $I_{x_{1:n}}$ to be empty, cf.\ Remark~\ref{r:1}.
 \end{remark}

\subsection{Specification of parameters}\label{s:elicitation}

It remains to specify the parameters $\alpha$ and $\beta$ of the PT, GFPT1, and GFPT2 priors introduced in Sections~\ref{s:3.1} and \ref{s:generelized_fpt}, and to specify a prior distribution of $N$. This section introduces some specific choices of $\alpha$ and $\beta$, which are later used in Section~\ref{s:num}. Priors for $N$ are discussed later in Remark~\ref{r:support} and in Section~\ref{s:num}.

%For the choice of $\alpha$, let us first focus on the case of $\mathrm{PT}(\alpha,I)$. Then for any $n\in\mathbb N$ and $x_{1:n}\in\{0,1\}^n$, by conditioning on $Y$ and taking expectation, Definition~\ref{def:beta_seq} gives
%\begin{equation}\label{e:PTpmf}
%\mathrm P(X_{1:n}=x_{1:n})=\prod_{j=1}^n \frac{\alpha^{1-x_j}_{x_{1:j-1},0}\alpha^{x_j}_{x_{1:j-1},1}}{\alpha_{x_{1:j-1},0}+\alpha_{x_{1:j-1},1}}.
%\end{equation}
%Thus, for any $n\in\mathbb N$ and $x_{1:n}\in \{0,1\}^n$, writing
%\[\alpha_{x_{1:n}}=c_{x_{1:n-1}}\tilde p_n(x_{1:n})\]
%where $c_{x_{1:n-1}}>0$ and $\tilde p_n$ is a PMF on $\{0,1\}^n$, we see that $\mathrm{PT}(\alpha,I)$ does not depend on the choice of $c_{x_{1:n-1}}$. Similarly, the distributions $\mathrm{FPT}(\alpha^{(n)},I^{(n)})$, $\mathrm{GFPT1}(\alpha,I)$, and 
%$\mathrm{GFPT2}(\alpha,\beta)$ do not depend on the choice of $c_{x_{1:n-1}}$. However, the variance of $Y_{x_{1:n}}$ will depend on this choice. %In the following we let for simplicity $c_{x_{1:n-1}}$ depend on $n$ only. 

Recalling Remark~\ref{r:bbb}, following common practice, setting
\begin{equation}\label{e:alpha-first}
\alpha_{x_{1:n-1},0} = \alpha_{x_{1:n-1},1} = \alpha_0 n^2,\quad x_{1:n-1}\in\mathbb N^{n-1},\quad n\in\mathbb N,
\end{equation}
for some $\alpha_0>0$, then $\Pi^{(\infty)}$ is almost surely absolutely continuous, cf.\ \cite{lavine}. {\color{black}Other choices include $\alpha_{x_{1:n}} = n^\delta$ or  $\alpha_{x_{1:n}} = \delta^n$, see \cite{Watson} for further details. If $I$ is given by the standard dyadic partitions, an application of Proposition~\ref{prop:mean} entails that all such choices where $\alpha_{x_{1:n}}$ depends only on $n$ lead to a prior mean of the GFPT1 random density equal to the uniform density.
}

%(see Definition~\ref{def:aaaa})
 
Another choice is inspired by the extended Newcomb-Benford law (the general significant-digit law of \cite{Hill}): 
Recall that in the present paper we consider binary number representations, cf.\ Section~\ref{s:intro}, but 
%everything in this paragraph 
the case of the usual binary numeral system extends to the general case of base-$q$ number representations with $q\ge 2$ an integer, so to make this point clear we write $q$ instead of 2 (later in Sections~\ref{sec:mbpt} and \ref{sec:comparison_real_data} we consider results for the most common cases $q=2$ and $q=10$).
Consider $X=.X_1X_2\ldots=\sum_{i=1}^\infty X_i q^{-i}$, where $q^{-i}$ is the interval length of the $i$th partition of $[0,1)$.
Extend $X$ by 
considering a continuous random variable $Z>0$ 
with base-$q$ number representation $Z=q^{M+1}X$, where $M$ is its order of magnitude and $X_1\not=0$ is its leading digit. Suppose $Z$ satisfies the extended Newcomb-Benford law (in base-$q$), which means that $\log_q (qX)$ is uniformly distributed between 0 and 1 -- one says that $Z$ spans all orders of magnitude. Equivalently,
for all $n\in\mathbb N$ and $x_{1:n}\in\{0,\ldots,q-1\}^n$, 
\begin{equation}\label{e:benford}
	p_n(x_{1:n}) = \begin{cases}
	\log_{q}\left[1 + \left( \sum_{i=1}^n  x_i q^{n - i}  \right)^{-1}\right] & \text{if }x_1\not=0,\\
	0 &\text{if }x_1=0.
	\end{cases}
\end{equation}
In fact many real-world datasets such as those involving incomes, city sizes, or seismic magnitudes, they span multiple orders of magnitude --    
see \url{https://testingbenfordslaw.com} for a list of real datasets obeying \eqref{e:benford}.

We exploit \eqref{e:benford} to specify $\alpha$ such that the distribution of $X_n \mid X_{1:n-1}$ under the PT prior is equal to the conditional distribution of the $n$-th digit under \eqref{e:benford}.
Formally, for $k = 0, \ldots, q$, let
\[
\mathrm	P(X_n = k \mid X_{1:n-1} = x_{1:n-1}) = \E[\mathrm P(X_n = k \mid X_{1:n-1} = x_{1:n-1}, Y)] = \frac{\alpha_{x_{1:n-1},k}}{\sum_{j=0}^{q-1} \alpha_{x_{1:n-1},j}}
\]
%be the distribution of $X_n \mid X_{1:n-1}$ under the prior, we 
and solve
\[
	\mathrm P(X_n = k \mid X_{1:n-1} = x_{1:n-1}) = \frac{p_n(x_{1:n-1},k)}{\sum_{j=1}^{q-1} p_n(x_{1:n-1},j)},
\]
which leads to
\begin{equation}\label{e:alpha_benford}
\alpha_{x_{1:n}}=\begin{cases}
c_n \frac{p_n(x_{1:n})}{\sum_{j=1}^{q-1} p_n(x_{1:n-1},j)}  & \text{if }x_1\not=0,\\
0 &\text{if }x_1=0.
\end{cases}
\end{equation}
Here, $(c_n)_{n \geq 1}$ is a sequence of user-specified positive parameters that control the variance of $Y$. Specifically, $\mathrm{Var}(Y_{x_{1:n}})$ is inversely proportional to $c_n$.
Finally, to complete the description of the distribution of $Z$, we need to specify a joint distribution of $M$ and $X$. We defer this to  
Sections~\ref{sec:mbpt} and \ref{sec:comparison_real_data}.

%We exemplify in Section~\ref{sec:comparison_real_data} how to complete the description of the distribution of $Z$ by specifying a model of $M\mid Z$. 

{Consider now a GFPT2 prior. Since we assumed every $\ell_{x_{1:n}}>0$, it is required that $\beta_{x_{1:n}}>0$ for all $n\in\mathbb N$ and $x_{1:n}\in\{0,1\}^n$. 
Assume $\E[L_{x_{1:n}}] = 2^{-n}$, that is, the expected value of the NBP coincides with the standard diadic partitions of $[0,1)$. Hence, $\beta_{x_{1:n-1},0} = \beta_{x_{1:n-1},1} = \beta_n$. In our experience, values of $R_{x_{1:n}}$ too close to zero or one lead to numerical instability issues when updating $R$ via a Metropolis-Hastings algorithm as in Remark \ref{rem:mcmc-gfpt2-collapsed}. %{rem:mcmc-gfpt2}. 
Therefore, we suggest setting $\beta_n = 2$ to avoid giving prior mass to those values of $R_{x_{1:n}}$. 
}

%{\color{blue} (R1Ma2, R2S7: Note the new Section 4.4 which earlier appeared in Section 6.)}

\subsection{Modelling multi-scale human data via Benford's law and generalized P{\'o}lya trees}\label{sec:mbpt}

Real-world datasets -- such as those involving incomes, city sizes, or seismic magnitudes -- often span multiple orders of magnitude, see \cite{Clauset}.
Modelling such data poses a challenge for Bayesian density estimation because traditional methods may have difficulty capturing both the coarse-scale variability and the fine-scale structure \citep{Tokdar}.
We show here how a simple extension of a GFPT1 prior can be used to obtain a simple, yet powerful, Bayesian model for density estimation in scenarios where data span multiple orders of magnitudes and their digits exhibit the Newcomb--Benford law.
The main idea lies in independently modelling the order of magnitude of the data and their digits, by assuming a GFPT1 prior for the digits.

In the following, we write $\mathrm{BPT}_q$ for the $\mathrm{GFPT1}(I,\alpha)$ prior obtained when $I$ is the sequence of standard base-$q$ nested partitions of the unit interval (see Section~\ref{s:elicitation}) and $\alpha$ is as in \eqref{e:alpha_benford}. We refer to $\mathrm{BPT}_q$ as the base-$q$ Benford P{\'o}lya tree.

Consider data $Z_i >0$, $i =1, \ldots, m$, represented in terms of their order of magnitudes and their base-$q$ digits as $Z_i = (M_i, \tilde Z_i)$, such that $Z_i = q^{M_i + 1} \tilde Z_i$ (see again Section~\ref{s:elicitation}). 
Assume  $$\tilde Z_1, \ldots, \tilde Z_m\mid Y,N\iid f\big(\cdot \mid Y^{(N)} \big)\sim \mathrm{BPT}_q.$$ 
Further, assume the $M_i$'s take values in $T\subset \mathbb Z$ of cardinality $|T|<\infty$ such that
\[\mathrm{P}(M_i = k \mid \omega) = \omega_k,\quad k\in T, \]
where we 
impose the prior 
\[\omega = (\omega_k \mid k \in T) \sim \mathrm{Dir}(\eta)\]
where $\mathrm{Dir}(\eta)$ denotes the Dirichlet distribution with parameter $\eta \in \mathbb R_+^{|T|}$. 
Finally, assume conditioned on $(Y,N,\omega)$ that the $M_i$'s and the $\tilde Z_i$'s are independent, and
a priori that $\omega$ is independent of $(Y,N)$.
We call the distribution of $(Y,N,\omega)$ %$(f\big(\cdot \mid Y^{(N)} \big),\omega)$ a random 
a base-$q$ \emph{\underline{m}ultiscale \underline{B}enford \underline{P}{\'o}lya \underline{t}ree} ($\mathrm{MBPT}_q$)  prior.  

The $\mathrm{MBPT}_q$ prior induces a \emph{scale invariant} (in the terminology of \cite{Hill}) random density on $T\times[0,1)$ equipped with the product of counting measure on $T$ and Lebesgue measure on $[0,1)$. This random density is given by 
\[
	f(x \mid Y^{(N)}, \omega) = \omega_k \,q^{N}\prod_{j=1}^N \prod_{d=0}^q Y^{I(\tilde x_j=d)}_{\tilde x_{1:j-1}, d}  \quad \text{if } x = (k, \tilde x)\in T\times[0,1)
\] 
 where $I[\cdot]$ denotes the indicator function. We call it the random $\mathrm{MBPT}_q$ PDF. When considering its posterior distribution, we call it the posterior $\mathrm{MBPT}_q$ PDF.
 
%Posterior computation is straightforward. Indeed,  $(Y, N)$ and $\omega$ are independent also a posteriori, with the posterior of $(Y, N) \mid Z_{1:m}$ given by a slight modification of Theorem \ref{t:2}. Moreover, $\omega \mid Z_{1:m} \sim \mathrm{Dir}(\eta^p)$ where $\eta^p_k = \eta_k + \sum_{i=1}^m I[M_i = k]$ for $k \in T$.
 
\section{Bayesian analysis}\label{s:Bayes}

 This section deals with the inferential goals discussed in Remarks~\ref{rem:bayes_model} and \ref{r:4}, considering a sample $Z_{1:m}$ of $m$ $[0,1)$-valued random variables with either a $\mathrm{GFPT1}$ or a $\mathrm{GFPT2}$ prior, cf.\ \eqref{e:sample1} and \eqref{e:sample2}.
The proofs of all theorems and one corollary in this section are found in the Supplemental material.

We use the following notation. 
Let the data be given by a a realization $z_{1:m}\in[0,1)^m$ of $Z_{1:m}$.
 For $j\in\mathbb N_0$ and $x_{1:j}\in\{0,1\}^{j}$, let $n_{x_{1:j}}(z_{1:m})$ be the number of observations $z_i$ ($i=1,\ldots,m$) falling in the interval $I_{x_{1:j}}$. By Definitions~\ref{def:5} and \ref{def:6}, we have almost surely that $n_{x_{1:j}}(Z_{1:m})=0$ if $\alpha_{x_{1:j}}=0$ and $j>1$, so assume $n_{x_{1:j}}(z_{1:m})=0$ if $\alpha_{x_{1:j}}=0$ and $j>1$.  

\subsection{Posterior simulation}\label{s:posterior}

%Posterior simulation of $(N,Y)$ under a $\mathrm{GFPT1}$ prior  and of $(N,Y,R)$ under a $\mathrm{GFPT2}$ prior  is discussed in the sequel.

\subsubsection{Posterior simulation under a GFPT1 prior}

We start by considering the posterior distribution of $(N,Y)$ for the case \eqref{e:sample1} where a GFPT1 prior has been specified. %$Z_1,\ldots,Z_m \sim \mathrm{GFPT1}(\alpha,I)$, 
%{\color{black}assuming $(\alpha,I)$ has been specified.} 
%and where $Z_1,\ldots,Z_m \mid (Y,N)\iid \mathrm{CFPT}(Y^{(N)},I^{(N)})$ with independent prior distributions for $N$ and $Y\sim\mathrm{BS}(\alpha)$. The following theorem derives the posterior distribution of first $N$ and second $Y$ conditioned on $N$.  Recall that $p=(p_0,p_1,\ldots)$ is the PMF of $N$ and $b(\cdot,\cdot)$ is the beta-function. 
%We need some notation. 
Let $\gamma(z_{1:m}, \alpha) = (\gamma_{x_{1:n}}(z_{1:m}, \alpha) \mid x_{1:j}\in\{0,1\}^j,\,j\in\mathbb N)$ be the infinite sequence with entries  
\begin{equation}
\label{e:gg}
	\gamma_{x_{1:j}}(z_{1:m}, \alpha) = 
		\alpha_{x_{1:j}}+n_{x_{1:j}}(z_{1:m}) %&\text{if }j>0,\\
		%0&\text{if }j=0.
	%\end{cases}
\end{equation}
and set $\gamma_{x_{1:0}}(z_{1:m}, \alpha)=0$.

\begin{theorem}\label{t:2} Consider the case \eqref{e:sample1}. Then
%Suppose $Z_1,\ldots,Z_m \iid  \mathrm{GFPT1}(\alpha,I)$. 
 $N\mid Z_{1:m}=z_{1:m}$ has a PMF $p(n \mid z_{1:m})$ with $p(0\mid z_{1:m})\propto p_0$ and for every $n\in\mathbb N$, 
\begin{equation}\label{e:post_n_1}
p(n \mid z_{1:m})  \propto
p_n \frac{\prod_{j=1}^n \prod_{x_{1:j-1} \in \{0, 1\}^{j-1}} \mathrm{beta}\left(
	\gamma_{x_{1:j-1},0}(z_{1:m}, \alpha), \gamma_{x_{1:j-1},1}(z_{1:m}, \alpha)\right)}{{\prod_{x_{1:n}\in \{0, 1\}^n} \ell_{x_{1:n}}^{n_{x_{1:n}}(z_{1:m})}}}.
\end{equation}
where the constant of proportionality depends on $z_{1:m}$.
Moreover, conditioned on both $N$ and $Z_{1:m}=z_{1:m}$, we have that
\begin{equation}\label{e:post-next1}
\hspace{1.25cm}Y^{(N)}\mid N,Z_{1:m}=z_{1:m}\sim\mathrm{FBS}(\gamma^{(N)}(z_{1:m},\alpha^{(N)}) )
\end{equation} 
and 
\begin{equation}\label{e:post-next2}
Y^{(>N)}\mid N,Z_{1:m}=z_{1:m}\sim\mathrm{BS}(\alpha^{(>N)})
\end{equation}
are independent.
\end{theorem}

\begin{remark}\label{r:sim1} The PMF in \eqref{e:post_n_1} does not belong to a known parametric family.
Since $Y^{(N)}$ is of varying dimension $2^N-1$, one may suggest to use the reversible jump Markov chain Monte Carlo algorithm \citep{gm94,green} for simulating from the posterior density $p(y^{(n)} \mid z_{1:m}) \propto p(y^{(n)},z_{1:m})$, cf.\ \eqref{e:joint1}. However, it is much easier to use Theorem~\ref{t:2}, since
 $Y$ conditioned on $(N,Z_{1:m})$ has a conjugate BS-prior, whilst
posterior simulation of $N$ is straightforward: First, $N\mid Z_{1:m}=z_{1:m}$ can be sampled either exactly, or approximately, or in an
asymptotically exact way via a Metropolis-Hastings algorithm  as
follows. If $N_{\max}$ is an upper bound for $N$, then sample $N$ exactly from a categorical distribution over $\{0, \ldots, N_{\max}\}$ with unnormalized weights
given by the right hand side in \eqref{e:post_n_1} for $n=0,\ldots,N_{\max}$. This method may also be used to provide an approximate simulation if 
$N_{\max}$ is a user-defined suitable upper bound for $N$. Alternatively, use a Metropolis-Hastings algorithm with equilibrium density given by \eqref{e:post_n_1} (we never found a need for such an algorithm in the examples considered later in this paper). 
%with e.g.\ a proposal PMF given by $(p_0,p_1,\ldots)$ and hence a Hastings ratio given by
%\begin{equation*}
%\frac{\prod_{j=1}^{\tilde n} \prod_{x_{1:j-1} \in \{0, 1\}^{j-1}}\mathrm b\left(\gamma_{x_{1:j-1},0}(z_{1:m}, \alpha), \gamma_{x_{1:j-1},1}(z_{1:m}, \alpha)\right)}
%{\prod_{j=1}^n \prod_{x_{1:j-1} \in \{0, 1\}^{j-1}}\mathrm b\left(\gamma_{x_{1:j-1},0}(z_{1:m}, \alpha), \gamma_{x_{1:j-1},1}(z_{1:m}, \alpha)\right)} 
%\times { \frac{\prod_{x_{1:n}\in \{0, 1\}^n} \ell_{x_{1:n}}^{n_{x_{1:n}}(z_{1:m})} }{ \prod_{x_{1:\tilde n}\in \{0, 1\}^{\tilde n}} \ell_{x_{1:\tilde n}}^{n_{x_{1:\tilde n}}(z_{1:m})} } }
%\end{equation*}
%if $n$ is the current state and $\tilde n$ is the proposal. 
Second, simply use \eqref{e:post-next1} and \eqref{e:post-next2} when simulating $Y$ conditioned on $N$ and $Z_{1:m}=z_{1:m}$.
\end{remark}

The posterior mean of the random PDF in \eqref{e:finite_pt_N} is the optimal Bayesian point estimator under the squared-error loss function \citep{Robert}. The following 
corollary becomes useful for the calculation of this estimator. % and it shows that sampling of $Y$ is not needed.

\begin{corollary}\label{cor:point_est}
The posterior mean of the random GFPT1 PDF in \eqref{e:finite_pt_N} is given by
\begin{align}\label{e:PM}
\E[f(x\mid Y^{(N)}) \mid Z_{1:m} = z_{1:m}] =\sum_{n=0}^\infty \frac{p(n \mid z_{1:m})}{\ell_{x_{1:n}}}
\prod_{j=1}^n \frac{\gamma_{x_{1:j}}(z_{1:m}, \alpha)}{\gamma_{x_{1:j-1},0}(z_{1:m}, \alpha) + \gamma_{x_{1:j-1},1}(z_{1:m}, \alpha)} 
\end{align} 
if  $x=.x_1x_2\ldots\in[0,1)$.
%where $p(n \mid z_{1:m})$ is given by \eqref{e:post_n_1} and may be calculated/approximated as discussed in Remark~\ref{r:sim1}. 
\end{corollary}

\begin{remark}\label{rem:point_estimator}
In the right hand side of \eqref{e:PM}, the term $p(n \mid z_{1:m})$ %is given by \eqref{e:post_n_1} and 
may be calculated as discussed in Remark~\ref{r:sim1}. Hence, the posterior mean in \eqref{e:PM} may be determined either by a numerical approximation or by the Monte Carlo estimate of 
$$\frac{1}{\ell_{x_{1:N}}}
\prod_{j=1}^N \frac{\gamma_{x_{1:j}}(z_{1:m}, \alpha)}{\gamma_{x_{1:j-1},0}(z_{1:m}, \alpha) + \gamma_{x_{1:j-1},1}(z_{1:m}, \alpha)}$$
when $N$ is sampled from $p(n\mid z_{1:m})$. We call such an estimate for a density estimate of the posterior GFPT1 PDF (or just density estimate).
\end{remark}

\subsubsection{Posterior simulation under a GFPT2 prior}\label{s:bayes_2}

Now, consider the posterior distribution of $(Y,R,N)$ for the case \eqref{e:sample2} where a GFPT2 prior has been specified.
  The following theorem derives the posterior distribution of first $N\mid R$, second $R\mid N$, and third $Y\mid (N,R)$, where we use the following notation. 
Recall the definitions \eqref{e:xx} and \eqref{e:gg}.
% Further notation is needed.
Let $R_0=(R_{x_{1:n-1},0}\mid x_{1:n-1}\in\{0,1\}^{n-1}, n\in\mathbb N)$. 
For $j\in\mathbb N_0$, to stress the dependence of $R^{(j)}_0$ (or equivalently of $R^{(j)}$), we use capital letters: For $x_{1:j}\in\{0,1\}^n$, {\color{black}let $N_{x_{1:j}} = n_{x_{1:j}}(z_{1:m})$ be the the number of observations $z_i$ falling in $I_{x_{1:j}}$, which depends on $R^{(j)}$.
%, cf.\ the notation introduced at the beginning of Section~\ref{s:Bayes}.
Similarly, let $\Gamma_{x_{1:j}} = \gamma_{x_{1:j}}(z_{1:m},\alpha)$ and $L_{x_{1:j}} = \ell_{x_{1:j}}$, which both depend  
 on $R^{(j)}$. For $n\in\mathbb N$, $Y^{(n)}$ can be identified by $Y^{(n)}_0=(Y_{x_{1:j-1},0}\mid x_{1:j-1}\in\{0,1\}^{j-1}, j=1,\ldots,n)$
since $Y_{x_{1:j-1},1}=1-Y_{x_{1:j-1},0}$. Let $Y^{(0)}_0=Y^{(0)}=\emptyset$. For $n\in\mathbb N_0$, the dimension of $Y^{(n)}_0$ is $2^n-1$. 
	For $j\in\mathbb N$ and $x_{1:j-1}\in\{0,1\}^{j-1}$, let $\nu_{x_{1:j-1}}$ be the Lebesgue measure on $[0,1)$ if $\alpha_{x_{1:j-1},0}>0$ and $\alpha_{x_{1:j-1},1}>0$, and $\nu_{x_{1:j-1}}$ be the Dirac measure concentrated at $k\in\{0,1\}$ if $\alpha_{x_{1:j-1},k}=0$ and $\alpha_{x_{1:j-1},1-k}>0$.
	For $n\in\mathbb N$, let $A_n=[0,1)^{2^n-1}$ be equipped with the corresponding Borel $\sigma$-algebra $\mathcal F_n$. Let $\mu_n$ be the product measure on $\mathcal F_n$ given by 
	\begin{equation}\label{e:mun}
	\mu_n=\prod_{j=1}^n\prod_{{x_{1:j-1}}\in\{0,1\}^{j-1}}\nu_{x_{1:j-1}}.
	\end{equation}  
	Equip $A_0=[0,1)^0=\{\emptyset\}$ with the trivial $\sigma$-algebra $\mathcal F_0$ and
	let $\mu_0=\mathrm{FBS}(\alpha^{(0)})$, cf.\ Definition~\ref{def:cpt2}. 

\begin{theorem}\label{t:3}
Consider the case \eqref{e:sample2}. Then
 $N$ conditioned on both  $R$  and
$Z_{1:m}=z_{1:m}$ has a PMF $p(n \mid R,z_{1:m})$ with $p(0\mid R, z_{1:m})\propto p_0$ and for every $n\in\mathbb N$, 
\begin{equation}\label{e:post_n_2}
p(n \mid R,z_{1:m})  \propto p_n \frac{\prod_{j=1}^n \prod_{x_{1:j-1} \in \{0, 1\}^{j-1}} \mathrm{beta}\big(\Gamma_{x_{1:j-1},0}, \Gamma_{x_{1:j-1},1}\big)}{\prod_{x_{1:n}\in \{0, 1\}^n} L_{x_{1:n}}^{N_{x_{1:n}}}}
\end{equation}
where the constant of proportionality depends on $(R,z_{1:m})$.
Further, for any $n\in\mathbb N_0$ with $p_n>0$, conditioned on both $N=n$ and $Z_{1:m}=z_{1:m}$, we have that $R^{(n)}$ and $R^{(>n)}$ are independent,
where
\begin{equation}\label{e:r_ge_n}
\hspace{1.2cm}R^{(>n)}\mid N=n,Z_{1:m}=z_{1:m}\sim\mathrm{BS}(\beta^{>(n)})
\end{equation} 
depends only on $n$, and where
$R^{(n)}$ or more precisely
$R_0^{(n)}$ has a density with respect to $\mu_n$ which is given by 
$p(r_0^{(n)}\mid n, z_{1:m})\propto 1$ when $n=0$ and for every $r_0^{(n)}\in(0,1)^{2^n-1}$ when $n>0$ by 
	\begin{equation}\label{e:post_r_2}
	p(r_0^{(n)} \mid n, z_{1:m})\propto
	\frac{\prod_{j=1}^n  \prod_{x_{1:j-1} \in \{0, 1\}^{j-1}} \mathrm{beta}\big(\Gamma_{x_{1:j-1},0}, \Gamma_{x_{1:j-1},1}\big) r_{x_{1:j-1, 0}}^{\beta_{x_{1:j-1, 0}}-1}r_{x_{1:j-1, 1}}^{\beta_{x_{1:j-1, 1}}-1}}{\prod_{x_{1:n}\in \{0, 1\}^n} L_{x_{1:n}}^{N_{x_{1:n}}}}
	\end{equation}
	where the constant of proportionality depends on $(n, z_{1:m})$ and where $L_{x_{1:n}}$ depends on $R_0^{(n)}=r_0^{(n)}$.
%Further, for any $n\in\mathbb N_0$ with $p_n>0$, conditioned on both $Y=y$, $N=n$, and $Z_{1:m}=z_{1:m}$, we have that $R^{(n)}$ and $R^{(>n)}$ are independent,
%where
%\begin{equation}\label{e:r_ge_n}
%\hspace{1.2cm}R^{(>n)}\mid Y=y,N=n,Z_{1:m}=z_{1:m}\sim\mathrm{BS}(\beta^{>(n)})
%\end{equation} 
%depends only on $n$, and where
%$R^{(n)}$ or more precisely
%$R_0^{(n)}$ has a density with respect to $\mu_n$ which is given by $p(r_0^{(n)}\mid y,n, z_{1:m})\propto 1$ when $n=0$ and for every $r_0^{(n)}\in(0,1)^{2^n-1}$ when $n>0$ by 
%\begin{equation}\label{e:post_r_2} 
%p(r_0^{(n)} \mid y,n, z_{1:m}) = 
%\frac{\prod_{j=1}^n \prod_{x_{1:j-1} \in \{0, 1\}^{j-1}} y_{x_{1:j-1, 0}}^{N_{x_{1:j-1, 0}}} y_{x_{1:j-1, 1}}^{N_{x_{1:j-1, 1}}} r_{x_{1:j-1, 0}}^{\beta_{x_{1:j-1, 0}}-1}r_{x_{1:j-1, 1}}^{\beta_{x_{1:j-1, 1}}-1}}{\prod_{x_{1:n}\in \{0, 1\}^n} L_{x_{1:n}}^{N_{x_{1:n}}}} 
%\end{equation}
%where the constant of proportionality depends only on $(y^{(n)},z_{1:m})$.
Finally, 
%for any $n\in\mathbb N_0$ with $p_n>0$, 
conditioned on both $(R,N)$ and $Z_{1:m}=z_{1:m}$, we have that 
\begin{equation}\label{e:hh1}
\hspace{0.0cm}\mbox{$Y^{(N)}\mid R,N,Z_{1:m}=z_{1:m}\sim\mathrm{FBS}(\Gamma^{(N)})$}
\end{equation} 
and 
\begin{equation}\label{e:hh2}
\mbox{$Y^{(>N)}\mid R,N,Z_{1:m}=z_{1:m}\sim\mathrm{BS}(\alpha^{(>N)})$}
\end{equation}
 are independent, where $\Gamma^{(N)}$ depends only on $R$ through $R_0^{(N)}$.
\end{theorem} 

%\begin{remark}\label{rem:mcmc-gfpt2} For simulation from the posterior distribution of $(N,Y,R)$ as given by Theorem~\ref{t:3}, we propose a  Metropolis-Hastings within Gibbs sampler where in each iteration we alternate between updating
%\begin{enumerate}
%\item[(I)] $N\mid R,Z_{1:m}=z_{1:m}$,
%\item[(II)] $Y\mid R,N,Z_{1:m}=z_{1:m}$,
%\item[(III)] $R\mid Y,N,Z_{1:m}=z_{1:m}$.
%\end{enumerate}
%Steps (I) and (II) are done in the same way as in Remark~\ref{r:sim1}. Step (III) is done using partly the simple BS-distribution in \eqref{e:r_ge_n} and partly a Metropolis-Hastings algorithm with equilibrium density given by \eqref{e:post_r_2}.
%% , where e.g.\ the proposal distribution may be given by $\mathrm{BS}(\beta^{(n)})$ so that the Hastings ratio becomes
%% \[
%% 	y_{x_{1:j-1, 0}}^{n_{x_{1:j-1, 0}}(\tilde r_0^{(j)},z_{1:m})-n_{x_{1:j-1, 0}}(r_0^{(j)},z_{1:m})} 
%% y_{x_{1:j-1, 1}}^{n_{x_{1:j-1, 1}}(\tilde r_0^{(j)},z_{1:m})-n_{x_{1:j-1, 1}}r_0^{(j)},z_{1:m})}
%% \]
%% when $(r_0^{(j)},r_1^{(j)})$ corresponds to the current state and $(\tilde r_0^{(j)},\tilde r_1^{(j)})$ corresponds to the proposal.
%\end{remark}

\begin{remark}\label{rem:mcmc-gfpt2-collapsed}
%	Similarly to \eqref{e:post_n_1}, 
%It is possible to analytically marginalize out $Y$ from the posterior of $(N, Y, R)$, and then to 
We use a Metropolis-Hastings within Gibbs sampler which alternates between updating
	\begin{enumerate}
		\item[(I)] $N\mid R,Z_{1:m}=z_{1:m}$,
		\item[(II)] $R\mid N,Z_{1:m}=z_{1:m}$,
		\item[(III)] $Y\mid R,N,Z_{1:m}=z_{1:m}$.
	\end{enumerate}
Steps (I) and (III) are done in the same way as in Remark~\ref{r:sim1}, using \eqref{e:post_n_2} for step (I) and \eqref{e:hh1}--\eqref{e:hh2} for step (III).	In step (II), $R^{(>n)} \mid N = n, Z_{1:m} = z_{1:m}$ is simply distributed as in \eqref{e:r_ge_n} and is independent of $R^{(n)} \mid N = n, Z_{1:m} = z_{1:m}$, which follows the unnormalized density \eqref{e:post_r_2}, where we use a Metropolis-Hastings update. 

%has a density with respect to $\mu_n$ given by $p(r_0^{(n)}\mid n, z_{1:m})\propto 1$ when $n=0$ and for every $r_0^{(n)}\in(0,1)^{2^n-1}$ when $n>0$ by 
%	\[
%	p(r_0^{(n)} \mid n, z_{1:m})\propto
%	\frac{\prod_{j=1}^n  \prod_{x_{1:j-1} \in \{0, 1\}^{j-1}} \mathrm b\big(\Gamma_{x_{1:j-1},0}, \Gamma_{x_{1:j-1},1}\big) r_{x_{1:j-1, 0}}^{\beta_{x_{1:j-1, 0}}-1}r_{x_{1:j-1, 1}}^{\beta_{x_{1:j-1, 1}}-1}}{\prod_{x_{1:n}\in \{0, 1\}^n} L_{x_{1:n}}^{N_{x_{1:n}}}},
%	\]
%	where the constant of proportionality depends only on $(y^{(n)}, z_{1:m})$.

%we use that $R^{(>n)}$ and $R^{(n)}$ are conditionally independent with distributions given by \eqref{e:r_ge_n} and \eqref{e:post_r_2}, respectively.
%	
%	Here, $R^{(>n)} \mid N = n, Z_{1:m} = z_{1:m}$ is still distributed as \eqref{e:r_ge_n} and is independent of $R^{(n)} \mid N = n, Z_{1:m} = z_{1:m}$ which has a density with respect to $\mu_n$ given by $p(r_0^{(n)}\mid y,n, z_{1:m})\propto 1$ when $n=0$ and for every $r_0^{(n)}\in(0,1)^{2^n-1}$ when $n\in\mathbb N$ by 
%	\[
%	p(r_0^{(n)} \mid y,n, z_{1:m})\propto
%	\frac{\prod_{j=1}^n  \prod_{x_{1:j-1} \in \{0, 1\}^{j-1}} \mathrm b\big(\Gamma_{x_{1:j-1},0}, \Gamma_{x_{1:j-1},1}\big) r_{x_{1:j-1, 0}}^{\beta_{x_{1:j-1, 0}}-1}r_{x_{1:j-1, 1}}^{\beta_{x_{1:j-1, 1}}-1}}{\prod_{x_{1:n}\in \{0, 1\}^n} L_{x_{1:n}}^{N_{x_{1:n}}}},
%	\]
%	where the constant of proportionality depends only on $(y^{(n)}, z_{1:m})$.
\end{remark}

\subsubsection{Posterior simulation under a MBPT prior}

{\color{black} 
Posterior computation is straightforward under a MBPT prior as given in Section~\ref{sec:mbpt}. Indeed,  $(Y, N)$ and $\omega$ are independent also a posteriori, with the posterior of $(Y, N) \mid Z_{1:m}$ given by a slight modification of Theorem \ref{t:2}. Moreover, $\omega \mid Z_{1:m} \sim \mathrm{Dir}(\eta^p)$ where $\eta^p_k = \eta_k + \sum_{i=1}^m I[M_i = k]$ for $k \in T$. }

\subsection{Consistency for the posterior of $N$}\label{sec:cons_n}

In this section, 
 we consider the posterior distribution for the case \eqref{e:sample1} where a GFPT1 prior has been specified but assume 
$Z_1, \ldots, Z_m$ are i.i.d.\ according to a PDF $f^*$, which we refer to as the true distribution of $Z_{1:m}$. We focus on the asymptotic behaviour of the latent variable $N$ when $m\rightarrow\infty$. At the end of Section~\ref{s:bounded-or-not} (Remark~\ref{r:13}) we % Remark~\ref{r:13} we %apply Theorem~\ref{t:inf_n} to 
motivate the use of GFPT2 priors.

We assume that $f^*$ {\color{black}is} Lebesgue almost everywhere LSC so that Theorem~\ref{t:1} applies. {\color{black} We use the notation $N$ for both the case of the GFPT1 model, with $N$ as in Definition~\ref{def:5}, and the case of the true model, with $N$ given by the coupling construction in Theorem~\ref{t:1} where $f$ is replaced by $f^*$ when considering the term $c_{x_{1:n}}$ in \eqref{e:res1}. This should not cause any confusion, since we
 denote probabilities calculated with respect to the true distribution by $\mathrm P^*$, while $\mathrm P$ still refers to the distribution under the GFPT1 model. We also use the notation $\mathrm P^*$ when considering the distribution of the stochastic process $Z_1,Z_2,\ldots$ under the true model.} 
 
 Recall that $p_n=\mathrm P(N=n)$ is the PMF of $N$ %under the GFPT1 model 
 and  $p(n\mid Z_{1:m})=\mathrm P(N=n\mid Z_{1:m})$ is the PMF under the posterior distribution of $N$, cf.\ \eqref{e:post_n_1}.
Let $\mathcal N$ be the support of the prior for $N$.
For the following Sections~\ref{s:bounded-or-not}--\ref{s:round-off}, assume that for some $0 \le n_{\min} \le n_{\max}\le\infty$ and some $K>0$,
\begin{equation}\label{e:assumption}
                \text{$\mathcal N = \{n_{\min}, \ldots, n_{\max}\}$ where if $n_{\max}=\infty$ then $ 0 < p_n < K p_{n+1}$ for all $n\in\mathcal N$}.
\end{equation}
% \begin{equation}\label{e:assumption}
% 	\text{either $\mathcal N = \{n_{\min}, \ldots, n_{\max}\}$ or $\mathcal N = \mathbb N_0$ with $ 0 < p_n < K p_{n+1}$ for all $n\in\mathbb N_0$}
% \end{equation}
% and some numbers $0 \le n_{\min} \le n_{\max}<\infty$ and $K>0$.
The last inequality in \eqref{e:assumption} is akin of local stability condition, cf.\ \cite{JM}.

\subsubsection{When $N$ is bounded or not}\label{s:bounded-or-not}

The following theorems consider consistency of the posterior distribution of $N$ under the true distribution and depending on whether $N$ is bounded (Theorem~\ref{t:consistency}) or not (Theorem~\ref{t:inf_n}). % and when the data are observed with a certain precision (Theorem~\ref{thm:cons2}).

\begin{theorem}\label{t:consistency}
Suppose that  under the true distribution $n^*\in\mathbb N_0$ digits are sufficient in the sense that
\begin{equation}\label{e:casea}
\mathrm P^*(N< n^*)<\mathrm P^*(N\le n^*)=1,
\end{equation}
cf.\ Remark~\ref{r:suff}.
In addition, assume $n^*\in\mathcal N$ 
and the $\alpha_{x_{1:n}}$ with $n\in\mathbb N$ and $x_{1:n}\in\{0,1\}^n$ are  bounded. Then 
\begin{equation}\label{e:cons}
	\mathrm P^*\left(\lim_{m \rightarrow \infty} p(n^* \mid Z_{1:m}) = 1\right)=1
\end{equation}
where $p(n^* \mid Z_{1:m})$ is given by \eqref{e:post_n_1}.
%almost surely with respect to the true distribution of $Z_{1:m}$.
\end{theorem}

\begin{remark}  By Theorem~\ref{t:1}, \eqref{e:casea} is equivalent to that $f^*$ is constant over all subintervals of the NBP $I$ after level $n^*$, that is, for every $x_{1:n^*}\in\{0,1\}^{n^*}$, $f^*$ is constant on $I_{x_{1:n^*}}$ and if $n^*>0$ then there is some
$x_{1:(n^*-1)}\in\{0,1\}^{n^*-1}$ so that $f^*$ is not constant on $I_{x_{1:(n^*-1)}}$. We interpret \eqref{e:cons} as it is asymptotic consistent to estimate the sufficient number of digits by the posterior distribution of $N$. % given by \eqref{e:post_n_1}.
\end{remark}

 \begin{theorem}\label{t:inf_n}
	  If $\mathrm P^*(N \le n) < 1$ for all $n \in \mathbb N_0$, then for any $n\in\mathbb N_0$,
	  \begin{equation}\label{e:inf_n}
		  \mathrm P^* \left( \lim_{m \rightarrow \infty} p(n \mid Z_{1:m}) > 0 \right){\color{black} = 0.}
	  \end{equation}
  \end{theorem}

  \begin{remark}\label{r:13}
The  condition in Theorem~\ref{t:inf_n} states that under the true distribution, $N$ is unbounded. Hence,  
  \eqref{e:inf_n} establishes asymptotic consistency, since a posteriori $N$ can be as large as it should be as $m$ increases.
  
	The class of PDFs $f^*$ for which Theorem~\ref{t:inf_n} applies encompasses piecewise constant PDFs over intervals $J_1,\ldots,J_K$ that do not agree with the NBP $I$. That is, $f^*(x) = f_k$ if $x \in J_k$, and there exists at least one $J_k$ such that $J_k \neq I_{x_{1:n}}$ for every $n \in \mathbb N_0$ and $x_{1:n} \in \{0, 1\}^n$.
	
	%As a prominent 
	For example, consider $I$ to be the standard diadic partitions of $[0, 1)$ and $f^*$ to be the PDF of a $\mathrm{Unif}[0,0.2)$ random variable. 
	 Then the condition in Theorem~\ref{t:inf_n} is satisfied, so under the GFPT1 prior and  as $m\to\infty$, a posteriori $N$ is unbounded. On the other hand, we
	show numerically  in Section~\ref{sec:simu2} that under {\color{black} various} GFPT2 priors,
 	a posteriori $N$ is bounded (effectively always smaller than 4) even if $m$ is large. Thus, the use of GFPT2 models that adaptively ``learn the nested partitions'' is appealing.
\end{remark}

\subsubsection{Round-off errors}\label{s:round-off}

We now analyse the impact on posterior inference caused by two types of round-off errors in the data. 
{\color{black}Specifically, we assume that data are observed with a precision of $\bar n\in\mathbb N_0$ digits and consider two scenarios. In the first one, we ignore the round-off and assume to have observed the data perfectly: We show in Theorem~\ref{thm:cons2} that this leads to an inconsistent posterior for $N$ as $m\rightarrow\infty$ even if $n^* \leq \bar n$.
In the second scenario, we explicitly account for the round-off and acknowledge that data contain no information beyond digit $\bar n$. That is, we assume that we do not actually observe an i.i.d. sample from $f^*$ but rather from $\bar f^*$, which is a piecewise constant approximation of $f^*$ at level $\bar n$ of the nested binary partition. 
In such a case, if $n^* \leq \bar n$, we obtain posterior consistency for $N$ in Corollary~\ref{cor:roundoff}.
}

\begin{theorem}\label{thm:cons2}
	Assume $\mathcal N=\mathbb N_0$ and the data are observed with a precision of $\bar n\in\mathbb N_0$ digits such that for $i=1,\ldots,m$, if $\bar n>0$ then for some $(x^i_1,x^i_2,\ldots, x^i_{\bar n})\in\{0,1\}^{\bar n}$,
	\begin{equation}\label{e:caseb}
	z_i = .x^i_1x^i_2\ldots x^i_{\bar n}00\ldots,
	\end{equation}
	and if $\bar n=0$ then $z_i =0$. %\eqref{e:caseb} holds.
	%For $i=1,\ldots,m$, l
	Correspondingly, let $\tilde Z_i$ be given by the first $\bar n$ digits of $Z_i$ and by 0's for the remaining digits.
	 Then, for every $n\in\mathbb N_0$,
	\[
		\mathrm P^*\left( \lim_{m\rightarrow\infty}p(n\mid \tilde Z_{1:m})=0 \right)=1.
	\]
	% If, instead, there exists $n_{\max}\in\mathbb N_0$ so that $p_n = 0$ for $ n=n_{\max}+1,n_{\max}+2,\ldots$ and $p_{n_{\max}} >0 $, in which case
	%  \[
	%  \mathrm P^*\left(\lim_{m\rightarrow\infty}p(n_{\max} \mid Z_{1:m})=1
	%            \right)=1.
	% %      N \mid Z_{1:m} \rightarrow n_{\max}.
	% \]
	\end{theorem}

% Theorem~\ref{thm:cons2} shows that the round-off error introduced by using 0's after the precision $\bar n$ leads to an inconsistent posterior distribution of $N$ as $m\to\infty$.
% This is still
{\color{black}The conclusion of Theorem~\ref{thm:cons2} is unchanged} if we use another fixed digits rule than in \eqref{e:caseb}, e.g., if we replace the 0's by 1's after level $\bar n$. 

{\color{black}We now consider the scenario where we explicitly account for the round-off error.}
The following corollary follows immediately from Theorems~\ref{t:1} and \ref{t:consistency}. For $i=1,\ldots,m$, let $Z_i=.X^i_1X^i_2\ldots$.

\begin{corollary}\label{cor:roundoff}
	%Instead of \eqref{e:caseb} 
	Assume the round-off error is treated as being uniformly distributed after level $\bar n\in\mathbb N_0$, that is, under the GFPT1 prior we 
have $N\le\bar n$ and  		
	\begin{equation}\label{eq:casec}
	.X^i_{\bar n+1}  X^i_{\bar n+2} \ldots \iid \mathrm{Unif}[0,1)\quad\mbox{for }i=1,\ldots,m.	
%					z_i = .x^i_1x^i_2\ldots x^i_{\bar n} x^i_{\bar n+1}  x^i_{\bar n+2} \ldots,\quad (x^i_1,x^i_2,\ldots, x^i_{\bar n})\in\{0,1\}^{\bar n},\quad i=1,\ldots,m,
	\end{equation}
	%with $.x^i_{\bar n+1}  x^i_{\bar n+2} \ldots \iid \mathrm{Unif}[0,1)$.
	 Further, assume that under the true distribution, in the sense of \eqref{e:casea}, $n^* \le \bar n$ digits are sufficient where $n^*\in\mathcal N$,  and that all the $\alpha_{x_{1:n}}$ are  bounded. Then
	\begin{equation}\label{e:ffff}
		\mathrm P^*\left(\lim_{m \rightarrow \infty} p(n^* \mid Z_{1:m}) = 1\right)=1.
	\end{equation}
\end{corollary}

\begin{remark}\label{r:support}
Under the GFPT1 prior, $N\le\bar n$ implies \eqref{eq:casec}, cf.\ \eqref{e:sdf}.
Equation \eqref{e:ffff} shows it is asymptotic consistent to estimate the sufficient number of digits by the posterior distribution of $N$. 

Under the true distribution $f^*$, if $n^*$ digits are sufficient, then $n^*$ is the largest possible value of $N$. Ideally we should therefore have that $n^*\in\mathcal N\subseteq\{0, \ldots, \bar n\}$ but in practice we do not know the value of $n^*$ (provided it exists). This suggests to let $\mathcal N=\{0, \ldots, \bar n\}$. This is also an intuitive suggestion as the data contains no information beyond level $n^*$ of the nested binary partitions.

{\color{black}On the other hand, if under the true distribution $n^* > \bar n$, then the posterior of $N$ will concentrate on $\bar n$. This is intuitive as assuming \eqref{eq:casec} means that we do not observe an i.i.d.\ sample from $f^*$ but from a density $\bar f^*$ which is piecewise constant beyond level $\bar n$ of the nested binary partition.}
\end{remark}

% \begin{remark}
% 	Assuming the data are observed with a precision of $n^*$ digits, in light of Theorem~\ref{thm:cons2}, we suggest that a prior for $N$ should have its support contained in $\{0, 1, \ldots, n^*\}$. 
% 	This is also an intuitive suggestion as the data contains no information beyond level $n^*$ of the nested binary partitions. \MBtext{See also Remark~\ref{r:priorN}.}
% \end{remark}
  
\subsection{Consistency for the posterior of $f$}\label{s:consistency_f}

In this section we assume the same setting as in Section \ref{sec:cons_n} and consider the random GFPT1 PDF $f(\cdot\mid Y^{(N)})$ and the random GFPT2 PDF $f(\cdot\mid Y^{(N)},R^{(N)})$ given by \eqref{e:finite_pt_N} and \eqref{e:finite_pt_N2}, respectively. 
For $i=1,2$, using the short hand notation $f=f(\cdot\mid Y^{(N)})$ if $i=1$ and $f=f(\cdot\mid Y^{(N)},R^{(N)})$ if $i=2$,  
recall that a posteriori  $f \mid Z_{1:m} \sim \Pi_i(\cdot \mid Z_{1:m})$, cf.\ \eqref{e:postrandomPDF} and \eqref{e:postrandomPDF2}.
  As the sample size increases,
we will show that $\Pi_i(\cdot \mid Z_{1:m})$ concentrates on infinitesimal small neighbourhoods of $f^*$ with respect to the Kolmogorov-Smirnov distance $d_{\mathrm{KS}}$ and the Hellinger distance $d_{\mathrm{H}}$:
For two PDFs $f_1$ and $f_2$ on $[0, 1)$, 
if $F_i(x) = \int_0^x f_i(x)\, \mathrm d x$ for $i=1, 2$ and $0\le x<1$, recall that
\[
	d_{\mathrm{KS}}(f_1, f_2) = \sup_{x \in [0, 1)} | F_1(x) - F_2(x) |
\]
and 
\[
	d_{\mathrm{H}}(f_1, f_2) = \frac{1}{2} \left(\int_0^1 (\sqrt{f_1(x)} - \sqrt{f_2(x)})^2 \mathrm d x\right)^{1/2}.
\] 

%The following theorem gives sufficient conditions for the posterior of $f$ to concentrate on vanishing neighbourhoods of $f^*$ in the weak topology. 

\begin{theorem}\label{thm:weak_consistency}
	Suppose that $\int_0^1 f^*(x) \log f^*(x)\, \mathrm d x < \infty$, $p_n > 0$ for all $n\in\mathbb N_0$, $\alpha_{x_{1:n}} > 0$ and, in case of GFPT2 priors, $\beta_{x_{1:n}} > 0$  whenever $n \in \mathbb N$ and $x_{1:n} \in \{0, 1\}^n$.
 Then, for any $\varepsilon > 0$ and $i=1,2$,
	\[
		\mathrm{P^*} \left( \lim_{m \rightarrow \infty} \Pi_i (d_{\mathrm{KS}}(f, f^*) < \varepsilon \mid Z_{1:m}) = 1 \right) = 1
	\]
where $f=f(\cdot\mid Y^{(N)})$ if $i=1$ and $f=f(\cdot\mid Y^{(N)},R^{(N)})$ if $i=2$.
\end{theorem}

%The proof of Theorem \ref{thm:weak_consistency} is based on the general theory developed in \cite{gvdv} and is deferred to Appendix~\ref{app:cons_f}.

Theorem~\ref{thm:weak_consistency} gives sufficient conditions for the posterior of $f$ to concentrate on vanishing neighbourhoods of $f^*$ in the weak topology, since convergence with respect to $d_{\mathrm{KS}}$ is stronger than weak covergence of probability measures on Euclidean spaces \citep[see e.g.\ Theorem~6 in][]{gibbs}. %$d_{\mathrm{KS}}$ metrizes the weak covergence of probability measures on Euclidean spaces 
 However,
as discussed in \cite{gvdv}, consistency under the weak topology is often considered ``too weak'' for density estimation purposes. Instead the next theorem gives sufficient conditions for consistency with respect to the Hellinger distance.
% given by
%\[
%	d_{\mathrm{H}}(f_1, f_2) = \frac{1}{2} \left(\int_0^1 (\sqrt{f_1(x)} - \sqrt{f_2(x)})^2 \mathrm d x\right)^{1/2}
%\] 
%for PDFs $f_1$ and $f_2$ on $[0, 1)$. 
Convergence with respect to $d_{\mathrm{H}}$ is stronger than weak convergence, and 
the topology induced by the Hellinger distance is equivalent to the one induced by the total variation norm \citep[see e.g.][]{kraft}.

\begin{theorem}\label{thm:hell_consistency}
	In addition to the assumptions of Theorem \ref{thm:weak_consistency}, %the prior for $N$ is such that $0 < 
	suppose $p_n < C e^{-cn}$ for some constants $C>0$ and $c > 0$ and for all $n\in\mathbb N$. Then, for any $\varepsilon > 0$,
	\[
		\mathrm{P^*} \left( \lim_{m \rightarrow \infty} \Pi_1 \left( d_{\mathrm H}(f(\cdot\mid Y^{(N)}), f^*) < \varepsilon  \mid Z_{1:m}\right) = 1 \right) = 1.
	\]
\end{theorem}

%See Appendix~\ref{app:cons_f} for the proof. %Note that Theorem~\ref{thm:hell_consistency} deals only with GFPT1 models 

\begin{remark}
	 For instance, the condition $0 < p_n < C e^{-ck}$ is satisfied by the Poisson distribution. 
	By contrast, posterior consistency for standard P{\'o}lya tree priors was established in \cite{Barron99} \citep[see also Theorem 7.16 in][]{gvdv} under the assumptions that the parameters $\alpha_{x_{1:n}}$ grow exponentially fast with $n$. 
	As noted in \cite{gvdv} this assumption is too strong, since practitioners usually follow \cite{lavine} by setting $\alpha_{x_{1:n}} = \alpha_0 n^2$. See also \cite{Castillo17} and \cite{giordano} for further developments.
\end{remark}

\section{Numerical illustrations}\label{s:num}

This section investigates several aspects of our models via numerical simulations as follows. Section \ref{sec:s1} assesses the impact of assuming $N$ random on density estimation by comparing the GFPT1 and standard (finite) P{\'o}lya tree priors, focusing on accuracy of the density estimates as well as prior and posterior variability. 
Section \ref{sec:simu2} compares the GFPT1 and GFPT2 priors in terms of density estimation accuracy and the posterior distribution of $N$. Moreover, we compare our models with the Optional P\'olya tree (OPT) model of \cite{WongMa} and the APT model of \cite{Ma17}.
Section \ref{sec:comparison_real_data} illustrates the advantages of the MBPT model on real datasets spanning multiple orders of magnitude.
Finally, Appendix F in Supplemental material % \ref{app:simulations} 
discusses the case of data that lies outside the unit interval, and how the use of a bijection to transform data living in $\mathbb R$ or $\mathbb R_+$ to $[0,1)$ affects density estimation, as well as the role of $N$ as a \emph{smoothing} parameter.
In particular, we observe that care must be taken when applying such a transformation to avoid that most of the datapoints get mapped near the boundaries of the interval $[0, 1)$: A simple scaling of the data by their empirical standard deviation before applying the transformation yields accurate posterior density estimates. Moreover, the importance of $N$ is unchanged, as it allows to adapt to the smoothness of the data.

Prior elicitation follows the discussion in Section \ref{s:elicitation}. Unless otherwise specified, a priori $N$  follows the truncation of a Poisson distribution to $\{0,1, \ldots, 20\}$ where the Poisson distribution has mean 5.

We need the following notation. For two probability densities $f$ and $g$, $\mathrm{TV}(f, g) = \frac{1}{2} \int |f(x) - g(x)|\, \mathrm d x$ denotes their total variation distance. For a curve $g: [a, b] \rightarrow \mathbb R$, denote its length by $L_g$ and let $W_g = L_g - |b - a|$. 
We use $W_g$ to compare the ``wigglyness'' of two PDFs: Let $f_0$ be a reference PDF, $\hat f_j$, $j=1, 2$ be two estimators of $f_0$, and $W_j = W_{\hat f_j - f_0}$ for $j=1,2$. Then we say that $\hat f_1$ is less wiggly than $\hat f_2$ if $W_1 < W_2$. Clearly, for any $\hat f$, $W_{\hat f - f_0} \geq 0$  and we have $W_{\hat f - f_0} = 0$ if and only if $\hat f = f_0$ (Lebesgue almost everywhere).

\subsection{Posterior results when using GFPT1 and P{\'o}lya tree priors}\label{sec:s1}

\begin{figure}[t]
	\centering
	\includegraphics[width=\linewidth]{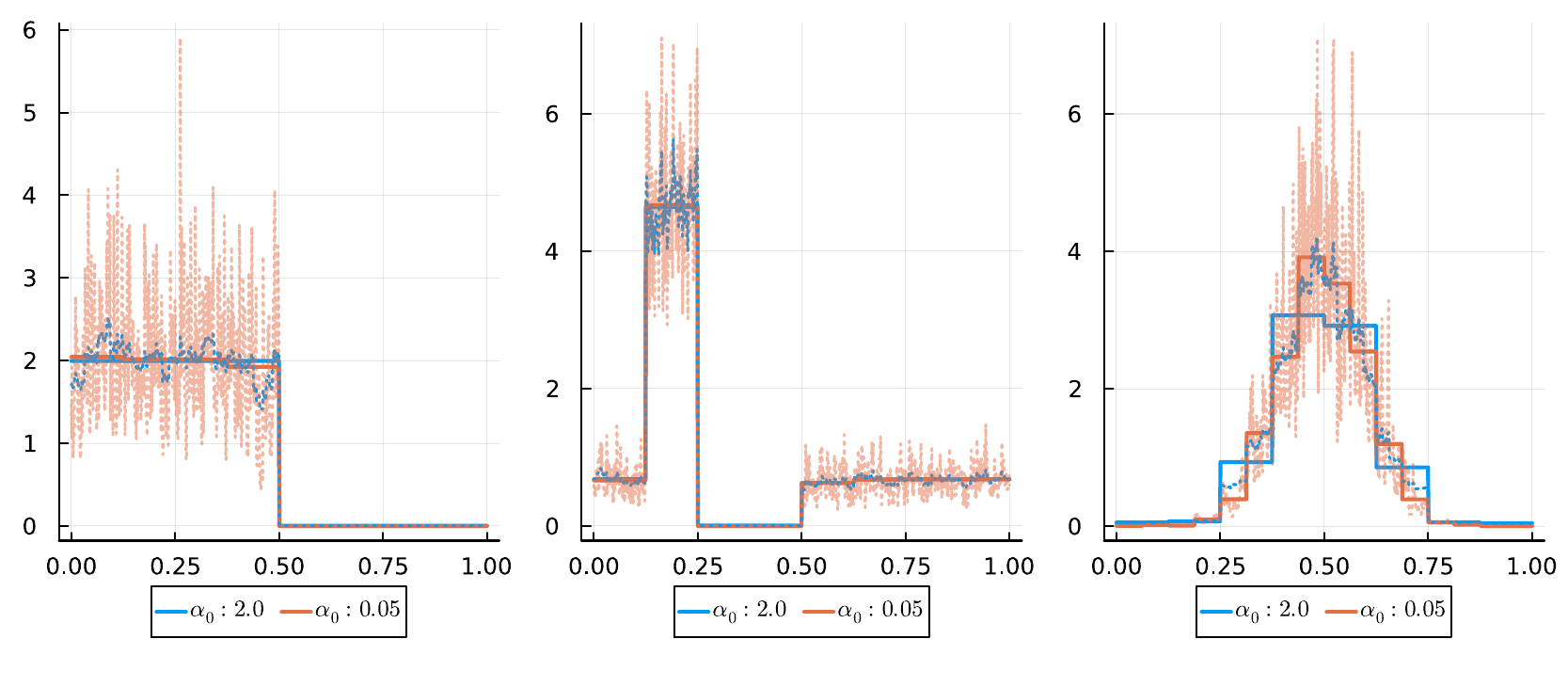}
	\caption{Posterior results for first simulation study in Section \ref{sec:s1}: Density estimates under GFPT1 (solid lines) and PT priors (dotted lines) when $\alpha_0 = 2$ (blue lines) and $\alpha_0=0.05$ (orange lines) where 
	the three columns correspond to the data examples DG1--DG3, respectively.}
	\label{fig:example1}
\end{figure}

We consider three data generating processes:
\begin{align}
	Z_1,\ldots,Z_m &\iid  \mathrm{Unif}(0, 0.5) \label{DG1} \tag{DG1} \\
	Z_1,\ldots,Z_m &\iid  \frac{1}{6} \mathrm{Unif}(0, 0.25) + \frac{1}{2} \mathrm{Unif}(0.125, 0.25) + \frac{1}{3} \mathrm{Unif}(0.5, 1) \label{DG2} \tag{DG2} \\
	Z_1,\ldots,Z_m &\iid \mathcal N(0.5, 0.1)|_{[0, 1)} \label{DG3} \tag{DG3}
\end{align}
where $\mathcal N(\mu, \sigma^2)|_{[0, 1)}$ denotes the truncation on the unit interval of the normal distribution with mean $\mu$ and variance $\sigma^2$.
We  simulated %begin by simulating 
$m=1000$ observations from each data generating process and 
 fitted
them using either a $\mathrm{GFPT1}(\alpha,I)$ or a $\mathrm{PT}(\alpha,I)$ prior where $\alpha$ is given by \eqref{e:alpha-first} with $\alpha_0 = 2$ or $\alpha_0 = 0.05$ and where $I$ is the sequence of standard diadic partitions. Thus, in our first simulation study, we consider 12 cases corresponding to the four different models for each of the three datasets.

%{\color{blue} (Comment to myself: Possibly rewrite the hole text when Mario has clarified the first red text below.)}
Figure~\ref{fig:example1} shows the density estimates  (as given in Remark~\ref{rem:point_estimator}) for the 12 cases. For all datasets, the PT prior yields extremely wiggly density estimates when $\alpha_0 = 0.05$ (dotted orange line), and these get only partially less wiggly for $\alpha_0 = 2$ (dotted blue line). On the other hand, for datasets DG1 and DG2 (the two first columns), the GFPT1 prior yields precise density estimates for both values of $\alpha_0$ (the solid blue and orange lines). However, for DG3 and $\alpha_0 = 2.0$, the density estimate under the GFPT1 prior is not very precise (last column, solid blue line).  This happens because  the posterior of $N$ concentrates on small values leading to a coarse approximation. When $\alpha_0 = 0.05$ instead, the posterior of $N$ concentrates to higher values leading to a more accurate density estimate (last column, solid orange line).

Next, we focus on the impact that assuming $N$ random has on prior and posterior variability. For the GFPT1 prior, we assume that a priori $N$  follows the truncation of a Poisson distribution to $\{0,1, \ldots, 20\}$ where the Poisson distribution has mean 10. 
As alternative, we consider a finite PT model which corresponds to setting $N = \delta_{10}$.
For both models, $\alpha$ is given by \eqref{e:alpha-first} for $\alpha_0 \in \{0.05, 0.1, 2.0, 10.0\}$.
Figure~\ref{fig:prior_variance} shows global credible bands for the random PDF under both priors: A priori, the introduction of a prior on $N$ does not seem to make a huge difference in terms of variability.
However, the posterior under the two models is strikingly different, as shown in Figure~\ref{fig:post_variance}: For all choices of the parameter $\alpha_0$ and under all three data generating processes (DG1)--(DG3), the posterior under the GFPT1 is much more concentrated, highlighting the benefits of assuming a prior for $N$ as opposed to fixing it to a large value. Moreover, Figure~\ref{fig:post_variance} (bottom row) shows the effect that $\alpha_0$ has on the posterior of $N$: Large values of $\alpha_0$ shrink $N$ to smaller values a posteriori, resulting in coarser density estimates.
In summary, under a GFPT1 prior, we find that specifying a small value for $\alpha_0$ enables greater adaptability via the prior on $N$ that regularize the density estimate.

\begin{figure}
	\centering
	\includegraphics[width=\linewidth]{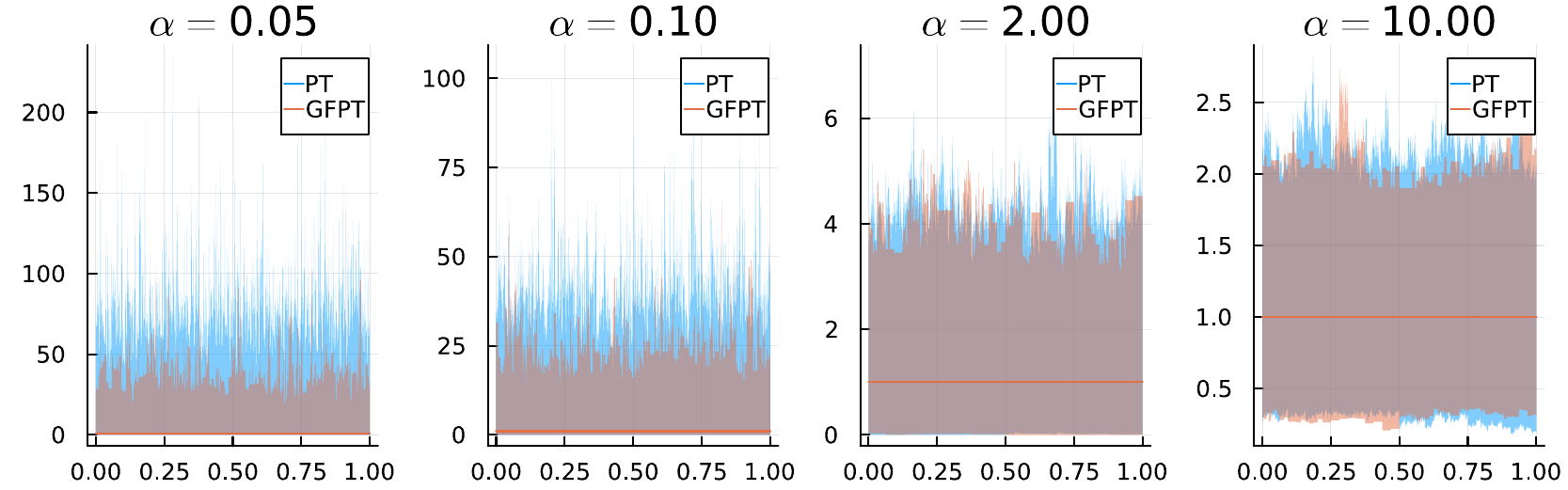}
	\caption{95\% global credible bands (shared area) and mean (solid line) of the random PDF under a GFPT1 prior or a finite PT prior when $\alpha$ is given by \eqref{e:alpha-first} for $\alpha_0 \in \{0.05, 0.1, 2.0, 10.0\}$. For the GFPT1 prior, a priori $N$  follows the truncation of a Poisson distribution to $\{0,1, \ldots, 20\}$ where the Poisson distribution has mean 10. For the finite PT prior, the NBP $I$ has depth 10.}
	\label{fig:prior_variance}
\end{figure}

\begin{figure}
	\centering
	\includegraphics[width=\linewidth]{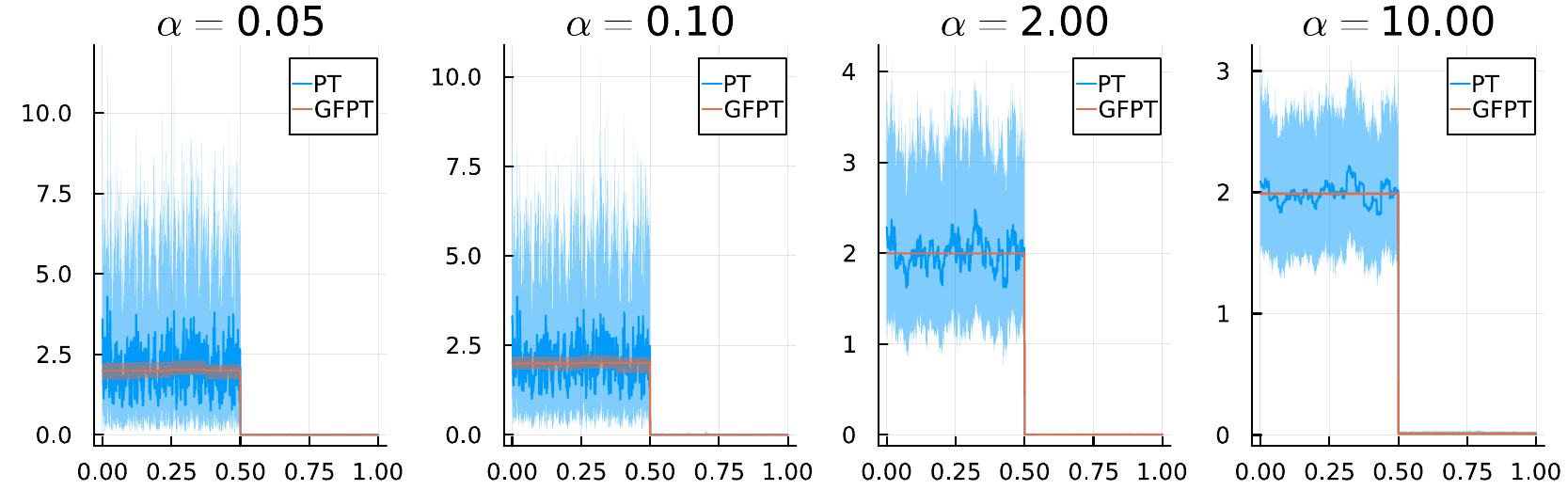}
	\includegraphics[width=\linewidth]{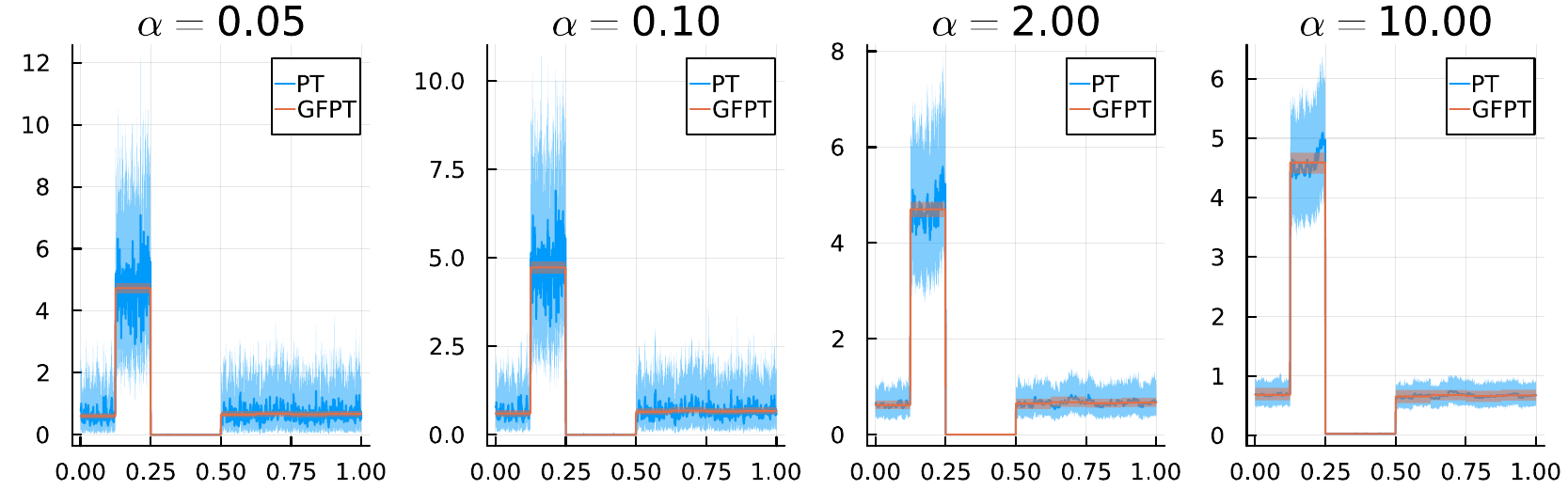}
	\includegraphics[width=\linewidth]{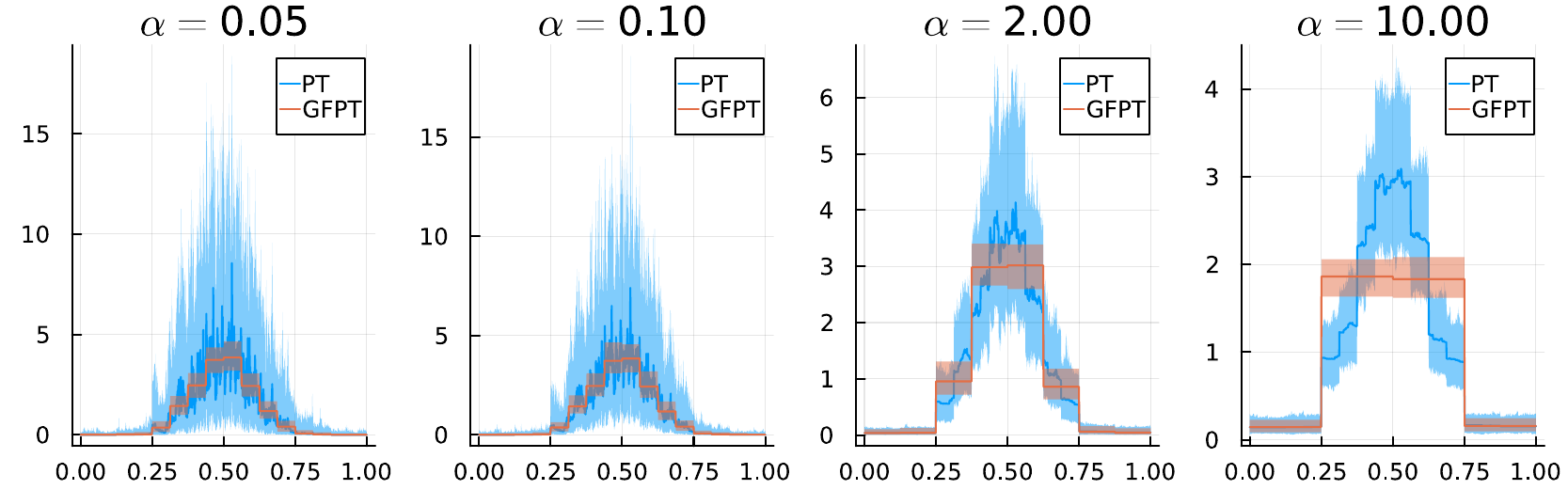}
	\caption{95\% global credible bands (shared area) and mean (solid line)
	of the posterior of the random PDF under the same specifications of Figure \ref{fig:prior_variance}, when data are generated by (DG1)--(DG3) (different rows, from top to bottom)}
	\label{fig:post_variance}
\end{figure}

In our second simulation study, we  confirmed such insights by generating $100$ independent datasets from each data generating process with $$m \in \{50, 100, 1000, 5000, 10000\}$$  and where 
each dataset was fitted using PT and GFPT1 priors with $$\alpha_0 \in \{0.05, 0.1, 2.0, 10.0\}.$$ Figure~\ref{fig:simu1_tv} show the mean across the 100 independent datasets of 
\begin{enumerate}
\item[(i)] the total variation distance distance between the density of the true data generating process  (for short the true density) and the density estimate (cf.\ Remark~\ref{rem:point_estimator}),
\item[(ii)] the measure of wigglyness of the difference between the true and estimated densities,
\item[(iii)] the a posteriori expectation of $N$,
\end{enumerate}
respectively.
As expected Figure~\ref{fig:simu1_tv} 
show that larger values of $\alpha_0$ correspond to smaller values of $N$ (a posteriori) and vice versa. This is particularly evident for datasets DG1 and DG3. In general, we see that GFPT1 priors outperform PT priors in terms of the quality of the density estimate. 
Moreover, larger $\alpha_0$ clearly yield less wiggly functions when using PT priors, while this dependence is not so clear when using GFPT1 priors.
In conclusion, we find that $\alpha_0 = 0.1$ provides the best trade off between wigglyness and accuracy of the density estimate, and we suggest using it as default value for GFPT1 priors.

\begin{figure}[ht!]
	\centering
	\includegraphics[width=\linewidth]{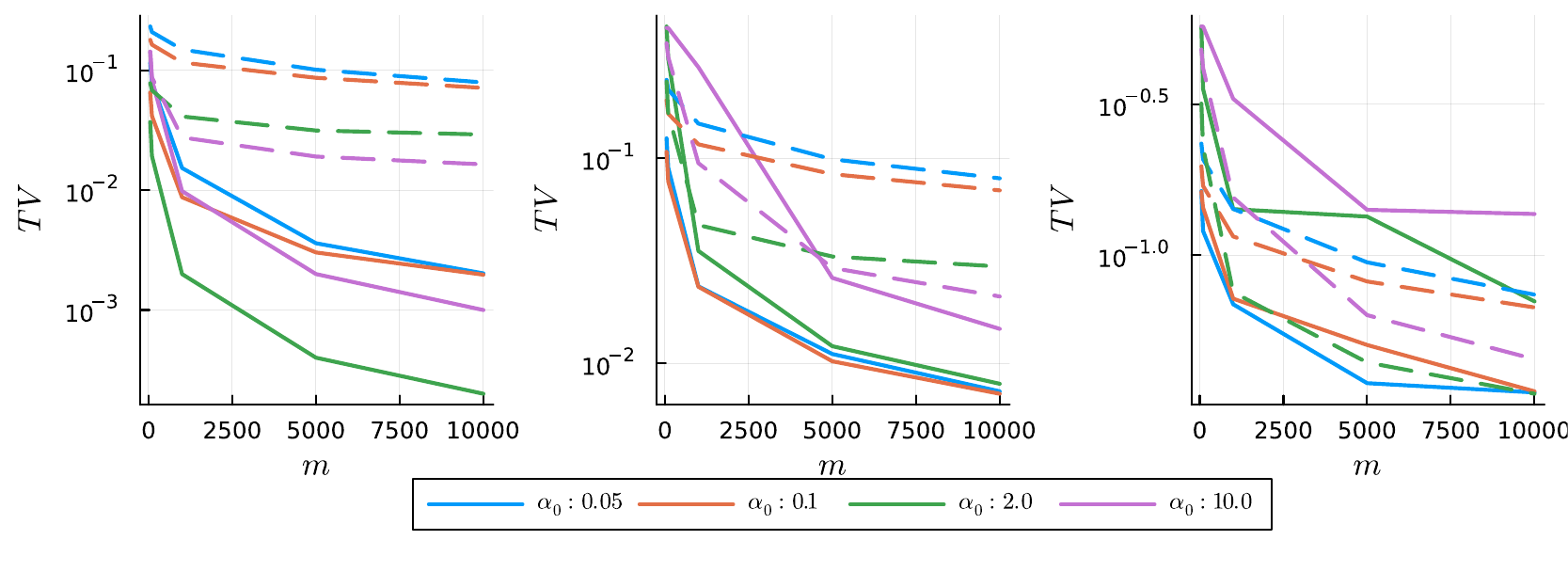}
	\includegraphics[width=\linewidth]{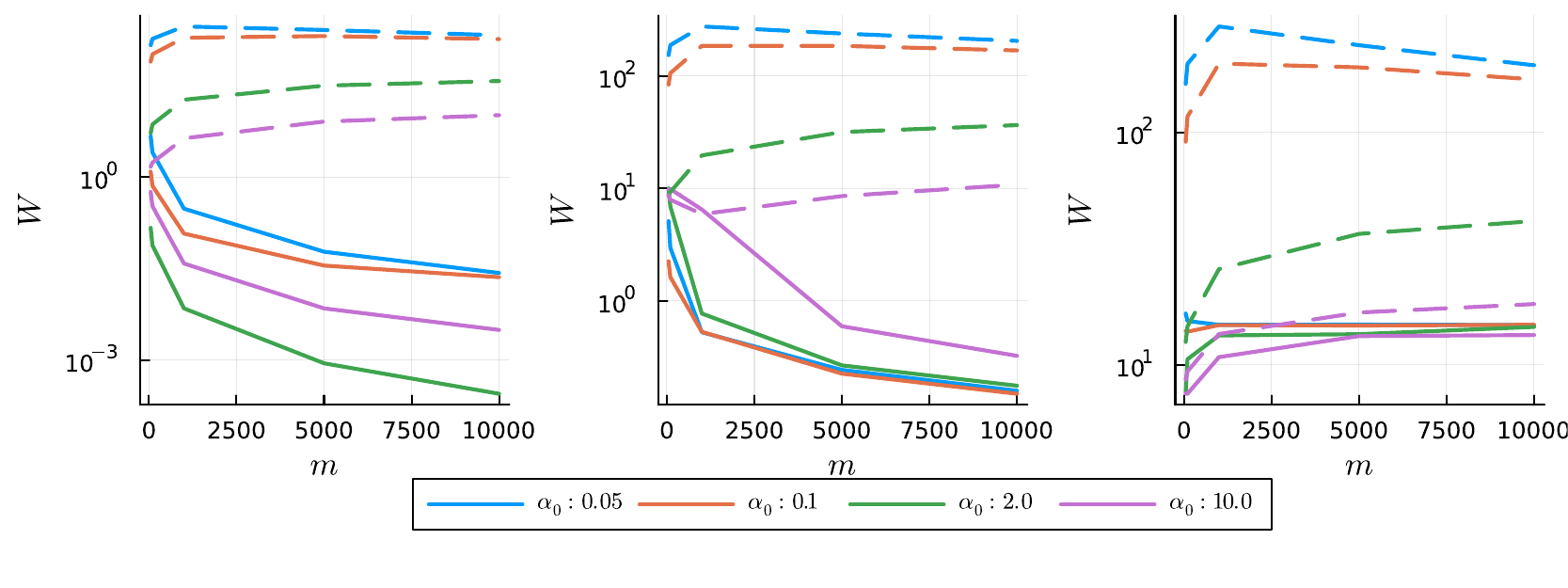}
	\includegraphics[width=\linewidth]{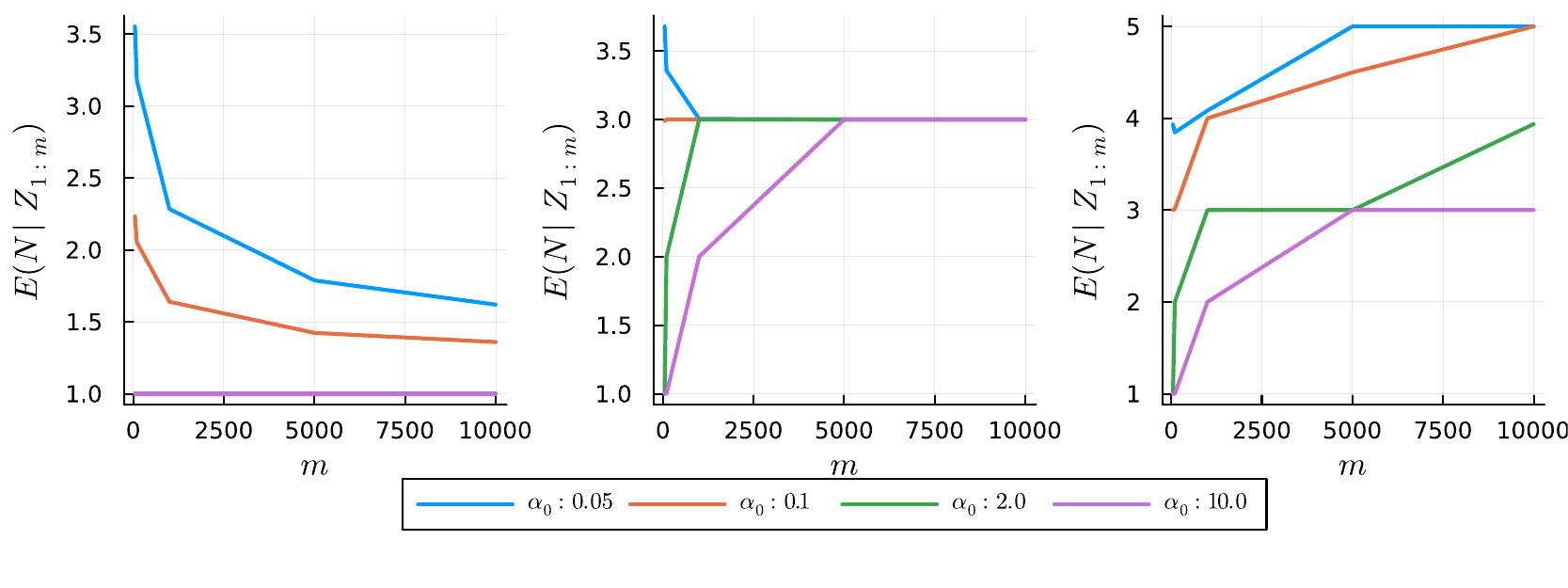}
	\caption{Posterior results for the second simulation study in Section \ref{sec:s1}. {From top to bottom:} Total variation (TV) distance between true and estimated densities, {$W$-function for the difference between true and estimated densities, and posterior expectation of $N$} as a function of the sample size $m$.
	The three columns correspond to the data examples DG1--DG3, respectively. Figures are averaged over 100 independent replicates. Solid and dashed lines refer to the use of GFPT1 and PT priors, respectively. Different colors refer to different values of $\alpha_0$.}
	\label{fig:simu1_tv}
\end{figure}

\subsection{Posterior results when using GFPT1 and GFPT2 priors}\label{sec:simu2}

This section investigates through simulation studies if there is any gain in extending GFPT1 models to GFPT2 models.
Moreover, we also compare our models with the Optional P\'olya tree (OPT) model of \cite{WongMa} and the APT model of \cite{Ma17}, implemented in the \texttt{R} package \texttt{PTT}.
For this we consider six different data generating processes, namely DG1--DG3 as in Section~\ref{sec:s1} and 
\begin{align}
	Z_1,\ldots,Z_m &\iid  \mathrm{Unif}[0, 0.2) \label{DG4} \tag{DG4} \\
	Z_1,\ldots,Z_m  &\iid  \frac{1}{2} \mathrm{Unif}[0, 0.2) + \frac{1}{2} \mathrm{Unif}[0.7, 0.9) \label{DG5} \tag{DG5} \\
	Z_1,\ldots,Z_m  &\iid  \frac{1}{2} \mathrm{Beta}(2, 15) + \frac{1}{2} \mathrm{Beta}(15, 2) \label{DG6} \tag{DG6}.
\end{align}

We consider a GFPT1 prior where $\alpha_0 = 0.1$ (as suggested at the end of Section~\ref{sec:s1})
and a GFPT2 prior where $\alpha_{x_{1:n}} = \alpha_0 n^2$  with
$\alpha_0 \in \{0.1, 2.0\}$ (in Section~\ref{sec:s1}, these values of $\alpha_0$ provided the best GFPT1 posterior estimates). As far as the prior for $R$ is concerned (cf.\ Defintion~\ref{def:6}), we assume $\beta_{x_{1:n}} = \beta_0$ for all  $n\in\mathbb N$ and all $x_{1:n} \in \{0, 1\}^n$, so that the prior mean for the NBP $I$ is  given by the standard dyadic partitions employed by the GFPT1 model.
We let $\beta_0 \in \{0.5, 2.0, 5.0\}$, which allowed us to understand the sensitivity of GFPT2 priors: When $\beta_0 = 0.5$, the $\mathrm{BS}(\beta)$ prior for $R$ assigns significant mass to values of $R_{x_{1:n}}$ near to $0$ or $1$, corresponding to almost empty intervals. Instead when $\beta_0\in\{2.0, 5.0\}$, the $\mathrm{BS}(\beta)$ prior for $R$ causes a random NBP which is more and more concentrated around the standard dyadic partitions. 
We also tried increasing the value of of $\beta_0$ when increasing the level $n$ of the partitions, but this led to worse mixing of the MCMC algorithm for posterior simulations (see Remark \ref{rem:mcmc-gfpt2-collapsed}; we used  10,000 iterations and discarding the first 1,000 as burn-in). In our examples, running this MCMC algorithm on a standard laptop takes only around five seconds.
The prior specification for the OPT and APT models follows the defaults of the \texttt{PTT} package.

\begin{figure}[]
	\centering
	\includegraphics[width=\linewidth]{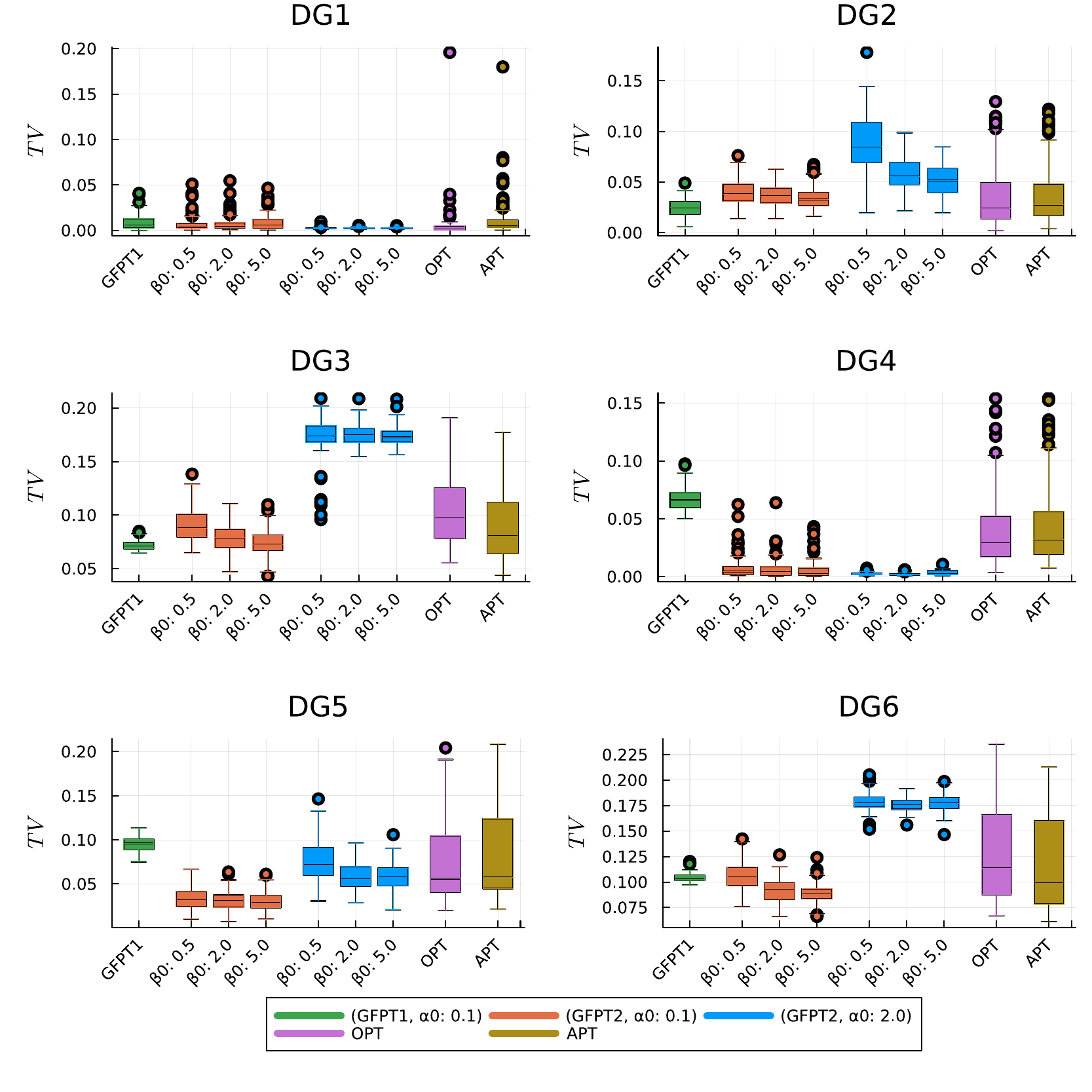}
	\caption{Total variation distance (TV) between true and estimated densities. Each boxplot summarizes the results from 100 independent replicates of the indicated data generating model.}
	\label{fig:s2_tv}
\end{figure}

\begin{figure}[]
	\centering
	\includegraphics[width=\linewidth]{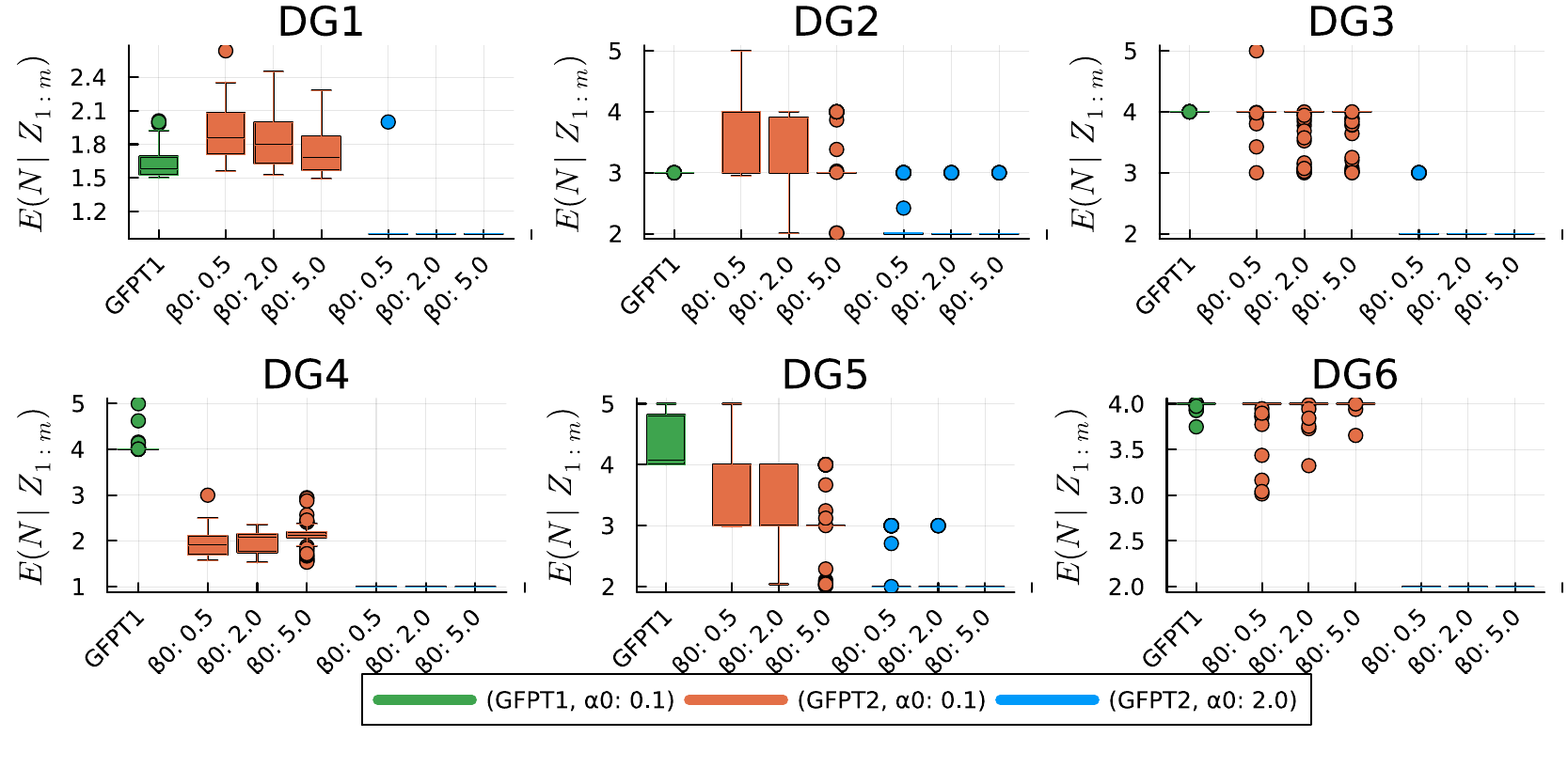}
	\caption{Posterior expectation of $N$ for the simulated datasets in Section \ref{sec:simu2}.  Each boxplot summarizes the results from 100 independent replicates of the indicated data generating model.}
	\label{fig:s2_n}
\end{figure}

For each data generating process, we simulated 1,000 observations and repeated each analysis on 100 independently simulated datasets. Figure \ref{fig:s2_tv} shows for the different data generating processes and parameter values, the total variation distance between true and estimated densities obtained with the different models. Moreover, Figure \ref{fig:s2_n} shows the posterior mean of $N$ under the GFPT1 and GPFT2 models.
We observer the following. First, the parameter $\beta_0$ does not seem to play a significant role in posterior inference. Second, when using GFPT2 priors, as in the case of GFPT1 priors, higher values of $\alpha_0$ correspond to smaller values of $N$ a posteriori. In some scenarios (e.g., DG1 and DG4) this yields also a more accurate estimation of the density because the density needs only a coarse partition to be well approximated. In other scenarios (e.g., DG3 and DG6), this results in a much poorer density estimate. In both such cases, the data generating density entails that the number of sufficient digits is infinite, and a finer partition results in better density estimates.
Focusing specifically on DG4 and DG5, observe that those data generating densities are piecewise continuous on intervals that do not agree with the standard diadic partitions employed for the GFPT1 prior. As a result, under the GFPT1 prior, a posterior $N$ tends to be large in those scenarios, while especially when $\alpha_0 = 2$ the GFPT2 model correctly captures the partitions needed to approximate those densities with a small number of intervals, resulting in smaller $N$ a posteriori.
In particular, for DG4, using the GFPT2 prior with $\alpha_0 = 2$ achieves almost a perfect density estimate.
Finally, the OPT and APT models yield comparable performance to the GFPT2 model in all settings except when data are generated from DG4 and DG5, in which case the GFPT2 models achieve superior density estimates.

\subsection{Fitted models for real datasets spanning multiple orders of magnitude}\label{sec:comparison_real_data}

We consider here five real-world datasets: 
\begin{enumerate}
\item[(i)] \emph{Twitter} records the number of friends of 40,000 users, 
\item[(ii)] \emph{eurodist} contains the pairwise distances (in kilometers) between 21 cities in Europe, 
\item[(iii)] \emph{GDP} collects the gross domestic products (in dollars) of 196 countries, 
\item[(iv)] \emph{census} is the number of citizens in more than 19,000 US cities as of 2009, 
\item[(v)] \emph{income} records the personal income of more than 50,000 inhabitants in California as of 2023. 
\end{enumerate}
Data for the first four dataset can be found at \url{https://github.com/jasonlong/benfords-law}, while data for the \emph{income} dataset is publicly available from \url{https://www.census.gov/data.html}.
Table~\ref{tab:empirical-range} reports, for each dataset, the empirical range (maximum and minimum value of the data) and two diagnostics computed from the \emph{log-mantissa} $U = \log_{10}(Z) - \lfloor \log_{10}(Z) \rfloor$: The Wasserstein-1 distance $W_1$ between the empirical distribution of $U$ and the $\mathrm{Unif}[0,1)$ law, and the Kolmogorov-Smirnov distance $d_{\mathrm{KS}}$ to $\mathrm{Unif}[0,1)$. Under the Newcomb-Benford law, the distribution is \emph{scale invariant}, i.e., the \emph{log-mantissa} $U$ is $\mathrm{Unif}[0,1)$-distributed, hence, small values of $W_1$ and $d_{\mathrm{KS}}$ indicate that our MBPT model is well ssupported by the data.
The datasets in Table~\ref{tab:empirical-range} span four, three, eight, six, and seven 
 orders of magnitude, respectively, making them particularly well suited to test the adequacy of the base-2 and base-10 multiscale Benford P{\'o}lya tree priors. 
 In addition, the $W_1$ and $d_{\mathrm{KS}}$ statistics quantify how close each dataset is to the Benford benchmark: \emph{eurodist} has the largest deviations (0.18 and 0.07), signalling weak scale-invariance; by contrast, \emph{census} is extremely close to uniform (0.001 and 0.003), with \emph{income}, \emph{GDP}, and \emph{Twitter} also showing relatively small deviations.
 This pattern anticipates our empirical findings, whereby the MBPT model delivers the best fit where the data spans a large number of orders of magnitude and $W_1$ and $d_{\mathrm{KS}}$ are close to zero.

\begin{table}
	\centering
	\begin{tabular}{c | c c c c c}
		Dataset & Twitter & Eurodist & GDP & Census & Income     \\ \hline
		Range & $(101, 2.3 \cdot 10^5)$ & $(158, 4532)$ & $(5 \cdot 10^7, 10^{14})$ & $(1, 8 \cdot 10^6)$ & $(4, 1.3 \cdot 10^7)$ \\
		$W_1$            & 0.06 & 0.18 & 0.06 & 0.001 & 0.04   \\
		$d_{\mathrm{KS}}$ & 0.03 & 0.07 & 0.02 & 0.003 & 0.009 
	\end{tabular}
	\caption{Empirical range of the observations and $W_1$ and $d_{\mathrm{KS}}$ diagnostics in the different datasets analyzed in Section \ref{sec:comparison_real_data}.}
	\label{tab:empirical-range}
\end{table}

We compare posterior results when using five different priors, namely $\mathrm{MBPT}_2$ and $\mathrm{MBPT}_{10}$ priors, GFPT1 and GFPT2 priors fitted on max-scaled data (i.e., for each dataset, we divide all observations by the maximum value), and a Dirichlet process mixture (DPM) of Gaussian densities \citep{EscobarWest, gvdv, MacEachern, Neal}, which is the de-facto standard model for Bayesian density estimation. The 
 parameters of the GFPT1 and GFPT2 priors are selected following Section~\ref{s:elicitation} and the insights developed in Section \ref{sec:s1}. Specifically, $\alpha_{x_{1:n}} = 0.1 n^2$ for both priors, and $\beta_{x_{1:n}} =02$ for the CFPT2 prior. 
For the prior on $N$, we proceed as follows. 
Let $n_{\max}$ be the maximum number of base-10 digits recorded in a given dataset. Then, for the $\mathrm{MBPT}_{10}$ prior, assume that $N$ follows the truncation of a Poisson distribution to $\{0,1, \ldots, n_{\max}\}$, where the Poisson distribution has mean $n_{\max} / 2$.
For the $\mathrm{MBPT}_{2}$, GFPT1, and the GFPT2 priors, assume that $N$ follows the truncation of a Poisson distribution to $\{0,1, \ldots, \lceil \log_2(10^{n_{\max}}) \rceil \}$, where the Poisson distribution has mean $\log_2(10^{n_{\max}}) / 2$. This ensures that under all five priors of $N$, the number of possible values of the sufficient digits is approximately the same. %, cf.\ Section~\ref{s:gen-num}.
Further, for  the $\mathrm{MBPT}_2$ and $\mathrm{MBPT}_{10}$ priors, we fix $T$ as the span of the orders of magnitudes of the data, set $\omega_k \propto 1$ for all $k$, and 
 let $c_n$ from \eqref{e:alpha_benford} be given by $c_n = c_0 n^2$ for $n\in\mathbb N$, 
where $c_0 = 0.1$ in case of $\mathrm{MBPT}_{10}$ and $c_0 = 2.0$ in case of $\mathrm{MBPT}_{2}$. This choice of $c_0$ entails that the variance of the random PDF  in \eqref{e:finite_pt_N} is approximately equal under both MBPT priors.
Finally, for the DPM prior, we use the implementation and follow the default prior elicitation strategy in the \texttt{BayesMix} library \citep{bayesmix}.

The posterior inferences using the different priors are compared in term of the widely applicable information criterion (WAIC), see \cite{watanabe}, which is a consistent estimator of the out-of-sample error \citep{Vehtari}.
Since the absolute value of the WAIC is irrelevant, we fix the DPM as a baseline prior and report in Table \ref{tab:model_comparison} the relative improvement of the WAIC for any model over the DPM (larger values of the relative improvement  correspond to better model performance, while negative values indicate that the DPM achieves a better fit to the data). %Table \ref{tab:model_comparison} shows the results. 
For datasets \emph{eurodist}, \emph{GDP}, and \emph{census}, all priors seem to perform similarly. For the datasets \emph{Twitter} and \emph{income} datasets, the multiscale Benford P{\'o}lya tree priors clearly outperforms all other priors, where in particular, the $\mathrm{MBPT}_{10}$ prior shows a remarkable improvement of roughly 19\% and 25\% over the DPM prior, respectively.
These results confirm the intuition obtained by looking at the diagnostics based on the \emph{log-mantissa} and the span of the datasets discussed above.
Moreover, in all scenarios, the  WAIC when using the GFPT2 prior slightly improves as compared to the GFPT1 prior.
Although not reported here, we have also tried fitting a variant of the MBPT models where the parameters $\alpha$ are not specified in accordance to the Newcomb-Benford law but according to the standard choice of \eqref{e:alpha-first}. We noticed substantial deterioration in the WAICs across all datasets, and in particular for \emph{Twitter}, for which this variant performs worse than the DPM baseline, and \emph{income}, for which this variant performs on par to the baseline.

We now try to explain the reasons behind the large improvements of the $\mathrm{MBPT}_{10}$ model  \emph{Twitter}, \emph{income} and \emph{census} data. For \emph{Twitter} and \emph{income}, we believe that the main driver is human-driven bunching around thresholds: For instance, Twitter users are incentivized to reach a given number of followers either for monetization (e.g., power of ten followers) or for psychological reasons. Similarly, salary bands are usually set just above psychological or contractual cut-offs.
For the \emph{census} data instead, there might be population counts clustering near administrative thresholds for municipal classification.
Such repeating motifs are naturally expressed as stable digit patterns within each order of magnitude, which is what our MBPT models enforce.

\begin{table*}[ht]
	\centering
\begin{tabular}{llllll}
	\toprule
	 Dataset & GFPT1 & GFPT2  & $\mathrm{MBPT}_{2}$ & $\mathrm{MBPT}_{10}$ & DPM \\
	\midrule
	Twitter & -0.015 & 0.006 & 0.150 & 0.189 & 0.000 \\
	Eurodist & -0.000 & 0.007 & -0.004 & -0.059 & 0.000 \\
	GDP & -0.002 & 0.060 & 0.059 & 0.049 & 0.000 \\
	Census & -0.031 & 0.035 & 0.036 & 0.084 & 0.000 \\
	Income & 0.002 & 0.045 & 0.210 & 0.247 & 0.000 \\
	\bottomrule
\end{tabular}
\caption{Relative improvement of the WAIC over the Dirichlet process mixture prior.}
\label{tab:model_comparison}
\end{table*}

Finally, we comment on the number of sufficient digits needed in the different datasets according to the GFPT and MBPT models.
Table \ref{tab:model_comparison_n} reports the (estimated) a posteriori expectation of $N$ in the different models.
Note that $N$ is typically small a posteriori, despite $n_{\max}$ being equal to 
six, four, seventeen, eight, and six in the different datasets respectively. 
Recall that the prior for $N$ is truncated at $n_{\max}$ for the $\mathrm{MBPT}_{10}$ prior and to $\lceil\log_2(10^{n_{\max}}) \rceil$ 
for the other priors.
Hence, the sufficient number of digits are effectively always less than $n_{\max}$. 
Except for the Twitter and income datasets, we observe that $N$ tends to be smaller a posteriori under the $\mathrm{MBPT}_2$ prior than under the GFPT1 and GFPT2 priors. This entails that the max-scaling of the data disaligns the data from the nested binary partitions. Indeed, using the GFPT2 prior, the posterior adaptively learns the NBP, leading to a smaller $N$ than under the GFPT1 prior that works with the fixed NBP. 
For all the different datasets, a posteriori $N$ tends to be smaller under the $\mathrm{MBPT}_{10}$ prior, 
however, the possible values taken by $X_{1:N}$ is $10^{N}$ under the $\mathrm{MBPT}_{10}$ prior and (approximately) $2^{N}$ under the other models. 
We report the corresponding (estimated) posterior expectations in Table \ref{tab:model_comparison_n}.
When using the $\mathrm{MBPT}_{10}$ prior a larger number of possible digits is required a posteriori than if using another of the priors, but 
compared with the $\mathrm{MBPT}_{2}$ prior, the difference may not be extreme, since the scale is coarser for $10^N$ than for $2^N$.
On the other hand, there are several situations in which the $\mathrm{GPFT2}$ prior causes a posteriori a much lower number of possible digits than if using another of the priors, thanks to the added flexibility of learning the NBP.

\begin{table}[ht]
	\centering
\begin{tabular}{l|cccc|cccc}
	\toprule
	&  \multicolumn{4}{c}{$\mathrm E[N\mid Z_{1:M}]$} & \multicolumn{4}{|c}{$\mathrm E[q^{N}\mid Z_{1:M}]$} \\\midrule
	Dataset & GFPT1 & GFPT2  & $\mathrm{MBPT}_{2}$ & $\mathrm{MBPT}_{10}$ & GFPT1 & GFPT2  & $\mathrm{MBPT}_{2}$ & $\mathrm{MBPT}_{10}$\\
	\midrule
	Twitter & 9.0 & 7.0 & 12.0 & 4.0 & 512.80 & 128 & 4096 & 10000  \\
	Eurodist & 2.93 & 2.78 & 2.25 & 1.9 & 7.72 & 7.51 & 6.25 & 98.72 \\
	GDP & 5.0 & 5.0 & 2.92 & 1.9 & 32.32 & 32.00 & 8.86 & 99.93 \\
	Census &  10.0 & 6.0 & 4 & 1.9 & 1024 & 64 & 16 & 99.99 \\
	Income & 10.0 & 6.0 & 12.0 & 4.0& 1024 & 64 & 4096 & 10000 \\
	\bottomrule
\end{tabular}
\caption{ For the five datasets and four priors, posterior expectation of $N$ and 
%the number of possible digits, i.e., 
$q^{N}$ with $q=10$ for the  $\mathrm{MBPT}_{10}$ prior and $q=2$ for the remaining ones.}
\label{tab:model_comparison_n}
\end{table}

\section{Concluding remarks and open problems}\label{s:conclusion} 
{
This paper develops a family of {generalised finite P{\'o}lya tree} (GFPT) priors for a random PDF in which the {random sufficient digits} are modelled directly.  
Working with digits rather than with more classical quantities such as the mean or the variance is admittedly non-standard and, at first sight, conceptually harder.  
However, the vast literature developed for prior elicitation in standard P{\'o}lya trees can be clearly adapted in our setting as well.
As shown by our construction of the multiscale Benford P{\'o}lya tree prior, there are situations in which reasoning about the distribution of the digits is both more natural and leads to significant performance gain in the quality of posterior inference.

A natural extensions is to consider multivariate data.
 Two simple approaches would be to define  a product nested binary partition on the unit cube and model each dimension independently with a GFPT1 or GFPT2 prior, or to rely on coordinate values given by space-filling curves such as Morton's or Hilbert's curves.
However, we anticipate that even in these simple extensions, more efficient algorithms for posterior inference would need to be devised in order to scale to high-dimensional settings, perhaps borrowing ideas from \cite{AwayaMa}.
On the other hand, more complex models would be needed to infer the dependence across dimensions. We leave this problems as interesting ideas for future research. Moreover, we also plan to investigate a
slight modification of the GFPT2 model obtained by the mixture distribution
 \[
	R_{x_{1:n}} \mid R_{x_{1:n-1}} \sim \begin{cases}
		\pi \delta_{\{1\}} + (1 - \pi) B(\beta_{x_{1:n-1},0}, \beta_{x_{1:n-1}, 1}) & \text{if }  R_{x_{1:n-1}} < 1
		\\
		\delta_{\{1\}} & \text{if }  R_{x_{1:n-1}} = 1
	\end{cases}
\]
where $0<\pi<1$ is a parameter. In this way, we introduce an ``optional'' stopping of the partitioning in some regions of the domain, similarly to what is done in \cite{WongMa}.

Theorem \ref{thm:hell_consistency} establishes  under weak assumptions almost sure consistency in the Hellinger distance  of the random PDF a posteriori when using  GFPT1 priors. We leave it as an open problem to establish similar convergence results under GFPT2 priors.

Ongoing work by us extends the setting in this paper, % beyond base-$q$ expansions, 
utilizing the numeric system induced by a continued fraction representation. That is, $X_1,X_2,...\in \mathbb N$ are the random digits of $X$ such that $X=1/(X_1+(1/X_2+...))$. 
Thus, the digits take value in an unbounded space, which complicates inference, yet promises sharper inference because continued fractions give the best rational approximations of real numbers.
}

\section*{Acknowledgements} 
\noindent 
Mario Beraha gratefully acknowledges support from the Italian Ministry of Education, University and Research (MUR), ``Dipartimenti di Eccellenza'' grant 2023-2027.
Jesper M{\o}ller is supported by The Danish Council for Independent Research —
Natural Sciences, grant DFF – 10.46540/2032-00005B. 

\FloatBarrier

\clearpage

\appendix

\begin{center}
   \LARGE Supplemental material for:\\
    ``Sufficient digits and density estimation: A Bayesian nonparametric approach using generalized finite P{\'o}lya trees''
\end{center}

 \setcounter{equation}{0}
\setcounter{figure}{0}
\setcounter{table}{0}

\renewcommand{\theequation}{\Alph{section}\arabic{equation}}
\renewcommand{\thefigure}{\Alph{section}\arabic{figure}}
\renewcommand{\thetable}{\Alph{section}\arabic{table}}

\begin{proof}[Proof of Theorem \ref{t:1}]
	%The first part %in Theorem~\ref{t:1}
	The existence and coupling of $N$ with $X$ such that \eqref{e:res1} is satisfied follows from \cite{supp-moeller}. 
	
	Let the situation be as in the second part of the theorem. %Theorem~\ref{t:1}. 
	Clearly, $g$ is a PDF on $H$. Recall that a set $U\subseteq H$ is open (with respect to $H$) if it is the intersection of $H$ with an open subset of $\mathbb R$. The following facts %(a) and (b) 
	are well-known, see e.g.\ \cite{edwards}. 
	\begin{enumerate}
	\item[(a)] A 
	real function $h$ is LSC on $H$ if and only if for every $u>0$ the set $U=\{x\in H\,|\,h(x)>u\}$ is open (with respect to $H$), i.e.,  
	for any $x\in U$ there is a neighbourhood $U_x$ which means that $x\in U_x\subseteq U$ and $U_x$ is open (with respect to $H$).
	\item[(b)] A finite sum of LSC functions on $H$ is a LSC function on $H$.
	\end{enumerate}
	   Now, suppose $x\in H$ and $0<u<g(x)$. Since $g(x)=\sum_{n=0}^\infty g_n(x)$, there exists some $k_x\in\mathbb N_0$ so that $\sum_{n=0}^{k_x} g_n(x)>u$. By (b), $\sum_{n=0}^{k_x} g_n$ is a LSC function on $H$, and hence by (a), the set $U_x=\{y\in H\,|\,\sum_{n=0}^{k_x} g_n(y)>u\}$ is open (with respect to $H$). Since $x\in U_x\subseteq U$, it follows that $g$ is LSC on $H$. The remaining statements in Theorem~\ref{t:1} are obviously true. 
	%Thereby the proof of Theorem~\ref{t:1} is completed.
	\end{proof}

\section{Proofs for Section \ref{s:finite_pt}}

\begin{proof}[Proof of Proposition \ref{prop:mean}]
By the total law of expectation and since $N$ and $Y$ are independent,
\[
\mathrm E[P(I_{x_{1:d}})\mid Y] = \sum_{n=0}^\infty p_n \mathrm E[P(I_{x_{1:d}}\mid Y,N=n)]\]
where
\[\mathrm E[P(I_{x_{1:d}}\mid Y,N=n)]=
\begin{cases}
\prod_{j=1}^{d} Y_{x_{1:j}} & \text{if }n\ge d,\\
\frac{\ell_{x_{1:d}}}{\ell_{x_{1:n}}}\prod_{j=1}^{n}Y_{x_{1:j}}& \text{if }n< d.
\end{cases}\]
%tower rule of expectation
%\[
%	\mathrm E[P(B_{x_{1:d}})] = \sum_{n \ge 1} p_n \mathrm E[P(B_{x_{1:d}}\mid Y,N=n)].
%\]
%Focusing on $P(B_{x_{1:d}}\mid Y,N=n)$, if $n \ge d$, clearly
%\[
%P(B_{x_{1:d}}\mid Y,N=n)=\prod_{j=1}^{d} Y_{x_{1:j}}.
%\]
%On the other hand, if $n<d$, $P(B_{x_{1:n}}\mid Y,N=n)=\prod_{j=1}^{n} Y_{x_{1:j}}$, and the density is uniform within $B_{x_{1:n}}$, so
%\[
%P(B_{x_{1:d}}\mid Y,N=n)=\Big(\prod_{j=1}^{n}Y_{x_{1:j}}\Big)\frac{\ell_{x_{1:d}}}{\ell_{x_{1:n}}}.
%\]
Taking expectations with respect to $Y$ leads to the result.
\end{proof}

\section{Proofs for Section \ref{s:posterior}}

\begin{proof}[Proof of Theorem \ref{t:2}]
Recall the definition in \eqref{e:mun} of the measure $\mu_n$ and the corresponding notation $Y_0^{(n)},A_n,\mathcal F_n$. 
%	Since $Y_{x_{1:j-1},1}=1-Y_{x_{1:j-1},0}$ for $j\in\mathbb N$, for every $n\in\mathbb N$, we get that $Y^{(n)}$ can be identified by $Y^{(n)}_0=(Y_{x_{1:j-1},0}\mid x_{1:j-1}\in\{0,1\}^{j-1}, j=1,\ldots,n)$, and we let $Y^{(0)}_0=Y^{(0)}=\emptyset$. For every $n\in\mathbb N_0$, the random vector $Y^{(n)}_0$ has dimension $2^n-1$. 
%	For $j\in\mathbb N$ and $x_{1:j-1}\in\{0,1\}^{j-1}$, let $\nu_{x_{1:j-1}}$ be the Lebesgue measure on $[0,1)$ if $\alpha_{x_{1:j-1},0}>0$ and $\alpha_{x_{1:j-1},1}>0$, and $\nu_{x_{1:j-1}}$ be the Dirac measure concentrated at $k\in\{0,1\}$ if $\alpha_{x_{1:j-1},k}=0$ and $\alpha_{x_{1:j-1},1-k}>0$.
%	For $n\in\mathbb N$, let $A_n=[0,1)^{2^n-1}$ be equipped with the corresponding Borel $\sigma$-algebra $\mathcal F_n$ and let $\mu_n$ be the product measure on $\mathcal F_n$ given by 
%	\begin{equation}\label{e:mun}
%	\mu_n=\prod_{j=1}^n\prod_{{x_{1:j-1}}\in\{0,1\}^{j-1}}\nu_{x_{1:j-1}}.
%	\end{equation}  
%	Equip $A_0=[0,1)^0=\{\emptyset\}$ with the trivial $\sigma$-algebra $\mathcal F_0$ and
%	let $\mu_0=\mathrm{FBS}(\alpha^{(0)})$, cf.\ Definition~\ref{def:cpt2}. 
	Then, $Y^{(N)}_0$ has state space $A=\cup_{n=0}^\infty A_n$ which we equip with the smallest $\sigma$-algebra $\mathcal F$ which contains $\cup_{n=0}^\infty\mathcal F_n$. Furthermore, define a reference measure $\mu$ on $(A,\mathcal F)$ so that $\mu(F)=\mu_n(F)$ whenever $F\in\mathcal F_n$ and $n\in\mathbb N_0$.
	
	Let $n\in\mathbb N_0$, $y^{(n)}=0$ if $n=0$, and $y^{(n)}=(y_{x_{1:j}}\mid x_{1:j}\in\{0,1\}^j,j=1,\ldots,n)$ if $n>0$ where $0\le y_{x_{1:j-1},0}\le 1$ and $y_{x_{1:j-1},1}=1-y_{x_{1:j-1},0}$. Then we can identity $y^{(n)}$ by $y^{(n)}_0=(y_{x_{1:j-1},0}\mid x_{1:j-1}\in\{0,1\}^{j-1}, 1\le j\le n)$.  Definitions~\ref{def:beta_seq} and \ref{def:cpt2} give that the distribution of $Y^{(N)}$ (or more precisely $Y^{(N)}_0$) is absolutely continuous with respect to $\mu$, with density
	\[p(y^{(n)})=\begin{cases}
	p_0 & \text{if }n=0,\\
	p_n\prod_{j=1}^n\prod_{x_{1:j-1} \in \{0, 1\}^{j-1}}  \mathrm B({y_{x_{1:j-1}, 0}} \mid \alpha_{x_{1:j-1, 0}},  \alpha_{x_{1:j-1, 1}}) & \text{if }n>0.
	\end{cases}\]
	 By Remark~\ref{rem:bayes_model}, the distribution of
	 $(Y^{(N)},Z_{1:m})$ 
	 is absolutely continuous with respect to the product measure of $\mu$ and the Lebesgue measure on $(0,1)^m$, with 
	 density $p(z_{1:m}, y^{(n)})$
	  such that 
	  for $n>0$, 
	\begin{equation}\label{e:joint1}
		\begin{aligned}
			& p(y^{(n)},z_{1:m}) =  p(y^{(n)}) \frac{\prod_{j=1}^n\prod_{x_{1:j-1} \in \{0, 1\}^{j-1}} y_{x_{1:j-1},0}^{n_{x_{1:j-1},0}(z_{1:m})} y_{x_{1:j-1},1}^{n_{x_{1:j-1},1}(z_{1:m})}}{{\prod_{x_{1:n}\in \{0, 1\}^n} \ell_{x_{1:n}}^{n_{x_{1:n}}(z_{1:m})} }} \\ 
			& = p_n \frac{\prod_{j=1}^n  
			 \prod_{x_{1:j-1} \in \{0, 1\}^{j-1}} y_{x_{1:j-1},0}^{n_{x_{1:j-1},0}(z_{1:m})} y_{x_{1:j-1},1}^{n_{x_{1:j-1},1}(z_{1:m})} \mathrm B({y_{x_{1:j-1}, 0}} \mid \alpha_{x_{1:j-1},0},  \alpha_{x_{1:j-1},1})}{{\prod_{x_{1:n}\in \{0, 1\}^n} \ell_{x_{1:n}}^{n_{x_{1:n}}(z_{1:m})} }}
		\end{aligned} 
	\end{equation}
	whilst if $n=0$ then $p(y^{(n)},z_{1:m}) =p_0$.
	
	The conditional distribution of $N$ given $Z_{1:m}=z_{1:m}$ has a PMF $p(n \mid z_{1:m})$ with $p(0\mid z_{1:m})\propto p(0)$ and where for every $n\in\mathbb N$, \eqref{e:joint1} in this supplementary material gives
	\[p(n \mid z_{1:m})  \propto p(n,z_{1:m}) 
		= \int_{A_n} p(y^{(n)},z_{1:m})\, \mathrm d\mu_n(y^{(n)}_0)\]
	%\begin{equation*}%\label{e:post_n_1}
	%\begin{aligned}
	%	& p(n \mid z_{1:m})  \propto p(z_{1:m}, n) 
	%	= \int_{A_n} p(z_{1:m}, y^{(n)})\, \mathrm d\mu_n(y^{(n)}_0) \\
	%%	&= P(N = n) \prod_{j=1}^n \prod_{x_{1:j} \in \{0, 1\}^j} \int_{(0, 1)^{2^{n-1}}}  y_{x_{1:j-1, 0}}^{n_{x_{1:j-1, 0}}(z_{1:m})} y_{x_{1:j-1, 1}}^{n_{x_{1:j-1, 1}}(z_{1:m})} \mathrm B({y_{1:j-1, 0}} \mid \alpha_{x_{1:j-1, 0}},  \alpha_{x_{1:j-1, 1}}) \mathrm d y \\
	%	&= p_n \prod_{j=1}^n \prod_{x_{1:j} \in \{0, 1\}^j} \mathrm b(\alpha_{x_{1:j-1},0} + n_{x_{1:j-1},0}(z_{1:m}), \alpha_{x_{1:j-1},1} + n_{x_{1:j-1},1}(z_{1:m}) ).
	%\end{aligned}
	%\end{equation*}
	which reduces to \eqref{e:post_n_1}. It also follows from \eqref{e:joint1} in this supplementary material that the conditional distribution of
	$Y^{(N)}$ given $(N,Z_{1:m})=(n,z_{1:m})$ is absolutely continuous with respect to $\mu_n$, with 
	 a density $p(y\mid n,z_{1:m})$ for $y\in A_n$ so that if $n>0$,
	\[\begin{aligned} 
	&p(y^{(n)} \mid n, z_{1:m})\\
	&\propto\prod_{j=1}^n  
			 \prod_{x_{1:j-1} \in \{0, 1\}^{j-1}} y_{x_{1:j-1, 0}}^{n_{x_{1:j-1, 0}}(z_{1:m})} y_{x_{1:j-1, 1}}^{n_{x_{1:j-1, 1}}(z_{1:m})} \mathrm B({y_{1:j-1, 0}} \mid \alpha_{x_{1:j-1, 0}},  \alpha_{x_{1:j-1, 1}})
			 \end{aligned} \]
			 whilst if $n=0$ then $Y^{(N)}=0$. 		 
			 This together with Definition~\ref{def:5} and the fact that $X$ and $Y^{(>N)}$ conditioned on $N$ are independent gives the last statement of the theorem.
\end{proof}

\begin{proof}[Proof of Corollary \ref{cor:point_est}] 
	We have
	\begin{align*}
	&\E[f(x\mid Y^{(N)},I^{(N)}) \mid Z_{1:m} = z_{1:m}]  = \E[ \E[  f(x\mid Y^{(N)},I^{(N)}) \mid  N, Z_{1:m}] \mid  Z_{1:m}=z_{1:m} ]\\
	&=\sum_{n=0}^\infty p(n \mid z_{1:m}) \E[ f(x\mid Y^{(N)},I^{(N)}) \mid  N = n, Z_{1:m} = z_{1:m}]\\
	&=\sum_{n=0}^\infty p(n \mid z_{1:m}) 
	\E \left[\frac{\prod_{j=1}^n Y_{x_{1:j-1},0}^{1-x_j} Y_{x_{1:j-1},1}^{x_j}}{\ell_{x_{1:n}}}\, \bigg|\, N=n, Z_{1:m} = z_{1:m}\right] 
	\end{align*}
	where the first equality follows from the law of total expectation and because of the conditional independence in \eqref{e:post-next1} and \eqref{e:post-next2}, the second equality is based on the law of total probability, and
	the last equality follows from \eqref{e:finite_pt_N}. Thereby, using \eqref{e:post-next1} we obtain \eqref{e:PM}.
\end{proof}

\begin{proof}[Proof of Theorem \ref{t:3}]
In a similar way as the proof of the first statement of Theorem~\ref{t:2}, it is possible to analytically marginalize out $Y$ from the posterior of $(N, Y, R)$, whereby the two first statements of Theorem \ref{t:3} follow. The last statement follows immediately from Theorem~\ref{t:2}.
%	The two first statements follow immediately from Theorem~\ref{t:2}. The last statement is verified in a similar way as the proof of the last statement of Theorem~\ref{t:2} by using Definition~\ref{def:6}, the hierarchical model discussed in Remark~\ref{r:4}, and the fact that conditioned on $N$ we have that $X$, $Y^{(>N)}$, and $R^{(>N)}$ are independent. 
\end{proof}

\section{Proofs for Section \ref{sec:cons_n}}

	\begin{proof}[Proof of Theorem \ref{t:consistency}]
		Theorem \ref{t:consistency} is obviously true if $\mathcal N=\{n^*\}$, so assume $\mathcal N$ has cardinality at least two.
		%The posterior distribution of $N$ is given by \eqref{e:post_n_1}, where 
		
		To get rid of the associated normalizing constant in \eqref{e:post_n_1}, we will prove the following equivalent statement of \eqref{e:cons}: For any $\tilde n \in \mathcal N \setminus\{n^*\}$,
		\begin{equation}\label{e:to-verify}
			 \mathrm P^*\left(\lim_{m \rightarrow \infty}{p(n^* \mid Z_{1:m})}/{p(\tilde n \mid Z_{1:m})} =\infty\right)=1.
		\end{equation}
		To establish \eqref{e:to-verify} in this supplementary material we use the following facts, considering any 
		$n \in \mathcal N$ such that $n+1 \in \mathcal N$.
		By the law of large numbers, 
	\begin{equation}\label{e:bb}
		\mathrm P^*\left(\lim_{m\rightarrow\infty}N_{x_{1:n}}/m= \textstyle\int_{I_{x_{1:n}}} f^*(v) \, \mathrm d v \right)=1
	\end{equation}	
	where for any $j\in\mathbb N_0$ and $x_{1:j} \in \{0, 1\}^j$, we use $N_{x_{1:j}}= n_{x_{1:j}}(Z_{1:m})$ as a shorthand notation {\color{black}(this should not be confused with the notation used in Section~\ref{s:bayes_2}).}
		By \eqref{e:post_n_1}, 
		\[
			\frac{p(n \mid Z_{1:m})}{p(n + 1 \mid Z_{1:m})} = \frac{p_n}{p_{n+1}} \times \frac{\prod_{x_{1:(n+1)} \in \{0, 1\}^{n+1}} \ell_{x_{1:(n+1)}}^{N_{x_{1:(n+1)}}}}{\prod_{x_{1:n} \in \{0, 1\}^{n}} \ell_{x_{1:n}}^{N_{x_{1:n}}} \mathrm b(\alpha_{x_{1:n},0} + N_{x_{1:n},0}, \alpha_{x_{1:n},1} + N_{x_{1:n},1} )}. 
		\]
		 Since $N_{x_{1:n}} = N_{x_{1:n}, 0} + N_{x_{1:n}, 1}$ and $\ell_{x_{1:n}} = \ell_{x_{1:n}, 0} + \ell_{x_{1:n}, 1}$, we get
		\begin{align}
			\frac{p(n \mid Z_{1:m})}{p(n + 1 \mid Z_{1:m})} &= \frac{p_n}{p_{n+1}} \prod_{x_{1:n} \in \{0, 1\}^n} \frac{\ell_{x_{1:n}, 0}^{N_{x_{1:n}, 0}} \ell_{x_{1:n}, 1}^{N_{x_{1:n}, 1}}}{(\ell_{x_{1:n}, 0} + \ell_{x_{1:n}, 1})^{N_{x_{1:n}, 0} + N_{x_{1:n}, 1}} } \nonumber\\ 
			& \hspace{2cm} \times  \frac{1}{\mathrm b(\alpha_{x_{1:n},0} + N_{x_{1:n},0}, \alpha_{x_{1:n},1} + N_{x_{1:n},1} )} . \label{e:aa}
	%				\\
	%		& \asymp{\color{red}\frac{p_n}{p_{n+1}}} \prod_{x_{1:n} \in \{0, 1\}^n} \frac{\ell_{x_{1:n}, 0}^{N_{x_{1:n}, 0}} \ell_{x_{1:n}, 1}^{N_{x_{1:n}, 1}}}{(\ell_{x_{1:n}, 0} + \ell_{x_{1:n}, 1})^{N_{x_{1:n}, 0} + N_{x_{1:n}, 1}} } \\
	%		& \hspace{2cm} \times \frac{(N_{x_{1:n}, 0} + N_{x_{1:n}, 1})^{N_{x_{1:n}, 0} + N_{x_{1:n}, 1} - 1/2} }{{\color{red}\sqrt{2\pi}}\,N_{x_{1:n}, 0}^{N_{x_{1:n}, 0} - 1/2} N_{x_{1:n}, 1}^{N_{x_{1:n}, 1}{\color{red}-1/2}}} 
	\end{align}
	Define $$i_{x_{1:n}}=\begin{cases}
	1 & \text{if }\int_{I_{x_{1:n}}} f^*(v) \, \mathrm d v > 0,\\ 0 & \text{otherwise}.
	\end{cases}$$ 
	If $i_{x_{1:n}}=0$ then $\mathrm P^*(N_{x_{1:n}}=0)=1$. 
	For two sequences $a_1,a_2,\ldots$ and $b_1,b_2,\ldots$ of real numbers, write $a_m \asymp b_m$ if $\lim_{m\rightarrow \infty} a_m / b_m = 1$. 
	In the remainder of this proof it is implicit that any convergence result hold almost surely under the true distribution as $m\to\infty$.
	Thus,
	\begin{align}
			&\frac{p(n \mid Z_{1:m})}{p(n + 1 \mid Z_{1:m})} 	
			 \asymp{\color{black}\frac{p_n}{p_{n+1}}}\times\nonumber\\
			 &\prod_{\substack{x_{1:n} \in \{0, 1\}^n:\\ 
			 i_{x_{1:n}}=1}} \frac{\ell_{x_{1:n}, 0}^{N_{x_{1:n}, 0}} \ell_{x_{1:n}, 1}^{N_{x_{1:n}, 1}}}{(\ell_{x_{1:n}, 0} + \ell_{x_{1:n}, 1})^{N_{x_{1:n}, 0} + N_{x_{1:n}, 1}} } 
			 \frac{(N_{x_{1:n}, 0} + N_{x_{1:n}, 1})^{N_{x_{1:n}, 0} + N_{x_{1:n}, 1} - 1/2} }{{\color{black}\sqrt{2\pi}}\,N_{x_{1:n}, 0}^{N_{x_{1:n}, 0} - 1/2} N_{x_{1:n}, 1}^{N_{x_{1:n}, 1}{\color{black}-1/2}}} \label{e:cc}
		\end{align}	
	thanks to equations \eqref{e:bb} and \eqref{e:aa} in this supplementary material, the condition on the $\alpha_{x_{1:n}}$'s, and Stirling's approximation of the beta function: $$\mathrm b(x, y) =  \sqrt{2\pi}\frac{x^{x-1/2} y^{y - 1/2}}{(x + y)^{x + y - 1/2}}\left(1+O(1/x)+O(1/y)\right).$$ 
		
	Now, consider any integer  $n \in \mathcal N$ with $n > n^*$.
	By Theorem~\ref{t:1}, for any $y\in I_{x_{1:n}}$  we have that $f^*(y)$ is constant and $\int_{I_{x_{1:n}}} f^*(v) \, \mathrm d v = f^*(y)\ell_{x_{1:n}}$. Hence, for $k=0,1$ and $i_{x_{1:n}}=1$, it follows from \eqref{e:bb} in this supplementary material that 
	\begin{equation}\label{e:mmm}
	\mathrm P^*\left(\lim_{m\rightarrow\infty}N_{x_{1:n},k} / N_{x_{1:n}}=\ell_{x_{1:n}, k} / \ell_{x_{1:n}}\right)=1.
	\end{equation}
	By equations \eqref{e:cc} and \eqref{e:mmm} in this supplementary material,
		\begin{align}
			\frac{p(n \mid Z_{1:m})}{p(n + 1 \mid Z_{1:m})} & \asymp
			\frac{p_n}{p_{n+1}}
			 \prod_{\substack{x_{1:n} \in \{0, 1\}^n:\\ i_{x_{1:n}}=1}} \frac{\ell_{x_{1:n}, 0}^{N_{x_{1:n}, 0}} }{(\ell_{x_{1:n}, 0} + \ell_{x_{1:n}, 1})^{N_{x_{1:n}, 0} } } 
			\frac{\ell_{x_{1:n}, 1}^{N_{x_{1:n}, 1}}}{(\ell_{x_{1:n}, 0} + \ell_{x_{1:n}, 1})^{ N_{x_{1:n}, 1}} } \nonumber\\
			&  \times \left(\frac{N_{x_{1:n}, 0} + N_{x_{1:n}, 1}}{N_{x_{1:n}, 0}}\right)^{N_{x_{1:n}, 0} - 1/2}   \left(\frac{N_{x_{1:n}, 0} + N_{x_{1:n}, 1}}{N_{x_{1:n}, 1}}\right)^{N_{x_{1:n}, 1} - 1/2}   \nonumber\\
			& \times [(N_{x_{1:n}, 0} + N_{x_{1:n}, 1})^{1/2}/{\color{black}\sqrt{2\pi}}]\nonumber\\
			& \asymp {\color{black}\frac{p_n}{p_{n+1}}}
			\prod_{x_{1:n} \in \{0, 1\}^n: \,i_{x_{1:n}}=1} (N_{x_{1:n}} / {(2\pi)})^{1/2}.\label{e:shit}
		\end{align}
	By \eqref{e:assumption}, the term $p_{n} / p_{n+1}$ is strictly positive and bounded, so \eqref{e:shit} in this supplementary material gives
	\begin{equation}\label{e:shitshit}
		\frac{p(n \mid Z_{1:m})}{p(n + 1 \mid Z_{1:m})}\to \infty. %\quad\mbox{as }m\rightarrow\infty.
	\end{equation}
	Let $\tilde n \in \mathcal N$ with $\tilde n > n^*$. By \eqref{e:assumption}, $\{\tilde n, \tilde n-1,\ldots,n^*\} \subseteq \mathcal N$. Hence,  by  \eqref{e:shitshit} in this supplementary material, 
	\[
		\frac{p(n^* \mid Z_{1:m})}{p(\tilde n \mid Z_{1:m})} = \frac{p(n^* \mid Z_{1:m})}{p(n^*+1 \mid Z_{1:m})} \cdots \frac{p(\tilde n - 1 \mid Z_{1:m})}{p(\tilde n \mid Z_{1:m})}\to\infty,
	\]
	whereby \eqref{e:to-verify} in this supplementary material is verified.
	
	Consider instead any $n \in \mathcal N$ with $n < n^*$.
	Define $d_{x_{1:n}} = 0$ if $f^*$ is  constant over $I_{x_{1:n}}$ and $d_{x_{1:n}} = 1$ otherwise. 
		Then \eqref{e:cc} in this supplementary material writes
		\begin{align*}
			&\frac{p(n \mid Z_{1:m})}{p(n + 1  \mid Z_{1:m})} 	
			 \asymp{\frac{p_n}{p_{n+1}}}  \\
			&  \times \prod_{\substack{x_{1:n} \in \{0, 1\}^n: \\ d_{x_{1:n}}=0,\,i_{x_{1:n}} = 1}} \frac{\ell_{x_{1:n}, 0}^{N_{x_{1:n}, 0}} \ell_{x_{1:n}, 1}^{N_{x_{1:n}, 1}}}{\ell_{x_{1:n}}^{N_{x_{1:n}, 0} + N_{x_{1:n}, 1}} }  \frac{(N_{x_{1:n}, 0} + N_{x_{1:n}, 1})^{N_{x_{1:n}, 0} + N_{x_{1:n}, 1} - 1/2} }{{\sqrt{2\pi}}\,N_{x_{1:n}, 0}^{N_{x_{1:n}, 0} - 1/2} N_{x_{1:n}, 1}^{N_{x_{1:n}, 1}{-1/2}}}  \\
			& \times \prod_{\substack{x_{1:n} \in \{0, 1\}^n: \\ d_{x_{1:n}}=1,\,i_{x_{1:n}} = 1}} \frac{\ell_{x_{1:n}, 0}^{N_{x_{1:n}, 0}} \ell_{x_{1:n}, 1}^{N_{x_{1:n}, 1}}}{\ell_{x_{1:n}}^{N_{x_{1:n}, 0} + N_{x_{1:n}, 1}} }  \frac{(N_{x_{1:n}, 0} + N_{x_{1:n}, 1})^{N_{x_{1:n}, 0} + N_{x_{1:n}, 1} - 1/2} }{{\sqrt{2\pi}}\,N_{x_{1:n}, 0}^{N_{x_{1:n}, 0} - 1/2} N_{x_{1:n}, 1}^{N_{x_{1:n}, 1}{-1/2}}}.
		\end{align*}
	In a similar way as \eqref{e:shit} in this supplementary material was obtained, we see that each  term in the first product above is $O((N_{x_{1:n}} / {(2\pi)})^{1/2})$. 
		Consider any $x_{1:n}$ with $d_{x_{1:n}} = 1$ and $i_{x_{1:n}} = 1$. Defining $$\epsilon_{x_{1:n}}={\int_{I_{x_{1:n},0}} f^*(v) \, \mathrm d v}\bigg/{\int_{I_{x_{1:n}}} f^*(v) \, \mathrm d v},$$ then $\epsilon_{x_{1:n}} \in (0, 1)$.  By \eqref{e:bb} in this supplementary material,
		\[
			\frac{N_{x_{1:n}, 0}}{N_{x_{1:n}}} \rightarrow   \epsilon_{x_{1:n},0}, \quad \frac{N_{x_{1:n},1}}{N_{x_{1:n}}} \rightarrow \epsilon_{x_{1:n},1} = 1 - \epsilon_{x_{1:n},0}.
		\]
		  Therefore, 
		\begin{align*}
			& \frac{\ell_{x_{1:n}, 0}^{N_{x_{1:n}, 0}} \ell_{x_{1:n}, 1}^{N_{x_{1:n}, 1}}}{\ell_{x_{1:n}}^{N_{x_{1:n}, 0} + N_{x_{1:n}, 1}} }  \frac{(N_{x_{1:n}, 0} + N_{x_{1:n}, 1})^{N_{x_{1:n}, 0} + N_{x_{1:n}, 1} - 1/2} }{{\sqrt{2\pi}}\,N_{x_{1:n}, 0}^{N_{x_{1:n}, 0} - 1/2} N_{x_{1:n}, 1}^{N_{x_{1:n}, 1}{-1/2}}}\\
		& \qquad = \frac{\ell_{x_{1:n}, 0}^{N_{x_{1:n}, 0}} \ell_{x_{1:n}, 1}^{N_{x_{1:n}, 1}}}{\ell_{x_{1:n}}^{N_{x_{1:n}, 0} + N_{x_{1:n}, 1}} } 
		\frac{N_{x_{1:n}0}^{N_{x_{1:n},0}-1/2}}{N_{x_{1:n},0}^{N_{x_{1:n},0}-1/2}}	 \frac{N_{x_{1:n}}^{N_{x_{1:n},1}-1/2}}{N_{x_{1:n},1}^{N_{x_{1:n},0}-1/2}} N_{x_{1:n}} / \sqrt{2\pi}	
			 \\
			&  \qquad \asymp \left(\frac{\ell_{x_{1:n},0}}{\ell_{x_{1:n}} \epsilon_{x_{1:n},0}}\right)^{N_{x_{1:n}}\epsilon_{x_{1:n},0}} \left(\frac{\ell_{x_{1:n},1}}{\ell_{x_{1:n}} (1-\epsilon_{x_{1:n},0})}\right)^{N_{x_{1:n}}(1-\epsilon_{x_{1:n},0})} \\
			& \qquad \qquad \times \epsilon_{x_{1:n},0}^{1/2} (1 - \epsilon_{x_{1:n},0})^{1/2} N_{x_{1:n}}^{1/2} / \sqrt{2\pi}
		\end{align*} 
	which is $O(m^{1/2})$ if $\ell_{x_{1:n}, 0} / \ell_{x_{1:n}} = \epsilon_{x_{1:n},0}$ and converges exponentially fast to 0 otherwise.
		To prove \eqref{e:to-verify} in this supplementary material, since for  $\tilde n \in \mathcal N$ with $\tilde n < n^*$,
		\[
			\frac{p(n^* \mid Z_{1:m})}{p(\tilde n \mid Z_{1:m})} = \frac{p(n^* \mid Z_{1:m})}{p(n^*-1 \mid Z_{1:m})} \cdots \frac{p(\tilde n + 1 \mid Z_{1:m})}{p(\tilde n \mid Z_{1:m})},
		\]
		  it is sufficient that there exists one $x_{1:j} \in \{0, 1\}^j$ with $\tilde n \leq j \leq n^*$ and $\ell_{x_{1:j}, 0} / \ell_{x_{1:j}} \neq \epsilon_{x_{1:j},0}$.
		By contradiction, suppose this is false whenever $\tilde n\le j \leq n^*$, that is, $\ell_{x_{1:j}, 0} / \ell_{x_{1:j}} = \epsilon_{x_{1:j},0}$ whenever $x_{1:j}\in\{0,1\}^j$ and $\tilde n\le j \leq n^*$.
		{\color{black}Then, for any integer $n \geq \tilde n$ and any $x_{1:n}\in\{0,1\}^n$, 
			\begin{align*}
				\mathrm P^*(Z_i \in I_{x_{1:n}}& \mid Z_i \in I_{x_{1:\tilde n}})\\& = \mathrm P^*(Z_i \in I_{x_{1:n}} \mid Z_i \in I_{x_{1:(n-1)}}) \cdots \mathrm P^*(Z_i \in I_{x_{1:(\tilde n+1)}} \mid Z_i \in I_{x_{1:\tilde n}}) \\
				&= \epsilon_{x_{1:n}} \epsilon_{x_{1:(n-1)}} \cdots \epsilon_{x_{1:(\tilde n + 1)}} = \frac{\ell_{x_{1:n}}}{\ell_{x_{1:\tilde n}}}.
			\end{align*}
			Hence, by the $\pi$-$\lambda$ theorem,  under the true model $Z_i \mid Z_i \in I_{x_{1:\tilde n}}$ is uniformly distributed on $I_{x_{1:\tilde n}}$. This is a contradiction since it means that $\tilde n$ digits were sufficient under the true model but we have assumed that $\tilde n < n^*$.
		}
	
		% This means that $n^* = 0$ and $f^*$ is the density of the uniform distribution on $[0, 1)$ so that the condition $n < n^*$ does not make sense.
	
		% Let $n=0$, and set $\ell_0 = r \ell_{\emptyset}$, $n_0 \asymp \gamma n$ for $r, \gamma \in (0, 1)$. Then
		% \begin{align*}
		% 	\frac{p(n \mid Z_{1:m})}{p(n + 1 \mid Z_{1:m})} \asymp \exp \Big\{ & n\left[ \gamma \log(r \ell_{\emptyset}) + (1 - \gamma) \ell_{\emptyset} - \log(\ell_{\emptyset}) \right] \\
		% 	& + n \left[ \log(n) - \gamma \log(\gamma n) - (1 - \gamma) \log((1 - \gamma)n) \right] \\
		% 	& + \frac{1}{2}\left[ \log(\gamma n) + \log((1 - \gamma)n) - \log(n) \right]
		% 	\Big\}
		% \end{align*}
		% which can be seen to converge to 0 unless $r = \gamma$.
		% The case for general $n$ follows similarly by replacing $\ell_{\emptyset}$ with $\ell_{x_{1:n}}$, $r$ with $r_{x_{1:n}}$ and $\gamma$ with $\gamma_{x_{1:n}}$. 
		% That is, if there exists $x_{1:n} \in \{0, 1\}^n$ such that
		% \begin{equation}\label{e:condition}
		% 	\frac{N_{x_{1:n}, 0}}{N_{x_{1:n}}} \not\rightarrow \frac{\ell_{x_{1:n}, 0}}{\ell_{x_{1:n}}}
		% \end{equation}
		% then $\frac{p(n \mid Z_{1:m})}{p(n + 1 \mid Z_{1:m})} \rightarrow 0$. The proof concludes by noting that requiring \eqref{e:condition} is equivalent to saying that the numbef of sufficient digits is $n$ which contradicts the assymption $n < n^*$.
	\end{proof}

\begin{proof}[Proof of Theorem \ref{t:inf_n}]
	{\color{black} Arguing as in the proof of Theorem~\ref{t:consistency} (case $n < n^*$), the condition implies that for any $n \in \mathbb N_0$, almost surely under the true model, $p(n \mid Z_{1:m}) / p(n+1 \mid Z_{1:m})\rightarrow0$ as $m\rightarrow\infty$. Thus the} result follows by induction.
\end{proof}	

\begin{proof}[Proof of Theorem \ref{thm:cons2}]
 Define $\tilde N_{x_{1:j}}$ with respect to $\tilde Z_{1:m}$ in the same way as we defined $N_{x_{1:j}}$ with respect to $Z_{1:m}$ in the proof of Theorem~\ref{t:consistency}. 
	Consider any integer $n > n^*$. By \eqref{e:post_n_1},
	\begin{align*}
		\frac{p(n^* \mid\tilde Z_{1:m})}{p(n \mid\tilde Z_{1:m})} &= \frac{p_{n^*}}{p_{n}} \frac{\prod_{x_{1:n} \in \{0, 1\}^n} \ell_{x_{1:n}}^{\tilde N_{x_{1:n}}}}{\prod_{x_{1:n^*} \in \{0, 1\}^{n^*}} \ell_{x_{1:n^*}}^{\tilde N_{x_{1:n^*}}}} \\ 
		& \hspace{2cm} \times  \frac{1}{\prod_{j=n^*+1}^n \prod_{x_{1:j} \in \{0, 1\}^j} \mathrm b(\alpha_{x_{1:j},0} +\tilde N_{x_{1:j},0}, \alpha_{x_{1:j},1} +\tilde N_{x_{1:j},1} )}. 
	\end{align*}
	 By \eqref{e:caseb}, for $j=\bar n,\bar n+1,\ldots$ and every $x_{1:j}\in\{0,1\}^j$, we have almost surely under the true distribution that
	\[
	\tilde N_{x_{1:j}} = \begin{cases}
			\tilde N_{x_{1:\bar n}} & \text{if } x_{1:j} = (x_{1:\bar n}, 0, \ldots, 0), \\
			0 & \text{otherwise}.
		\end{cases}
	\]
	Further, define $A_j = \{x_{1:j}\in \{0,1\}^j: x_{1:j} = (x_{1:\bar n}, 0, \ldots, 0) \text{ for some } x_{1:\bar n} \}$. Then, 
	\begin{align*}
		\frac{p(n^* \mid\tilde Z_{1:m})}{p(n \mid\tilde Z_{1:m})} &= \frac{p_{n^*}}{p_{n}} \prod_{x_{1:\bar n} \in \{0, 1\}^{\bar n}} \left(\frac{\ell_{x_{1:\bar n}, 0, \ldots, 0}}{\ell_{x_{1:n^*}}}\right)^{\tilde N_{x_{1:\bar n}}} \\
		& \hspace{2cm} \times  \frac{1}{\prod_{j=\bar n+1}^n \prod_{x_{1:j} \in A_j} \mathrm b(\alpha_{x_{1:j},0} +\tilde N_{x_{1:j},0}, \alpha_{x_{1:j},1})} \\
		& \hspace{2cm} \times  \frac{1}{\prod_{j=\bar n+1}^n \prod_{x_{1:j} \in \{0, 1\}^{j} \setminus A_j } \mathrm b(\alpha_{x_{1:j}, 0}, \alpha_{x_{1:j}, 1})}.
	\end{align*}
	Stirling's approximation of the beta function gives $\mathrm b(x, y) = \Gamma(y) x^{-y}(1+O(1/x))$ for fixed $y>0$. Therefore,
	\begin{align*}
		\frac{p(n^* \mid\tilde Z_{1:m})}{p(n \mid\tilde Z_{1:m})} = \left[ \prod_{x_{1:\bar n} \in \{0, 1\}^{\bar n}} \left(\frac{\ell_{x_{1:\bar n}, 0, \ldots, 0}}{\ell_{x_{1:\bar n}}}\right)^{\tilde N_{x_{1:\bar n}}} (N_{x_{1:\bar n}, 0})^{\alpha_{x_{1:j}, 1}} \right] O(1)
	\end{align*}
	where the term $[\cdots]$ goes to $0$ as $m\rightarrow\infty$, since for some $x_{1:\bar n}$, $\tilde N_{x_{1:\bar n}} \rightarrow \infty$ as $m \rightarrow \infty$.
\end{proof}	

\section{Proofs for Section \ref{s:consistency_f}}\label{app:cons_f}

In the following, for two PDFs $f_1$ and $f_2$ on $[0,1)$, let $$\mathrm{KL}(f_1, f_2) = \int_0^1 f_1(x) \log \frac{f_1(x)}{f_2(x)} \mathrm d x$$ denote the Kullback-Leibler divergence of $f_1$ from $f_2$. 

\begin{definition}[Kullback-Leibler support]
Let $\Pi$ be a distribution on the set of PDFs on $[0, 1)$. We say that a PDF $f$ belongs to the Kullback-Leibler support of $\Pi$ if, for any $\varepsilon > 0$,
\[
	\Pi \left(\{g: \mathrm{KL}(f, g) < \varepsilon \}\right) > 0.
\]
\end{definition}

\begin{proof}[Proof of Theorem \ref{thm:weak_consistency}]
	To prove the statement it is sufficient to show that $f^*$ belongs to the Kullback-Leibler support of the prior. This follows from Theorem~6.16 and Example~6.20 in \cite{supp-gvdv}.

	Consider the case of the GFPT1  prior $\Pi_1$, cf.\ \eqref{e:prior11}. Clearly,
	\begin{align*}
		\Pi_1\left(\{f: \mathrm{KL}(f^*, f) < \varepsilon \}\right) &= \sum_{n=0}^\infty \Pi_1\left(\{f: \mathrm{KL}(f^*, f) < \varepsilon \} \mid N = n\right) p_n \\
		& \geq \Pi_1\left(\{f: \mathrm{KL}(f^*, f) < \varepsilon \} \mid N = \bar n \right) p_{\bar n}
	\end{align*}
	for any $\bar n\in\mathbb N_0$. Following closely  the proof of Theorem 7.1 in \cite{supp-gvdv}, 
	we will show that $$\Pi_1\left(\{f: \mathrm{KL}(f^*, f) < \varepsilon \} \mid N = \bar n \right)>0$$ for sufficiently large $\bar n$. Setting 
	 $y_{x_{1:j-1}, k} = \mathrm P^*(I_{x_{1:j-1}, k}) / \mathrm P^*(I_{x_{1:j-1}})$, define the discretization of $f^*$ at level $n$ of the NBP $I$ as 
	\begin{equation}\label{eq:discretized_dens}
		f^*_n(x) = \prod_{j=1}^n y_{x_{1:j-1}, 0}^{1-x_j} y_{x_{1:j-1}, 1}^{x_j} {\ell_{x_{1:j-1}}}/{\ell_{x_{1:j}}}  \quad \text{if } x=.x_1x_2\ldots\in[0,1).
	\end{equation}
	Hence, for almost any $x\in[0,1)$,
	$$f^*(x) = \lim_{n\rightarrow \infty} f^*_n(x) = \prod_{j=1}^\infty  y_{x_{1:j-1}, 0}^{1-x_j} y_{x_{1:j-1}, 1}^{x_j}{\ell_{x_{1:j-1}}}/{\ell_{x_{1:j}}}.$$ 
	In a similar fashion, rewrite the random PDF in \eqref{e:finite_pt_N} as 
	\[
		f(x \mid Y^{(N)}) = \prod_{j=1}^N Y_{x_{1:j-1}, 0}^{1-x_j} Y_{x_{1:j-1}, 1}^{x_j} {\ell_{x_{1:j-1}}}/{\ell_{x_{1:j}}} \quad \text{if } x=.x_1x_2\ldots\in[0,1).
	\]
	Then, when $N=\bar n$, 
	\begin{align*}
		& \mathrm{KL}(f^*, f(\cdot \mid Y^{(\bar n)})) = \int_0^1 f^*(x) \log \frac{f^*(x)}{f(x \mid Y^{(\bar n)})}\, \mathrm d x \\ 
		&= \int_0^1 f^*(x)\left(\sum_{j=1}^{\bar n} \log \frac{y_{x_{1:j-1}, 0}^{1 -x_j} y_{x_{1:j-1}, 1}^{x_j}}{Y_{x_{1:j-1}, 0}^{1 -x_j} Y_{x_{1:j-1}, 1}^{x_j}}  + \sum_{j=\bar n+1}^{\infty} \log \left(y_{x_{1:j-1}, 0}^{1-x_j} y_{x_{1:j-1}, 1}^{x_j} \frac{\ell_{x_{1:j-1}}}{\ell_{x_{1:j}}} \right) \right) \mathrm d x
	\end{align*}
	Again from Theorem 7.1 in \cite{supp-gvdv} we see that for any $\delta > 0$,
	\[
		\Pi_1 \left( \int_0^1 f^*(x)\left(\sum_{j=1}^{\bar n} \log \frac{y_{x_{1:j-1}, 0}^{1 -x_j} y_{x_{1:j-1}, 1}^{x_j}}{Y_{x_{1:j-1}, 0}^{1 -x_j} Y_{x_{1:j-1}, 1}^{x_j}} \right) \mathrm d x < \delta \,  \bigg|\, N = \bar n \right) > 0.
	\] 
	On the other hand, if $\bar n$ is large enough the term 
	\[
		\int_0^1 f^*(x) \left(\sum_{j=\bar n+1}^{\infty} \log \left(y_{x_{1:j-1}, 0}^{1-x_j} y_{x_{1:j-1}, 1}^{x_j} \frac{\ell_{x_{1:j-1}}}{\ell_{x_{1:j}}} \right) \right) \mathrm d x 
	\]
	is negligible by the assumption $\int f^*(x) \log f^*(x)\,\mathrm dx < \infty$ and Lemma B.10 in \cite{supp-gvdv}. Hence, $f^*$ belongs to the Kullback-Leibler support of $\Pi_1$, whereby the proof is complete in case of the GFPT1 prior.

	The proof under the GFPT2 prior follows by observing that conditioned on $R$, the GFPT2 prior coincides with a GFPT1 prior, and then applying Proposition~6.28 in \cite{supp-gvdv}.
\end{proof}

For the proof of Theorem \ref{thm:hell_consistency}, we need some additional notation. For a (semi)metric space $(S, d)$ and a totally bounded subset $C \subset S$, define the $\varepsilon$-covering number of $C$ as the smallest integer of balls of radius $\varepsilon$ (with respect to $d$) needed to cover $C$. Denote this $\varepsilon$-covering number by $N(\varepsilon, C, d)$. 
The proof of Theorem \ref{thm:hell_consistency} is based on the following general result \citep[Theorem 6.23 in][]{supp-gvdv} where $\mathcal P$ denotes the set of absolutely continuous probability distributions on $[0, 1)$ and this set is equipped {\color{black} with the Borel $\sigma$-algebra generated by the Hellinger distance.}

\begin{theorem}\label{thm:general_consistency}
Let $\Pi\in \mathcal P$.	Consider the data model $Z_1, \ldots, Z_m \mid g \iid g$ with prior PDF $g \sim \Pi$, and denote by $\Pi(\cdot \mid Z_{1:m})$ the associated posterior distribution. 
	Assume that the data were generated i.i.d.\ from a probability measure $\mathrm P^*$ with a PDF $f^*$ which belongs to the Kullback-Leibler support of $\Pi$, and 
	  for every $\delta > 0$, there exist partitions $\mathcal P = \mathcal P_{m, 1} \cup \mathcal P_{m, 2}$ and a constant $C>0$ such that for all sufficiently large values of $m$, 
	\begin{itemize}
		\item[$(i)$] $\log N(\delta, \mathcal P_{m, 1}, d_{\mathrm H}) \leq m \delta^2$,
		\item[$(ii)$]   $\Pi(\mathcal P_{m, 2}) \leq \mathrm e^{- C m}$. 
	\end{itemize}
 Then, for any $\varepsilon > 0$,
	\[
		\mathrm P^*\left(\lim_{m \rightarrow \infty} \Pi(\{f \colon d_{\mathrm H}(f, f^*) < \varepsilon \} \mid Z_{1:m}) = 1  \right) = 1.
	\]
\end{theorem}

\begin{proof}[Proof of Theorem \ref{thm:hell_consistency}]
	We apply Theorem \ref{thm:general_consistency}. Since $f^*$ belongs to the Kullback-Leibler support of the prior, cf.\ the proof of Theorem \ref{thm:weak_consistency}, we just need to verify items $(i)$ and $(ii)$ in Theorem \ref{thm:general_consistency}. To do so, we follow closely the proof of Theorem 7.16 in \cite{supp-gvdv}:
	For a given PDF $f$ on $[0,1)$, denote by $f_k$ its discretization at level $k$ of the NBP $I$, cf.\ \eqref{eq:discretized_dens}. Consider partitions 
	$$\mathcal P_{m, 1} = \{f \colon \|\log (f / f_{k_m})\|_\infty \leq \varepsilon^2 / 8\},\quad\mathcal P_{m, 2} = \mathcal P \setminus \mathcal P_{m, 1},$$  such that $A \log m \leq k_m \leq B \log m$ where $0 < A < B$ are constants.
	Then the proof of item $(i)$ above follows the same lines as the proof of Theorem~7.16 in \cite{supp-gvdv}. Concerning item $(ii)$,
	{\color{black} write 
	$$\Pi_1 \left(f(\cdot\mid Y^{(N)})\in\mathcal P_{m, 2}\right) = \sum_{n = 0}^\infty p_n \Pi_1 \left(f(\cdot\mid Y^{(N)})\in\mathcal P_{m, 2} \,\big|\, N = n\right)$$
	and observe that $\Pi_1 (\cdot \mid N = n)$ gives positive probability only to those densities that are piecewise constant after level $n$ of the NBP $I$.
%	, so $\Pi_1 (\mathcal P_{m, 2} \mid N = n)=0$ for $n \le k_m$. 
	Thus,}
	\begin{equation*}
		\Pi_1 \left(f(\cdot\mid Y^{(N)})\in\mathcal P_{m, 2}\right) = \sum_{n = k_m + 1}^{\infty} p_n 
		\Pi_1 \left(f(\cdot\mid Y^{(N)})\in\mathcal P_{m, 2} \,\big|\, N = n\right)
		 \leq D e^{- c k_m}
	\end{equation*}
	for some constant $D$ depending only on $C$ and $c$, cf.\ Theorem~\ref{thm:hell_consistency}. Hence, item~(ii) above follows by taking $\delta = \delta(m) = m^{1/2} \log m$.
\end{proof}

\section{Simulations for data on $\mathbb R$ or $\mathbb R_+$}\label{app:simulations}

As discussed at the end of Section 1.1, our methodology can be applied to data supported on spaces larger than $[0,1)$ by means of suitable transformations. In this appendix, we investigate the impact of such transformations on posterior inference.

\subsection{Data supported on $\mathbb R$}\label{app:simu_r}

We generate $m = 1000$ observations from a continuous distribution on $\mathbb R$, specifically
\[
	Z_1, \ldots, Z_m \iid \frac{1}{3} t_3(-4, 1) + \frac{2}{3} \mathcal N(3, 1),
\]
where $t_a(\mu, s)$ denotes the Student $t$ distribution with $a$ degrees of freedom, location $\mu$, and scale $s$.

Data are mapped to $[0,1)$ via a bijective transformation $g$. In particular, we consider the logistic function $g(x) = 1 / (1 + e^{-x})$ and the probit function $g(x) = \Phi(x)$, where $\Phi$ is the CDF of the standard normal distribution. We fit a standard Pólya tree and a GFPT1 model (with the default prior specification as in Sections \ref{sec:s1} and \ref{sec:simu2}, and $\alpha_0=0.1$) to the transformed data and obtain posterior samples of the random PDF on $[0,1)$ induced by the two models. These posterior samples are then mapped back to $\mathbb R$ via the usual change-of-variable formula, and the density estimate is obtained by taking the pointwise sample average of the transformed densities.

\begin{figure}[t]
	\centering
	\includegraphics[width=0.45\linewidth]{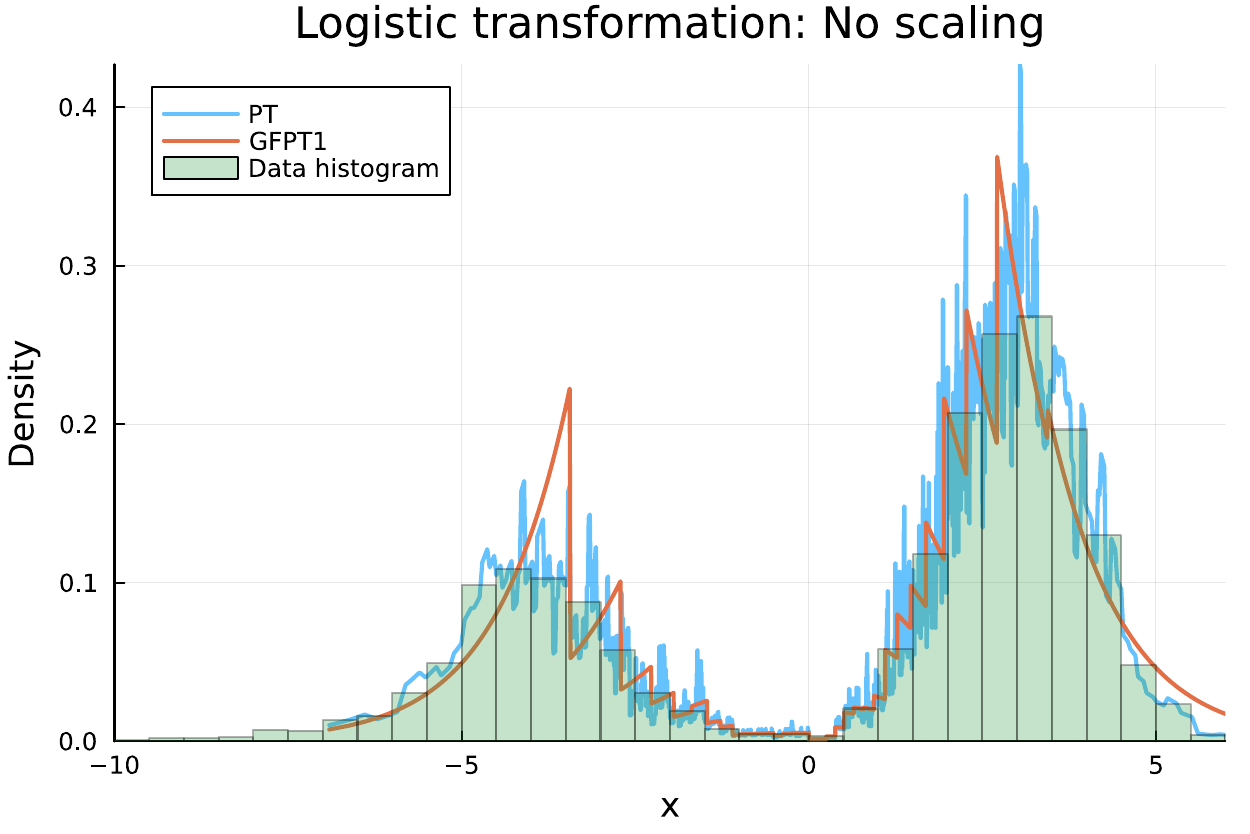}%
	\includegraphics[width=0.45\linewidth]{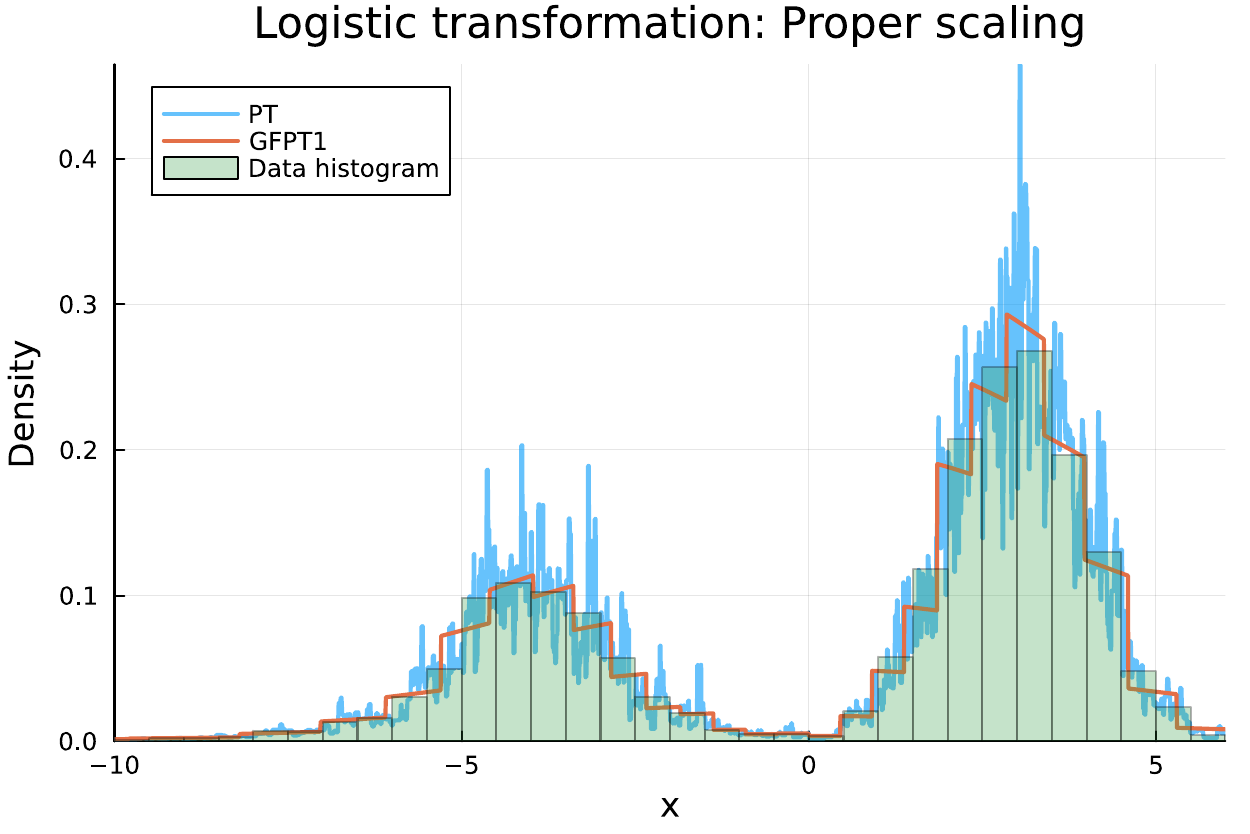}
	\includegraphics[width=0.45\linewidth]{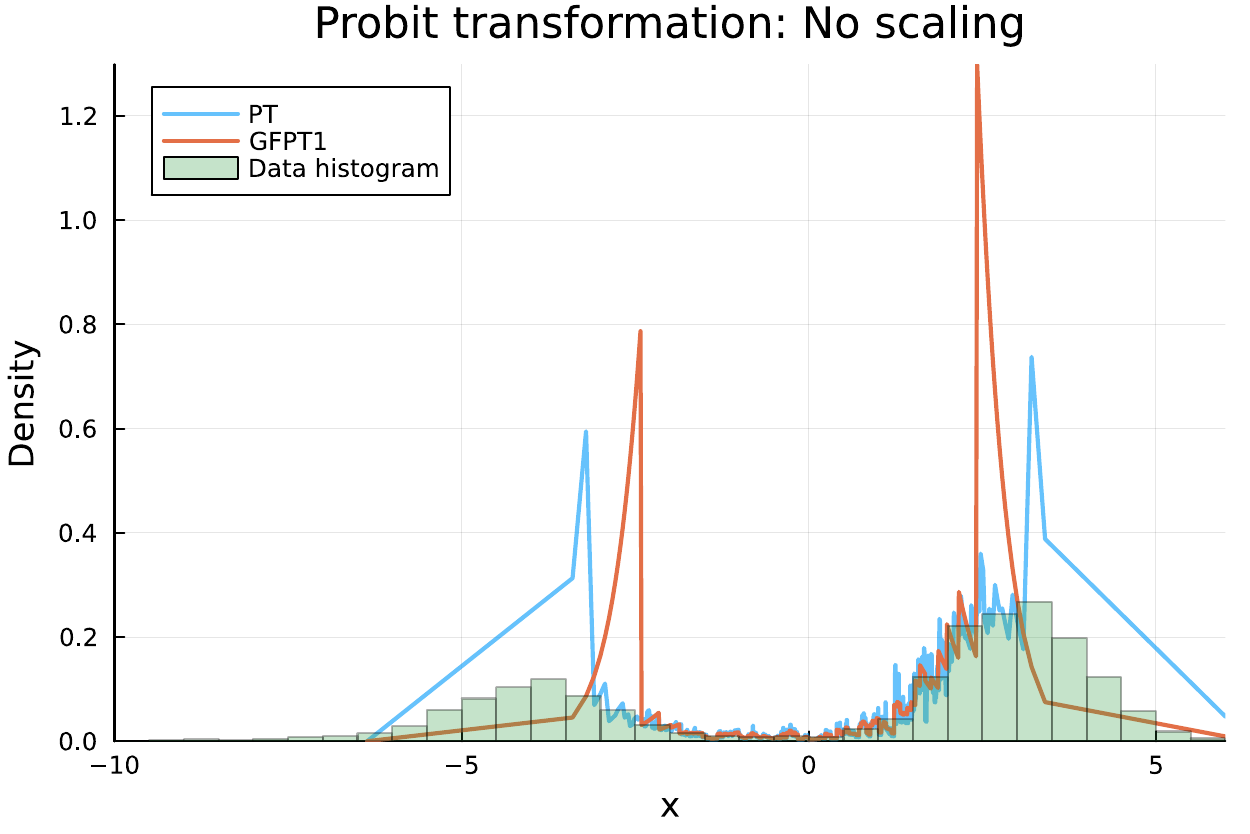}%
	\includegraphics[width=0.45\linewidth]{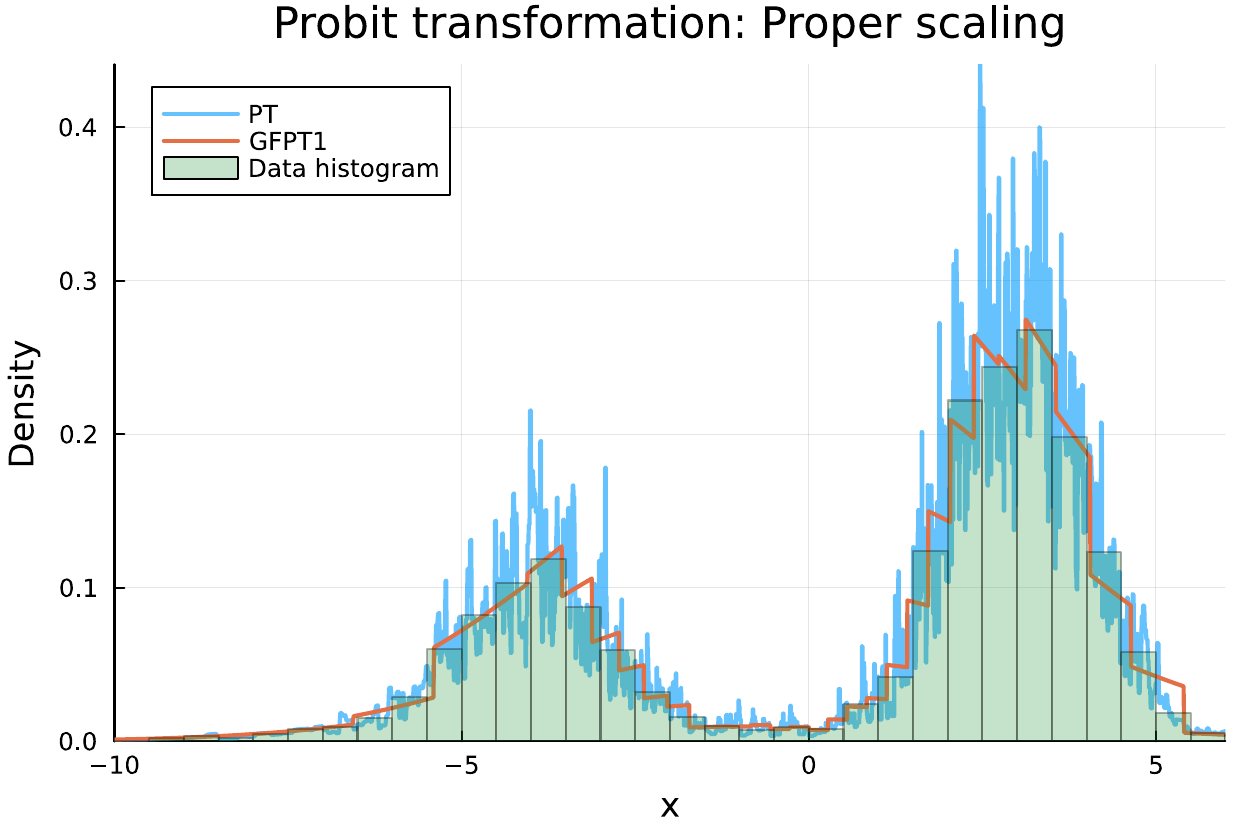}
	\caption{Posterior inference for the simulation in Appendix \ref{app:simu_r}. Top row: posterior density estimates under PT and GFPT1 models when data are transformed using the logistic function, without (left) and with (right) proper scaling. Bottom row: data transformed using the probit function.}
	\label{fig:simu_r}
\end{figure}

The two leftmost panels of Figure \ref{fig:simu_r} show the density estimates obtained using the logistic and probit transformations without scaling. The resulting fits are rather poor, especially under the probit transformation, in the left and right tails of the distribution. This behaviour is a consequence of the fact that the transformation shrinks tail observations towards $0$ and $1$, making it difficult for both the PT and GFPT1 models, when combined with the standard NBP $I$, to capture the tail behaviour.

One possibility would be to adopt an alternative NBP that allocates smaller intervals near the boundaries $0$ and $1$ and larger ones in the interior of the unit interval. A simpler solution is to account for the scale of the data by using a scaled version of the transformation, $g(x/s)$, where $s$ is the empirical standard deviation of the data; we refer to these as transformations with proper scaling. The two rightmost panels of Figure \ref{fig:simu_r} display the density estimates obtained when the properly scaled transformations are used. Posterior density estimates improve substantially, especially under the GFPT1 model. As in Section \ref{sec:s1}, the PT-based estimates remain extremely wiggly, whereas the GFPT1 prior yields smoother density estimates that better adapt to the underlying data-generating distribution.

\subsection{Data supported on $\mathbb R_+$.}\label{app:simu_r_plus}

\begin{figure}[t]
	\centering
	\includegraphics[width=0.6\linewidth]{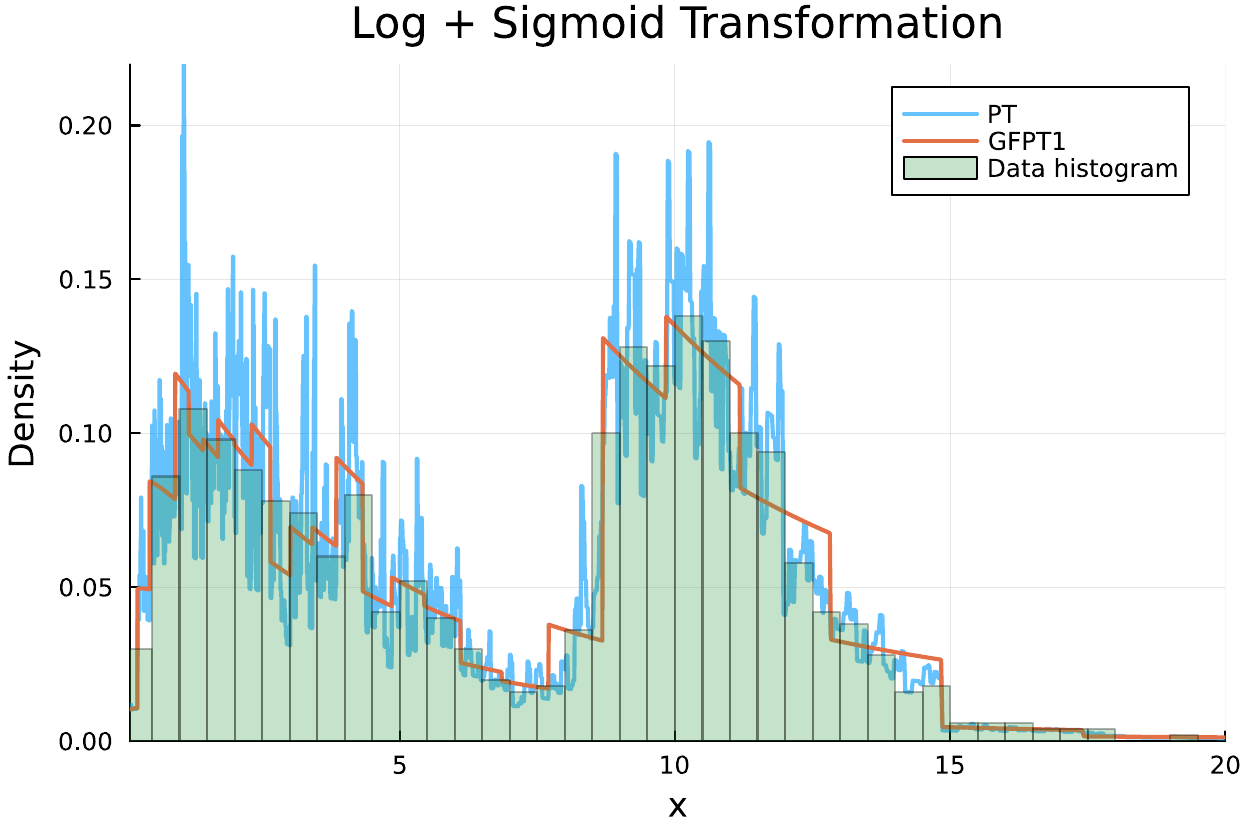}
	\caption{Density estimates for the simulation in Appendix \ref{app:simu_r_plus}.}
	\label{fig:simu_r_plus}
\end{figure}

We generate $m = 1000$ observations from
\[
	Z_1, \ldots, Z_m \iid \frac{1}{2} \mathrm{Ga}(2, 2) + \frac{1}{2} \mathrm{Ga}(3, 1; 8),
\]
where $\mathrm{Ga}(a, b; \mu)$ denotes the gamma distribution with shape $a$, rate $b$, and shift $\mu$ (which defaults to zero for the standard gamma distribution), that is, $X \sim \mathrm{Ga}(a, b; \mu)$ iff $X - \mu \sim \mathrm{Ga}(a, b)$.

We considered several bijective transformations $g$ from $\mathbb R_+$ to $[0,1)$, but most of them led to unsatisfactory results, with density estimates failing to capture the right tail of the true data-generating distribution. As in Section \ref{app:simu_r}, this is due to the fact that any bijection from $\mathbb R_+$ to $[0,1)$ necessarily maps large values close to $1$, and the PT and GFPT1 models on the standard NBP $I$ are not well suited to modelling such compressed tails.

In this setting, we found a two-stage approach to work particularly well. First, data are mapped from $\mathbb R_+$ to $\mathbb R$ using the logarithmic transformation. The transformed data are then centred so that their sample mean is zero, and the logistic transformation with proper scaling is applied. As before, once the models are fitted on $[0,1)$, posterior samples of the random PDF are mapped back to $\mathbb R_+$ via the change-of-variable formula.

Figure \ref{fig:simu_r_plus} displays the density estimates obtained with the GFPT1 and standard PT models; the same qualitative comments as in Appendix \ref{app:simu_r} apply.

\end{document}